\documentclass[a4paper,UKenglish,cleveref, autoref, thm-restate,colorlinks=true]{lipics-v2021}

\pdfoutput=1 
\hideLIPIcs  
\nolinenumbers 

\usepackage{listings}
\usepackage{ifthen}
\usepackage{xspace}

\newboolean{talk}
\setboolean{talk}{false}
\newboolean{paper}
\setboolean{paper}{false}

\newboolean{IEEEstyle}
\setboolean{IEEEstyle}{false}
\newboolean{lipicsstyle}
\setboolean{lipicsstyle}{false}
\newboolean{eptcsstyle}
\setboolean{eptcsstyle}{false}
\newboolean{entcsstyle}
\setboolean{entcsstyle}{false}
\newboolean{sigplanstyle}
\setboolean{sigplanstyle}{false}
\newboolean{easychairstyle}
\setboolean{easychairstyle}{false}
\newboolean{scrartclstyle}
\setboolean{scrartclstyle}{false}
\newboolean{lncsstyle}
\setboolean{lncsstyle}{false}

\newboolean{acmmart}
\setboolean{acmmart}{false}

\newboolean{needstheorems}
\setboolean{needstheorems}{false}
\newboolean{withimages}
\setboolean{withimages}{false}
\newboolean{withproofs}
\setboolean{withproofs}{true}
\newboolean{appendix}
\setboolean{appendix}{false}

\newboolean{french}
\setboolean{french}{false}

\setboolean{lipicsstyle}{true}
\setboolean{withproofs}{true}

\usepackage[utf8]{inputenc}
  \ifthenelse{\boolean{lipicsstyle}}{}{\usepackage[english]{babel}}

\usepackage[usenames]{xcolor}
\usepackage{amsmath}
\usepackage{amssymb}
\usepackage{amsfonts}
\usepackage{graphicx}
\usepackage{bussproofs}
\usepackage{cmll}
\usepackage{stmaryrd}
\usepackage{multirow}
\usepackage{hyperref}
\usepackage{cleveref}
\usepackage{bbm}
\usepackage{appendix}
\usepackage{mathtools}
\usepackage{complexity}

\usepackage{appendix}
\capturecounter{theorem}
\ifthenelse{\boolean{acmmart}}{}{
	\capturecounter{proposition}
	\capturecounter{lemma}
}
\usepackage{marginnote}

\newcommand{\NoteProof}[1]{
	\techrep{\ifthenelse{\boolean{appendix}}{
	\marginnote{Originally at p. \pageref{#1}}
	}{
	\marginnote{{Proof at p.\,{\pageref{app:#1}}}}
	}}
}
\newcommand{\applabel}[1]{$\phantomsection\label{app:#1}$}

\newcommand{\grameq}{::=}
\newcommand{\defeq}{:=}
\newcommand{\la}[1]{\lambda #1.}
\newcommand{\pair}[2]{(#1,#2)}

\newcommand{\red}[1]{{\color{red} {#1}}}

\newcommand{\cyan}[1]{{\color{cyan} {#1}}}

\newcommand{\symfont}[1]{\mathtt{#1}}

\newcommand{\msym}{\symfont{m}}
\newcommand{\esym}{\symfont{e}}
\newcommand{\wsym}{\symfont{w}}

\newcommand{\tm}{t}
\newcommand{\tmtwo}{s}
\newcommand{\tmthree}{u}

\newcommand{\var}{x}
\newcommand{\vartwo}{y}
\newcommand{\varthree}{z}

\newcommand{\val}{v}
\newcommand{\valtwo}{\val'}
\newcommand{\valthree}{\val''}

\newcommand{\bang}{\mathsf{!}}

\newcommand{\rootRew}[1]{\mapsto_{#1}}
\newcommand{\Rew}[1]{\rightarrow_{#1}}

\newcommand{\tens}{\otimes}
\newcommand{\lolli}{\multimap}
\newcommand{\substr}{\varogreaterthan}
\newcommand{\whyn}{{\mathsf{?}}}

\newcommand{\cuta}[2]{\red[#1{\shortrightarrow}#2\red]}
\newcommand{\cutsub}[2]{\{#1{\shortrightarrow}#2\}}
\newcommand{\para}[3]{\cyan[#1{\parr}#2,#3\cyan]}
\newcommand{\suba}[3]{\cyan[#1{\substr}#2,#3\cyan]}
\newcommand{\dera}[2]{\cyan[#1{\whyn}#2\cyan]}

\newcommand{\ax}{\symfont{ax}}

\newcommand{\axmsym}{\ax\msym}
\newcommand{\axesym}{\ax\esym}
\newcommand{\axe}{\axesym}
\newcommand{\axmone}{\axmsym_{1}}
\newcommand{\axmtwo}{\axmsym_{2}}

\newcommand{\rtoaxmone}{\rootRew{\axmone}}
\newcommand{\rtoaxmtwo}{\rootRew{\axmtwo}}
\newcommand{\rtotens}{\rootRew{\tens}}
\newcommand{\rtololli}{\rootRew{\lolli}}

\newcommand{\toaxmone}{\Rew{\axmone}}
\newcommand{\toaxmtwo}{\Rew{\axmtwo}}

\newcommand{\tololli}{\Rew{\lolli}}

\newcommand{\axeone}{\axe_{1}}
\newcommand{\axetwo}{\axe_{2}}

\newcommand{\rtoaxeone}{\rootRew\axeone}
\newcommand{\toaxeone}{\Rew\axeone}
\newcommand{\rtoaxetwo}{\rootRew\axetwo}
\newcommand{\toaxetwo}{\Rew\axetwo}

\newcommand{\rtobang}{\rootRew{\bang}}
\newcommand{\tobang}{\Rew{\bang}}
\newcommand{\rtow}{\rootRew{\wsym}}
\newcommand{\tow}{\Rew{\wsym}}

\newcommand{\toms}{\Rew{\mssym}}

\newcommand{\escsym}{\symfont{ESC}}
\newcommand{\tom}{\Rew{\msym}}
\newcommand{\toe}{\Rew{\esym}}
\newcommand{\toesc}{\Rew{\escsym}}
\renewcommand{\toms}{\toesc}
\newcommand{\toscnw}{\Rew{\escsym_{\neg\wsym}}}
\newcommand{\tomsnw}{\toscnw}

\makeatletter
\newsavebox{\@brx}
\newcommand{\llangle}[1][]{\savebox{\@brx}{\(\m@th{#1\langle}\)}%
  \mathopen{\copy\@brx\kern-0.7\wd\@brx\usebox{\@brx}}}
\newcommand{\rrangle}[1][]{\savebox{\@brx}{\(\m@th{#1\rangle}\)}%
  \mathclose{\copy\@brx\kern-0.7\wd\@brx\usebox{\@brx}}}
\makeatother

\newcommand{\ctxholep}[1]{\langle #1\rangle}
\newcommand{\ctxhole}{\ctxholep{\cdot}}
\newcommand{\ctxholefp}[1]{\llangle #1\rrangle}

\newcommand{\ctx}{C}
\newcommand{\ctxtwo}{D}
\newcommand{\ctxthree}{E}

\newcommand{\ctxp}[1]{\ctx\ctxholep{#1}}
\newcommand{\ctxtwop}[1]{\ctxtwo\ctxholep{#1}}

\newcommand{\ctxfp}[1]{\ctx\ctxholefp{#1}}

\newcommand{\lctx}{L}
\newcommand{\lctxtwo}{\lctx'}

\newcommand{\lctxp}[1]{\lctx\ctxholep{#1}}

\newcommand{\lctxtwop}[1]{\lctxtwo\ctxholep{#1}}

\newcommand{\vctx}{V}

\newcommand{\evar}{e}
\newcommand{\evartwo}{f}
\newcommand{\evarthree}{g}

\newcommand{\mvar}{m}
\newcommand{\mvartwo}{n}
\newcommand{\mvarthree}{o}

\newcommand{\exval}{\val_{\esym}}
\newcommand{\exvaltwo}{\exval'}

\newcommand{\mval}{\val_{\msym}}

\newcommand{\fv}[1]{\symfont{fv}(#1)}
\newcommand{\mfv}[1]{\symfont{mfv}(#1)}
\newcommand{\efv}[1]{\symfont{efv}(#1)}
\newcommand{\fn}[1]{\symfont{names}(#1)}
\newcommand{\ov}[1]{\symfont{ov}(#1)}

\newcommand{\emptylist}{\epsilon}

\newcommand{\seasym}{sea}
\newcommand{\cons}{{:}}

\newcommand{\tomachhole}[1]{\leadsto_{#1}}
\newcommand{\tomach}{\tomachhole{}}
\newcommand{\tomachine}{\tomach}
\newcommand{\tomachsea}{\tomachhole{\seasym}}
\newcommand{\tomachseaone}{\tomachhole{\seasym_1}}
\newcommand{\tomachseatwo}{\tomachhole{\seasym_2}}
\newcommand{\tomachseathree}{\tomachhole{\seasym_3}}
\newcommand{\tomachseafour}{\tomachhole{\seasym_4}}
\newcommand{\tomachseafive}{\tomachhole{\seasym_5}}
\newcommand{\tomachseasix}{\tomachhole{\seasym_6}}
\newcommand{\tomachseaseven}{\tomachhole{\seasym_7}}
\newcommand{\tomachseaeight}{\tomachhole{\seasym_8}}

\renewcommand{\l}{\lambda}

\newcommand{\dom}{\symfont{dom}}
\newcommand{\domp}[1]{\symfont{dom}(#1)}

\newcommand{\names}[1]{\symfont{names}(#1)}

\newcommand{\nctxholep}[2]{\langle#1\rangle_{#2}}
\newcommand{\nctxhole}[1]{\nctxholep{\cdot}{#1}}

\newcommand{\name}{a}
\newcommand{\nametwo}{b}
\newcommand{\namethree}{c}

\newcommand{\pool}{P}
\newcommand{\pooltwo}{\pool'}

\newcommand{\rename}[1]{#1^\alpha}

\newcommand{\appr}{\mathbb{A}}

\newcommand{\mach}{\symfont{M}}

\newcommand{\state}{Q}
\newcommand{\statetwo}{\state'}
\newcommand{\statethree}{\state''}

\newcommand{\States}{\symfont{States}}

\newcommand{\compilrelsym}{\triangleleft}
\newcommand{\compilrel}[2]{#1\compilrelsym#2}
\newcommand{\decode}[1]{\underline{#1}}

\newcommand{\prsym}{\symfont{pr}}
\newcommand{\tomachpr}{\tomachhole{\prsym}}

\newcommand{\run}{r}
\newcommand{\runtwo}{\run'}

\newcommand{\size}[1]{|#1|}
\newcommand{\sizep}[2]{|#1|_{#2}}

\newcommand{\sizepr}[1]{\sizep{#1}{\symfont{pr}}}
\newcommand{\sizesea}[1]{\sizep{#1}{\symfont{sea}}}

\newcommand{\ie}{{\em i.e.}\xspace}

\newcommand{\ih}{{\emph{i.h.}}\xspace}

\newcommand{\tostrat}{\Rew{\symfont{str}}}

\newcommand{\deriv}{\ensuremath{e}}

\newcommand{\derivtwo}{\ensuremath{e'}}
\newcommand{\derivthree}{\ensuremath{e''}}
\newcommand{\derivfour}{\ensuremath{e'''}}

\newcommand{\env}{E}

\newcommand{\twostate}[2]{(#1 \mid #2)}

\newcommand{\tomachaxmone}{\tomachhole{\axmone}}
\newcommand{\tomachaxmtwo}{\tomachhole{\axmtwo}}
\newcommand{\tomachaxmtwob}{\tomachhole{\axmtwo'}}
\newcommand{\tomachaxeone}{\tomachhole{\axeone}}
\newcommand{\tomachaxetwo}{\tomachhole{\axetwo}}

\newcommand{\tomachlolli}{\tomachhole{\lolli}}
\newcommand{\tomachbang}{\tomachhole{\bang}}

\newcommand{\set}[1]{\{#1\}}

\newcommand{\ectx}{E}
\newcommand{\ectxp}[1]{\ectx\ctxholep{#1}}

\newcommand{\ectxtwo}{\ectx'}
\newcommand{\ectxtwop}[1]{\ectxtwo\ctxholep{#1}}

\newcommand{\ectxthree}{\ectx''}

\newcommand{\reflemma}[1]{Lemma~\ref{l:#1}}
\newcommand{\reflemmap}[2]{Lemma~\ref{l:#1}.\ref{p:#1-#2}}

\newcommand{\refthm}[1]{Theorem~\ref{thm:#1}}

\newcommand{\refprop}[1]{Proposition~\ref{prop:#1}}
\newcommand{\refsect}[1]{Sect.~\ref{sect:#1}}

\newcommand{\refapp}[1]{Appendix~\ref{app:#1}}

\ifthenelse{\boolean{easychairstyle}}{}{
  \ifthenelse{\boolean{lncsstyle}}{}{	
	
  }
}
\newcommand{\reffig}[1]{Fig.~\ref{fig:#1}}

\newcommand{\refdef}[1]{Def.~\ref{def:#1}}

\newcommand{\namctxholea}{\langle \cdot\rangle_{\name}}
\newcommand{\namctxholep}[2]{\langle #1\rangle_{#2}}
\newcommand{\namctx}{\mathbbm{C}}
\newcommand{\namctxtwo}{\namctx'}
\newcommand{\namctxthree}{\namctx''}

\newcommand{\namctxp}[2]{\namctx\namctxholep{#1}{#2}}
\newcommand{\namctxtwop}[2]{\namctxtwo\namctxholep{#1}{#2}}

\newcommand{\namvctx}{\mathbbm{V}}

\newcommand{\mute}{\circ}

\newcommand{\Gsym}{\symfont{G}}
\newcommand{\togp}[1]{\Rew{\Gsym #1}}
\newcommand{\tog}{\togp{}}

\newcommand{\dfv}[1]{\symfont{dv}(#1)}

\newcommand{\vgctx}{V_{G}}

\newcommand{\gctx}{G}
\newcommand{\gctxp}[1]{\gctx\ctxholep{#1}}

\newcommand{\lRew}[1]{\; \mbox{}_{#1}{\leftarrow}}
\newcommand{\lRewn}[1]{\; \mbox{}^{*}_{#1}{\leftarrow}\ }

\newcommand{\multiForm}{\Gamma}
\newcommand{\multiFormtwo}{\Delta}

\newcommand{\multiform}{\multiForm}

\newcommand{\hastype}{\,{:}\,}

\newcommand{\form}{A}
\newcommand{\formtwo}{B}
\newcommand{\formthree}{C}

\newcommand{\twostatetrans}[5]{\twostate{#1}{#2}  #3 \twostate{#4}{#5}}

\newcommand{\tosesame}{\tomachhole{\symfont{SESAME}}}

\newcommand{\nfsea}[1]{\symfont{nf_{sea}}(#1)}

\newcommand{\bigo}{{\mathcal{O}}}

\newcommand{\aform}{X_{\msym}}

\newcommand{\rightRuleSym}{r}
\newcommand{\leftRuleSym}{l}

\newcommand{\tensRightRule}{\tens_\rightRuleSym}
\newcommand{\tensLeftRule}{\tens_\leftRuleSym}

\newcommand{\lolliRightRule}{\lolli_\rightRuleSym}
\newcommand{\lolliLeftRule}{\lolli_\leftRuleSym}

\newcommand{\bangRightRule}{\bang_\rightRuleSym}
\newcommand{\bangLeftRule}{\bang_\leftRuleSym}

\newcommand{\contr}{{\mathsf{c}}}
\newcommand{\contrRule}{\contr}
\newcommand{\weak}{{\mathsf{w}}}
\newcommand{\weakRule}{\weak}

\newcommand{\cut}{\symfont{cut}}

\newcommand{\gcdec}[1]{\symfont{GC}(#1)}

\newcommand{\tognotw}{\Rew{\Gsym_{\neg\wsym}}}

\newcommand{\Rule}{\symfont{r}}

\newcommand{\techrep}[1]{\ifthenelse{\boolean{withproofs}}{#1}{}}
\newcommand{\camerar}[1]{\ifthenelse{\boolean{withproofs}}{}{#1}}
\usepackage{tikz}
\usetikzlibrary{calc}
\usetikzlibrary{matrix,arrows}
\usetikzlibrary{decorations.shapes}
\usetikzlibrary{decorations.text}
\usetikzlibrary{decorations.pathmorphing}
\usetikzlibrary{decorations.markings}
\usetikzlibrary{positioning}

\tikzset{
node distance=1.3cm, auto,
every node/.style={font=\scriptsize },
ocenter/.style={baseline={([yshift=-.5ex, xshift=-.5ex]current bounding box)}},  
labelBeginAbove/.style={postaction={decorate,decoration={markings,mark=at position 0 with {\node[inner sep= 0.6pt, above=1pt]{\tiny #1};}} } },
labelBeginBelow/.style={postaction={decorate,decoration={markings,mark=at position 0 with {\node[inner sep= 0.6pt, below=1pt]{\tiny #1};}}}},
labelEndAbove/.style={postaction={decorate,decoration={markings,mark=at position 1 with {\node[inner sep= 0.6pt, above=2pt]{\tiny #1};}}}},
labelEndBelow/.style={postaction={decorate,decoration={markings,mark=at position 1 with {\node[inner sep= 0.6pt, below=2pt]{\tiny #1};}}}},
labelEndRight/.style={postaction={decorate,decoration={markings,mark=at position 1 with {\node[inner sep= 0.6pt, right=2pt]{\tiny #1};}}}},
labelEndLeft/.style={postaction={decorate,decoration={markings,mark=at position 1 with {\node[inner sep= 0.6pt, left=2pt]{\tiny #1};}}}}
}

\newcommand{\med}[2]{
($(#1)!.5!(#2)$)
}

\usepackage{tabularray}

\title{IMELL Cut Elimination with Linear Overhead}
\author{Beniamino Accattoli}{Inria \& LIX, \'Ecole Poytechnique, UMR 7161, France}{beniamino.accattoli@inria.fr}{https://orcid.org/0000-0003-4944-9944}{}
\author{Claudio {Sacerdoti Coen}}{Alma Mater Studiorum - Università di Bologna, Italy}{}{}{}

\funding{Sacerdoti Coen has been supported by the INdAM-GNCS project
  "Modelli Composizionali per l'Analisi di Sistemi Reversibili Distribuiti" and by the Cost Action CA2011 EuroProofNet.}

\supplement{}
\supplementdetails[subcategory={Source Code}, cite={}, swhid={}]{Software}{https://github.com/sacerdot/sesame/}

\authorrunning{B. Accattoli and C. Sacerdoti Coen} 

\Copyright{Beniamino Accattoli and Claudio Sacerdoti Coen} 
\ccsdesc{Theory of computation → Operational semantics}

\EventEditors{Jakob Rehof}
\EventNoEds{1}
\EventLongTitle{9th International Conference on Formal Structures for Computation and Deduction (FSCD 2024)}
\EventShortTitle{FSCD 2024}
\EventAcronym{FSCD}
\EventYear{2024}
\EventDate{July 10-13, 2024}
\EventLocation{Tallinn, Estonia}
\EventLogo{}
\SeriesVolume{299}
\ArticleNo{21}

\begin{document}
%
\maketitle

\begin{abstract}
Recently, Accattoli introduced the Exponential Substitution Calculus (ESC) given by untyped proof terms for Intuitionistic Multiplicative Exponential Linear Logic (IMELL), endowed with rewriting rules at-a-distance for cut elimination. He also introduced a new cut elimination strategy, dubbed \emph{the good strategy}, and showed that its number of steps is a time cost model with polynomial overhead for the ESC/IMELL, and the first such one.

Here, we refine Accattoli's result by introducing an abstract machine for ESC and proving that it implements the good strategy and computes cut-free terms/proofs within a \emph{linear} overhead.
\keywords{Lambda calculus, linear logic, abstract machines.}
\end{abstract}

\section{Introduction}
\label{sect:intro}
One of the advances of the last decade in the theory of $\l$-calculus is the study of reasonable cost models carried out by Accattoli, Dal Lago, and co-authors \cite{DBLP:conf/rta/AccattoliL12,DBLP:journals/corr/AccattoliL16,DBLP:conf/lics/AccattoliC15,DBLP:conf/lics/AccattoliLV21,DBLP:conf/lics/AccattoliLV22}, completing the research program started by Dal Lago and Martini almost 20 years ago \cite{DBLP:conf/cie/LagoM06,DBLP:journals/tcs/LagoM08,DBLP:conf/icalp/LagoM09,DBLP:conf/fopara/LagoM09}.

\subparagraph{The General Problem.} Focusing on time, the underlying problem is whether the number of $\beta$-steps (of a fixed evaluation strategy) can be taken as a measure of time that is polynomially equivalent to that of random access machines (RAMs); or, equivalently, to that of Turing machines. This is a priori unclear because there are families of terms suffering of \emph{size explosion}: the size of terms grows exponentially with the number of $\beta$-steps, independently of the adopted evaluation strategy. The difficulty is finding a strategy, together with a simulation on RAMs of the strategy, working within a overhead polynomial in:
\begin{itemize}
\item \emph{Length}: the number of steps of the strategy, and 
\item \emph{Input}: the size of the initial term.
\end{itemize}
When a strategy admits such a simulation we say that it is a \emph{a polynomial cost model}. Note that the number of steps of the strategy can be whatever (it need not be polynomial), it is the \emph{overhead} of the simulation that has to be polynomial in the above parameters.

The overhead is required to be polynomial because reasonable frameworks are not necessarily linearly related; this is the reason why the complexity class $\P$ is so important. For instance, Turing machines simulate RAMs only within a quadratic overhead. When possible, however, it is interesting to know what is the exact degree of the overhead with respect to RAMs, which are the commonly accepted abstraction of modern computers. When the overhead is linear---which is the optimal situation---then the polynomial cost model is trustable in practice: it does not hide an acceptable but possibly costly polynomial overhead.

Today, it is well-known that strategies such as weak head, head, and leftmost $\beta$-reduction provide polynomial time cost models, and that the same is true for their call-by-value and call-by-need variants \cite{DBLP:conf/fpca/BlellochG95,DBLP:conf/birthday/SandsGM02,DBLP:journals/tcs/LagoM08,DBLP:conf/rta/AccattoliL12,DBLP:journals/corr/AccattoliL16,DBLP:conf/lics/AccattoliC15,DBLP:conf/lics/AccattoliLV21,DBLP:conf/ppdp/BiernackaCD21,DBLP:conf/csl/AccattoliL22,DBLP:journals/pacmpl/BiernackaCD22}. 

\subparagraph{LSC and the Sub-Term Property.} The key tool for the polynomially bounded simulations underlying the recent advances is Accattoli and Kesner's \emph{linear substitution calculus} (shortened to LSC) \cite{DBLP:conf/rta/Accattoli12,DBLP:conf/popl/AccattoliBKL14}, which is an intermediary setting between the $\l$-calculus and RAMs. 

The LSC is a simple $\l$-calculus with explicit substitutions refining a previous calculus by Milner \cite{DBLP:journals/entcs/Milner07,DBLP:journals/corr/abs-2312-13270} and where evaluation is \emph{micro-step}, that is, substitution acts on a variable occurrence at a time, rather than on all at once (which is \emph{small-step}). A distinguished feature of the LSC is its rewriting rules \emph{at a distance}, that is, rules where explicit substitutions do not percolate through the term structure but rather act through some contexts. The relevance of the LSC for reasonable time is linked to the \emph{sub-term property} of its standard strategies \cite{DBLP:journals/corr/AccattoliL16}: at any point, only sub-terms of the initial term are duplicated, thus allowing one to bound the cost of each duplication with the size of the input. 

No strategy of the $\l$-calculus has the sub-term property, which, as discussed by Accattoli \cite{DBLP:journals/lmcs/Accattoli23}, is the essence of the size explosion problem. The LSC crystallizes exactly what is needed for refining the $\l$-calculus to retrieve the sub-term property and circumvent size explosion.

\subparagraph{LSC and Abstract Machines.} It is also well-known that strategies of the LSC can usually be implemented with environment-based abstract machines whose overhead is \emph{bi-linear}, that is, linear in both the input and length parameters, as shown by Accattoli et al. \cite{DBLP:conf/icfp/AccattoliBM14}. Abstract machines handle three tasks, namely the \emph{decomposition of the substitution process}, the \emph{search for redexes}, and \emph{$\alpha$-renaming}. The difference between the LSC and abstract machines is that the LSC only handles the substitution process, the critical one for avoiding size explosion. Search and $\alpha$-renaming, indeed, tend to take only a linear overhead. That said, handling search and $\alpha$-renaming provides an in-depth understanding of evaluation and it is mandatory for obtaining precise bounds from the cost model.

\subparagraph{Lifting to Linear Logic.} In 2022, Accattoli started a generalization of the mentioned recent results for the $\l$-calculus to the wider framework of linear logic \cite{DBLP:journals/lmcs/Accattoli23}. Linear logic can be seen as a micro-step system with a tight control over duplications, somewhat similar to the LSC. Accattoli's starting point is the observation that, despite such similarity, no reasonable cut elimination strategy for linear logic was available. In particular, no cut elimination strategy with the sub-term property was known. 

In \cite{DBLP:journals/lmcs/Accattoli23}, he introduces a generalization of the LSC to intuitionistic multiplicative exponential linear logic (IMELL). His \emph{exponential substitution calculus} (ESC) is an untyped calculus of proof terms for the sequent calculus proofs of IMELL, endowed with cut elimination at a distance and having IMELL as typing system. He then designs a new cut elimination strategy, dubbed \emph{good strategy}, which has the sub-term property and which is a polynomial cost model for ESC, the first such cost model for an expressive fragment of linear logic.

\subparagraph{This Paper.} Here, we design an environment-based abstract machine implementing the good strategy of ESC, show that it has bi-linear overhead, and provide an OCaml implementation. The result is not surprising, and yet we believe that it is interesting, for various reasons. Firstly, it is the first linear overhead result for an expressive fragment of linear logic. It allows one to extract precise time bounds from the length of cut elimination, and can thus be useful, for instance, in linear-logic-based implicit computational complexity.

Secondly, our study differs in many details from other environment-based machines. There are in fact at least two literatures related to our work (surveyed in the next paragraph). A recent one about cost analyses of machines for $\l$-calculi and an older one about machines for linear logic, here referred to as \emph{linear machines}. The main differences are the following:
\begin{itemize}
\item \emph{Natural deduction $vs$ sequent calculus}: machines usually implement calculi based on natural deduction, while ESC is based on sequent calculus, and some aspects are different. For instance, our machine has no argument stack (as it is expected, after Herbelin \cite{DBLP:conf/csl/Herbelin94}).

\item \emph{Micro-step source and no structural equivalence}: the implemented calculus is usually small-step, while here it is micro-step. In the literature, when the calculus is micro-step it is usually implemented by the machine up to a notion of structural equivalence. Here, there is no need of structural equivalence.

\item \emph{Strong $vs$ weak}: our machine computes cut-free proofs, while other linear machines in the literature do not, as they usually only perform weak (or surface) evaluation (that is, they do not evaluate under abstraction and inside promotions). To our knowledge, ours is the first strong linear machine.

\item \emph{Strong with no backtracking}: machines for strong evaluation usually have transitions for sequentially backtracking once the evaluation of a sub-term is over. Here, we adopt a recent new technique for strong machines by Accattoli and Barenbaum \cite{DBLP:conf/aplas/AccattoliB23}, that avoids sequential backtracking. The idea is to assign distinct jobs to each sub-term and simply jumping to the next job once the current one is over. This approach structures the machine in a very different way and provides drastically simpler strong machines.

\item \emph{Complexity analysis}: the complexity of linear machines is never studied in the literature.

\item \emph{OCaml implementation}: linear machines are usually studied theoretically and never implemented. We provide an OCaml implementation verifying our complexity analysis. In particular, it is the first implementation of the technique for strong machines of \cite{DBLP:conf/aplas/AccattoliB23}.
\end{itemize}
Summing up, our machine is simple, given the complex setting that it implements, thanks to the absence of stacks and backtracking. In fact, it uses just one (non-trivial) data structure. 

Methodologically, the main difference between our work and the literature on linear machines is that those machines were developed to provide insights about $\l$-calculi and functional programming (such as no garbage collection, in-place updates, and single pointer property) while here we proceed the other way around, using the recent theory of the $\l$-calculus to provide insights about linear logic.

\subparagraph{Architecture of the Result.} The good strategy that we implement is a non-deterministic but \emph{diamond} strategy, where diamond means that the choices do not affect termination nor evaluation lengths (otherwise the number of steps would not be a well-defined cost model). Our implementation is in two phases. The first and main one, performed by a deterministic machine dubbed \emph{SESAME} (Strong Exponential Substitution Abstract Machine without Erasure), never erases. The second one is a simple garbage collection pass over the output of SESAME, that produces a cut-free term. For the correctness of SESAME, we thus need to relate a non-deterministic strategy and a deterministic machine, which is slightly unusual. For that, we follow the abstract recipe of Accattoli et al. \cite{DBLP:conf/lics/AccattoliCC21} which is here simplified because our case is not up to structural equivalence and there is no implosiveness.

\subparagraph{Related Work.} Accattoli and co-authors studied at length the overhead of machines for $\l$-calculi both micro-step \cite{DBLP:conf/icfp/AccattoliBM14,DBLP:journals/corr/AccattoliBM15,DBLP:conf/lics/AccattoliC15} and small-step \cite{DBLP:conf/wollic/Accattoli16,DBLP:conf/ppdp/AccattoliB17,DBLP:journals/scp/AccattoliG19,DBLP:conf/ppdp/AccattoliCGC19,DBLP:conf/lics/AccattoliCC21,DBLP:conf/lics/AccattoliLV22,DBLP:conf/aplas/AccattoliB23}. Biernacka et al. similarly studied machines for strong call-by-value \cite{DBLP:conf/ppdp/BiernackaCD21} and strong call-by-need \cite{DBLP:journals/pacmpl/BiernackaCD22}.

Lafont considered the first linear logic (LL) machine, based on categorical combinators for intuitionistic LL (ILL) \cite{DBLP:journals/tcs/Lafont88}. Abramsky considered an environment-based machine for ILL and a chemical-style machine for LL, both doing only weak evaluation \cite{DBLP:journals/tcs/Abramsky93}. The latter machine was then studied by Mikami and Akama \cite{DBLP:conf/tlca/MikamiA99} and Sato and co-authors \cite{DBLP:journals/entcs/SatoSY02,DBLP:journals/entcs/MackieS08}. Turner and Wadler study the use of IMELL for memory management and give two machines, one with the single pointer property and one without it but with a memoization mechanism \cite{DBLP:journals/tcs/TurnerW99}. Alberti and Ritter also deal with the single pointer property \cite{alberti1998efficient}. Bonelli gives a machine based on an unusual sequent calculus for ILL \cite{DBLP:journals/entcs/Bonelli06}. 

The interaction abstract machine (IAM) is an unusual token-based machine for LL proofs first studied by Danos and Regnier \cite{DBLP:journals/entcs/DanosR96} and Mackie \cite{DBLP:conf/popl/Mackie95}. Recent work by Accattoli et al. showed that the IAM overhead is unreasonable for both time and space \cite{DBLP:journals/pacmpl/AccattoliLV21,DBLP:conf/lics/AccattoliLV21,DBLP:phd/hal/Vanoni22}. Mackie also studied interaction nets-based implementations of LL \cite{DBLP:journals/tcs/Mackie00,DBLP:journals/iandc/MackieP02}.

\section{The Exponential Substitution Calculus}
\label{sect:esc}
In this section, we briefly recall Accattoli's exponential substitution calculus (ESC) \cite{DBLP:journals/lmcs/Accattoli23} which is an untyped calculus having IMELL as system of simple types. 

For more explanations, we refer the reader to \cite{DBLP:journals/lmcs/Accattoli23}. The main differences with respect to \cite{DBLP:journals/lmcs/Accattoli23} is that, for lack of space, we omit tensor and we only deal with micro-step rewriting rules (omitting the small-step exponential rule, which is derivable). In IMELL, linear implication is 'more important' than tensor, as tensors are not needed to simulate the $\l$-calculus. \techrep{In \refapp{tensors}}\camerar{In Appendix B of the technical report \cite{accattoli2024imell}}, we explain how to straightforwardly extend our study to tensor.

\begin{figure}[t!]
\centering
\begin{tabular}{c}
$\begin{array}{r@{\hspace{.4cm}} rll@{\hspace{.4cm}}|@{\hspace{.4cm}} r@{\hspace{.4cm}} rlllll@{\hspace{.4cm}} r@{\hspace{.4cm}} rlllll}
\textsc{Mult. values} &\mval & \grameq & \mvar \mid \la\var\tm
&
\textsc{Values} &\val,\valtwo & \grameq & \mval \mid \exval\\
\textsc{Exp. values} &\exval & \grameq & \evar \mid  \bang\tm
&
\textsc{Terms} &\tm&\grameq& \val  \mid \cuta\val\var\tm \mid  \suba\mvar\val\var\tm \mid \dera\evar\var\tm
\end{array}$

\\[5pt]
\hline
\\[-10pt]

$     \arraycolsep=2pt\begin{array}{r@{\hspace{.3cm}} rll@{\hspace{.3cm}}r@{\hspace{.3cm}} rllll}
\textsc{Cut ctxs} &\env & \grameq & \ctxhole 
\mid \cuta\val\var\env 
&
\textsc{Left ctxs} &\lctx & \grameq & \ctxhole 
\mid \cuta\val\var\lctx \mid \suba\mvar\val\var\lctx  \mid \dera\evar\var\lctx
\\[4pt]
\textsc{Value ctxs} & \vctx & \grameq & \ctxhole \mid \la\var\ctx \mid \bang\ctx 
&
\textsc{Ctxs} & \ctx & \grameq & \vctx \mid \cuta\vctx\var \tm  \mid \suba\mvar\vctx\var\tm \mid \lctxp\ctx
\end{array}$
\\[8pt]
\hline
\\[-8pt]
$\begin{array}{rlll}
\multicolumn{3}{c}{\textsc{Micro-step root multiplicative rules}}

\\
\cuta{\mval}\mvar\ctxfp\mvar & \rtoaxmone & \ctxfp\mval
\\
\cuta{\mvartwo}\mvar\ctxp{\suba\mvar\val\var \tm} & \rtoaxmtwo & \ctxp{\suba\mvartwo\val\var \tm}
\\
\cuta{\la\vartwo\tmtwo}\mvar \ctxp{\suba\mvar\val\var \tm}
& \rtololli & 
\ctxp{ \cuta\val\vartwo \lctxp{\cuta\valtwo\var \tm }} 

& \mbox{with }\tmtwo = \lctxp{\valtwo}
\\[4pt]
\multicolumn{3}{c}{\textsc{Micro-step root exponential rules}}
\\
\cuta{\exval} \evar \ctxfp{\evar}
& \rtoaxeone & 
\cuta{\exval}\evar\ctxfp{\exval}
\\
\cuta{\evartwo}\evar\ctxp{\dera\evar\var \tm}
& \rtoaxetwo &
\cuta{\evartwo}\evar\ctxp{\dera\evartwo\var \tm}
\\
\cuta{\bang\tmtwo}\evar\ctxp{\dera\evar\var \tm}
& \rtobang &
\cuta{\bang\tmtwo}\evar\ctxp{\lctxp{\cuta\val\var \tm}} 
& \mbox{with }\tmtwo = \lctxp{\val}
\\
\cuta{\exval}\evar \tm
& \rtow & 
\tm \ \ \ \  \mbox{if }\evar\notin\fv\tm
\end{array}$
\\[38pt]
\hline
\\[-10pt]
	\begin{tabular}{c@{\hspace{.4cm}} cc}
\textsc{Contextual closure}
&\AxiomC{$\tm \rootRew{a} \tmtwo$}
\UnaryInfC{$\ctxp{\tm} \Rew{a} \ctxp{\tmtwo}$}
\DisplayProof
		& 
		for $a \in \set{
\axmone,\axmtwo,\lolli, \axeone,\axetwo,\bang,\wsym}$
	\end{tabular}
\\[8pt]
\hline
\\[-8pt]
$\begin{array}{rlll}
\multicolumn{4}{c}{\textsc{Notations}}
\\
\textsc{Multiplicative} & \tom  &\defeq &  \toaxmone\cup \toaxmtwo \cup \tololli
\\
\textsc{Exponential} & \toe  &\defeq & \toaxeone \cup \toaxetwo \cup \tobang \cup \tow
\\
\textsc{ESC} &\toesc  &\defeq & \tom \cup \toe
\\
\textsc{Non-erasing ESC} &\tomsnw  &\defeq & \tom \cup \toaxeone \cup \toaxetwo \cup \tobang
	\end{array}$
\end{tabular}
\caption{The Exponential Substitution Calculus (ESC).}
\label{fig:esc}
\end{figure} 
\subparagraph{Values and Terms.} The grammars of ESC are in the upper part of \reffig{esc}. Variables are of two disjoint kinds, multiplicative and exponential, and we refer to variables of unspecified kind using $\var,\vartwo,\varthree$. Values are the proof terms associated to axioms or to the right rules and, beyond variables, are \emph{abstractions} $\la\var\tm$ and \emph{promotions} $\bang\tm$. The proof terms decorating left rules are \emph{subtractions} $\suba\mvar\val\var\tm$, \emph{derelictions} $\dera\evar\var\tm$,  and \emph{cuts} $\cuta\val\var\tm$, which is in red because of its special role.  Note that cuts and subtractions are \emph{split}, that is, have values (rather than terms) as left sub-terms. 
The constructors $\la\var\tm$, $\suba\mvar\val\var\tm$, $\dera\evar\var\tm$, and $\cuta\val\var\tm$ bind $\var$ in $\tm$. We identify terms up to $\alpha$-renaming. Free variables (resp.  multiplicative/exponential variables) are defined as expected, and noted $\fv\tm$ (resp. $\mfv\tm$ and $\efv\tm$). We use $\size\tm$ for the number of constructors in $\tm$, and $\sizep\tm\var$ for the number of free occurrences of $\var$ in $\tm$.
There is a notion of \emph{proper term} ensuring the linearity of multiplicative variables and the exponential boundary of promotions, \techrep{detailed in \refapp{esc}}\camerar{detailed in Appendix A of the tech report \cite{accattoli2024imell}}. The only relevant case is: \emph{$\bang\tm$ is proper if $\tm$ is proper and $\mfv\tm = \emptyset$}. In the following, terms are assumed to be proper.

\subparagraph{Contexts and Plugging.} The broadest notion of context that we consider is \emph{general contexts} $\ctx$, which simply allow the hole $\ctxhole$ to replace any sub-term in a term. Because of split cuts and subtractions (that is, the fact that their left sub-term is a value rather than a term), the definition relies on the auxiliary notion of \emph{value context} $\vctx$. The definition also uses \emph{left contexts} $\lctx$, which are contexts under left constructors (or, for binary left constructors, under the right sub-term) that play a key role in the system---their use in defining $\ctx$ is just to keep the grammar compact. We also need \emph{cut contexts} $\ectx$, which are noted with $\ectx$ because they shall play the role of machine environments in \refsect{BAM}.

A fact used pervasively is that every term $\tm$ writes, or \emph{splits} uniquely as $\tm = \lctxp\val$. For instance $\dera\evar\mvar\cuta\val\mvartwo \la\evarthree\suba\mvar\mvartwo\evartwo\evartwo$ splits as $\lctx = \dera\evar\mvar\cuta\val\mvartwo\ctxhole$ and $\val =\la\evarthree\suba\mvar\mvartwo\evartwo\evartwo$. 

Because of split cuts and subtractions, the definition of plugging $\ctxp\tm$ (or $\ctxp\ctxtwo$) of a term $\tm$ (or a context $\ctxtwo$) in a context $\ctx$ is slightly tricky, as it has to preserve the split shape. We refer the reader to \cite{DBLP:journals/lmcs/Accattoli23} for such details, the definition is mostly as expected (the only two subtle cases are for $\cuta\ctxhole\var\tm$ and $\suba\mvar\ctxhole\var$). Plugging can capture variables and we use $\ctxfp\tm$ when we want to prevent it.

\begin{figure}[t!]
\begin{center}
\begin{tblr}{c}
$\begin{tblr}{rclll}
\textsc{Formulas} & \form,\formtwo,\formthree & \grameq & \aform \mid \form \tens \formtwo \mid \form \lolli \formtwo \mid \bang\form
\end{tblr}$

	\\[5pt]\hline
\begin{tblr}{cc}	
\SetCell[c=2]{c}\textsc{Multiplicative rules}
\\[3pt]
	\AxiomC{$\form \neq \bang \formtwo$}
	\RightLabel{$\ax_{\msym}$}
	\UnaryInfC{$  \mvar\hastype\form \vdash \mvar\hastype \form$}	
	\DisplayProof
&
	\AxiomC{$  \var\hastype\form,\multiForm \vdash \tm\hastype\formtwo$}
	\RightLabel{$ \lolliRightRule $}
	\UnaryInfC{$  \multiForm \vdash \la\var\tm\hastype\form \lolli \formtwo$}
	\DisplayProof
\\[6pt]
\SetCell[c=2]{c}{\AxiomC{$  \multiForm \vdash \lctxp\val \hastype\form$}
	\AxiomC{$  \multiFormtwo, \var\hastype\formtwo \vdash \tm\hastype\formthree$}
		\AxiomC{$\multiForm\#(\multiFormtwo, \var\hastype\formtwo)$, $\mvar$ fresh}
	\RightLabel{$ \lolliLeftRule $}
	\TrinaryInfC{$  \multiForm, \multiFormtwo, \mvar\hastype \form \lolli \formtwo \vdash \lctxp{\suba\mvar\val\var \tm} \hastype\formthree$}
	\DisplayProof}
\end{tblr}
	\\[8pt]\hline
\begin{tblr}{cc}	
\SetCell[c=2]{c}\textsc{Exponential rules}
\\[3pt]
	&
	\AxiomC{}
	\RightLabel{$\ax_{\esym}$}
	\UnaryInfC{$  \evar\hastype\bang\form \vdash \evar\hastype \bang\form$}	
	\DisplayProof
	\\[6pt]
	\AxiomC{$  \multiForm, \var\hastype\form \vdash \tm\hastype\formtwo$}
	\AxiomC{$\evar$ fresh}
	\RightLabel{$ \bangLeftRule$}
	\BinaryInfC{$  \multiForm, \evar\hastype\bang\form \vdash \dera\evar\var\tm \hastype\formtwo$}
	\DisplayProof
	&
		\AxiomC{$  \bang\multiForm \vdash \tm\hastype\form$}
	\RightLabel{$ \bangRightRule$}
	\UnaryInfC{$  \bang\multiForm \vdash \bang\tm\hastype\bang\form$}
	\DisplayProof
	\\[6pt]
	\AxiomC{$  \multiForm \vdash \tm\hastype\form$}
				\AxiomC{$\evar$ fresh}
	\RightLabel{$ \weakRule $}
	\BinaryInfC{$  \multiForm, \evar\hastype\bang\formtwo \vdash \tm\hastype \form$}
	\DisplayProof
&
		\AxiomC{$  \multiForm, \evar\hastype\bang\formtwo, \evartwo\hastype\bang\formtwo  \vdash \tm\hastype\form$}
	\RightLabel{$ \contrRule $}
	\UnaryInfC{$  \multiForm, \evar\hastype\bang\formtwo \vdash \cutsub\evar\evartwo \tm \hastype \form$}
	\DisplayProof
\end{tblr}
	\\[8pt]\hline
\begin{tblr}{c}	
\textsc{Cut}
\\[3pt]
	\AxiomC{$\multiForm   \vdash \lctxp\val \hastype \form$}
	\AxiomC{$\multiFormtwo, \var\hastype \form \vdash \tm\hastype\formtwo$}
	\AxiomC{$\multiForm\#(\multiFormtwo, \var\hastype \form)$}
 	\RightLabel{$\cut$}
 	\TrinaryInfC{$  \multiForm, \multiFormtwo\vdash \lctxp{\cuta\val\var\tm} \hastype\formtwo$} 	
	\DisplayProof	
\end{tblr}
\end{tblr}
\end{center}
\caption{IMELL has a type system for ESC.}
\label{fig:esc-types}
\end{figure}

\subparagraph{Types.} The formulas of IMELL, and the deductive rules of the sequent calculus annotated with ESC terms, are in \reffig{esc-types}. They are taken directly from \cite{DBLP:journals/lmcs/Accattoli23}. The typing system for  ESC is exactly the standard sequent calculus for IMELL. Both formulas and rules are standard, but for the decoration with proof terms and the side conditions about variable names of the form $\multiForm\#\multiFormtwo$, which is a shortcut for $\dom\multiForm\cap\dom\multiFormtwo = \emptyset$. Linear implication $\lolli$ is also referred to as \emph{lolli}. The only atomic formula that we consider, $\aform$, is multiplicative. There is no multiplicative unit because, in presence of the exponentials, $1$ can be simulated by $!(\aform \lolli\aform)$. We distinguish between multiplicative and exponential axioms, in order to decorate them with the corresponding kind of variable.

Note that the weakening and contraction rules do not add constructors to terms. This is crucial in order to keep the calculus manageable. Note also that the decorations of the cut and $\lolliLeftRule$ rules are \emph{split}, as explained in the previous section. Clearly, typed terms are proper.  

In this paper, types are only referred to in \reflemma{clashfree-implies-cutfree} below.

\subparagraph{Multiplicative Cut Elimination Rules.} The rewriting rules are in \reffig{esc}. The ESC has three multiplicative rules, in particular two for axioms, depending on whether they are acted upon ($\toaxmone$) or used to rename another multiplicative (thus linear) variable ($\toaxmtwo$). Rule $\toaxmone$ is expressed generically for multiplicative values $\mval$ (that is, multiplicative variables $\mvar$ and abstractions $\la\var\tm$). In $\toaxmone$, $\toaxmtwo$, and $\tololli$, it is silently assumed that $\ctx$ does not capture $\mvar$ in $\toaxmtwo$ and $\tololli$ (what is noted $\ctxfp\mvar$ in $\toaxmone$, while $\ctxholefp\cdot$ is not used in $\toaxmtwo$, and $\tololli$ because $\ctx$ might capture other variables in $\tm$ and $\val$). Note that, since $\ctx$ cannot capture $\mvar$ in these rules and terms are assumed to be proper, the hole of $\ctx$ cannot be contained in a $!$; this kind of context is called a \emph{multiplicative context} in \cite{DBLP:journals/lmcs/Accattoli23}.

In $\tololli$, the rule has to respect split cuts, which is why, for writing the reduct, the sub-term $\tmtwo$ is split on-the-fly. An example of $\tololli$ step follows:
\begin{center}$\begin{array}{rlll}
\cuta{\la\evar\dera\evar\mvar\mvar}\mvartwo\suba\mvartwo{\bang\evartwo}\mvarthree\mvarthree 
& \tololli &
\cuta{\bang\evartwo}\evar\dera\evar\mvar\cuta\mvar\mvarthree \mvarthree.
\end{array}$\end{center}

\subparagraph{Exponential Rules.} There are also four exponential rules, with again two rules for axioms. Replacement of variables ($\toaxeone$) and erasure ($\tow$) are expressed generically for exponential values $\exval$ (that is, exponential variables $\evartwo$ and promotions $\bang\tm$), interaction with derelictions (in $\toaxetwo$ and $\tobang$) instead requires inspecting $\exval$. 
Rule $\tobang$ removes the dereliction, copies the promotion body, and puts it in a cut---in proof nets jargon, \emph{it opens the box}. To preserve the split shape, the body of the promotion is split and only the value is cut. An example: 
\begin{center}$\begin{array}{rlll}
\cuta{\bang\cuta\evartwo\evarthree \evarthree}\evar \la\mvar\dera\evar{\evar'} \suba\mvar{\evar'}\mvartwo\mvartwo
& \tobang &  
\cuta{\bang\cuta\evartwo\evarthree \evarthree}\evar\la\mvar\cuta\evartwo\evarthree \cuta{\evarthree}{\evar'} \suba\mvar{\evar'}\mvartwo\mvartwo.
\end{array}$\end{center}
Note that $\tobang$ entangles \emph{interaction with a dereliction} and \emph{duplication}, which is not what proof nets usually do (but this is what the LSC does).  It is silently assumed that $\ctx$ does not capture $\evar$ in $\toaxetwo$ and $\tobang$ but it might capture other variables in $\tm$.

\subparagraph{Comments about the Rewriting Rules.} Unsurprisingly, proper terms are stable by reduction \cite{DBLP:journals/lmcs/Accattoli23}. We also use a notion of position in terms.

\begin{definition}[Positions]
\label{def:position}
A \emph{position} in a term $\tm$ is a decomposition $\tm=\ctxp\tmtwo$ such that $\tmtwo$ is a sub-term of $\tm$.
\end{definition}

For later defining the good strategy, we identify a redex with its position, which is a context. Every step $\tm \tomsnw \tmtwo$ reduces a redex of shape $\tm=\ctxp{\cuta\val\var\ctxtwop{\tm_{\var}}}$ where $\tm_{\var}$ is an \emph{occurrence} of $\var$, \ie a sub-term of $\tm$ of shape $\var$, $\suba\var\val\vartwo\tmtwo$, or $\dera\var\vartwo\tmtwo$. The \emph{redex position of $\tomsnw$ steps} is the context $\ctxp{\cuta\exval\var\ctxtwo}$. The \emph{redex position of $\tow$ steps} is the context closing the root step. We write $\ctx:\tm\toesc\tmtwo$ for a redex of position $\ctx$ in $\tm$.

\subparagraph{Clashes.} The presence of many constructors in an untyped setting gives rise to \emph{clashes}, that is, irreducible cuts.
\begin{definition}[Clashing, smooth, clash-free terms]
A \emph{clash} is a term of the form $\cuta\mval\evar\tm$ or $\cuta\exval\mvar\tm$. A term $\tm$ is \emph{clashing} if it has a clash as sub-term. Moreover, $\tm$ is \emph{clash-free} if, co-inductively,
$\tm$ is not clashing and $\tmtwo$ is clash-free for any step $\tm \toms \tmtwo$.
\end{definition}
Clash-freeness is guaranteed by all forms of types (simple, polymorphic, recursive, and multi types, for instance), in particular by the IMELL types in \reffig{esc-types}, as ensured by the next lemma. Here, we adopt the untyped calculus but consider only clash-free terms, as in \cite{DBLP:journals/lmcs/Accattoli23}.
\begin{lemma}[\cite{DBLP:journals/lmcs/Accattoli23}]
\label{l:clashfree-implies-cutfree}
Let $\tm$ be a proper term.
\begin{enumerate}
\item If $\tm$ has no clashes and $\tm \not\toms$ then $\tm$ is cut-free.
\item If $\tm$ is typable (by the IMELL type system in \reffig{esc-types}) then $\tm$ is clash-free.
\end{enumerate}
\end{lemma}

\subparagraph{Postponement of Garbage Collection.} The erasing rule $\tow$, that models garbage collection, can be postponed. This point shall be crucial for our study.

\begin{proposition}[Postponement of garbage collection, \cite{DBLP:journals/lmcs/Accattoli23}]
\label{prop:postp-gc} 
If $\tm \toms^{*} \tmtwo$ then $\tm \tomsnw^{*} \tow^{*} \tmtwo$.
\end{proposition}

\subparagraph*{Basic Evaluation.} As an intermediary step towards our results, we consider also a restricted form of ESC evaluation, forbidding reduction under all constructors but the right sub-terms of cuts, and dub it \emph{basic evaluation}. It is one possible ESC analogous of weak evaluation in the $\l$-calculus (which forbids evaluation under abstraction). In $\l$-calculi with explicit substitutions, weak normal forms of closed terms are \emph{answers}, {i.e.} abstractions possibly surrounded by explicit substitutions. We obtain a similar property for basic evaluation.

\begin{definition}[Basic evaluation, answer]
A step $\ctx:\tm \toms \tmtwo$ (that is, a redex of position $\ctx$ in $\tm$) is \emph{basic} if $\ctx$ is a cut context, which is also noted with $\tm \Rew{\symfont{b}} \tmtwo$, or $\tm \Rew{\symfont{b}a} \tmtwo$ with $a\in\set{\axmone,\axmtwo,\lolli, \axeone,\axetwo,\bang,\wsym}$ if we want to specify the kind of step. Moreover, $\Rew{\symfont{b}\neg\wsym}$ denotes a non-erasing (that is, not $\wsym$) basic step. A term $\tm$ is an \emph{answer}  if $\tm = \ectxp\val$ with $\ectx$ a cut context and $\val$ not a variable. 
\end{definition}

\begin{toappendix}
\begin{lemma}
\label{l:weak-harmony}
\NoteProof{l:basic-harmony}
Let $\tm$ be a closed clash-free term. Either $\tm$ has a $\Rew{\symfont{b}\neg\wsym}$-step or it is an answer.\label{l:basic-harmony}
\end{lemma}
\end{toappendix}

Basic evaluation is enough to simulate in ESC the call-by-name/value weak evaluation of closed $\l$-terms, via Girard's translations to linear logic. Moreover, basic non-erasing evaluation is deterministic (and diamond, a property defined in \refsect{strategy}, when erasing steps are considered). We omit the proofs of these facts, because neither basic evaluation nor $\l$-terms are the focus of the paper.

\subparagraph*{Out Cuts and Garbage Collection.} We shall mainly deal with cuts that are not contained in any other cut, the \emph{out cuts}, which induce a notion of \emph{out variable}.
\begin{definition}[Out cuts, out variables]
The \emph{out cuts} of a term $\tm$ are those cuts in $\tm$ that are not contained in any other cut, that is, if $\tm = \ctxp{\cuta\val\var\tmtwo}$ then $\cuta\val\var\tmtwo$ is an out cut of $\tm$ if $\ctx$ cannot be written as $\ctx=\ctxtwop{\cuta\vctx\vartwo\tmthree}$ for some $\ctxtwo$, $\vctx$, $\vartwo$, and $\tmthree$. The set of \emph{out variables} $\ov\tm$ of $\tm$ contains the variables having at least one occurrence out of all cuts of $\tm$.
\end{definition}

Out variables allow us to define a lax notion of cut-free terms where there are some cuts, but they can only be garbage collection cuts, or concern variable occurrences contained in garbage, so that the term becomes cut-free after garbage collection.
\begin{definition}[Cut-freeness up to garbage]
\label{def:cut-free-upto-gc}
A term $\tm$ is cut-free up to garbage if $\var\notin\ov\tmtwo$ for all the out cuts $\cuta\val\var\tmtwo$ of $\tm$.
\end{definition}

\section{The Good Strategy}
\label{sect:strategy}
 \begin{figure*}[t]

 \newcases{nullcases}
    {\ }
    {$##$\hfil} {$##$\hfil}
    {\lbrace} {.}
     \arraycolsep=2pt
  \tabcolsep=2pt
\begin{tabular}{c}
  \begin{tabular}{c@{\hspace{.08cm}}|@{\hspace{.08cm}}cccc}
\multicolumn{2}{c}{\textsc{Dominating free variables of contexts}}&
\\[4pt]
{$\begin{array}{rlll}
\dfv\ctxhole & \defeq & \emptyset
\\[3pt]
\dfv{ \la\var\ctx} & \defeq &  
		\dfv\ctx {\setminus} \set{\var}
\\[3pt]
 \dfv{ \bang \ctx}& \defeq & \dfv{  \ctx}
\\[3pt]
 \dfv{\cuta\val\var\ctx} & \defeq &
		\dfv\ctx {\setminus} \set{\var} 
\\[3pt]
 \dfv{\cuta{\vctx}\var\tm} & \defeq & \dfv\vctx
\end{array}$
}
&
{
$\begin{array}{rlll}
  \dfv{\suba\mvar\vctx\var \tm} & \defeq & \set\mvar\cup \dfv\vctx
\\[3pt]

  \dfv{\suba\mvar\val\var \ctx} & \defeq & \begin{nullcases}
		\set\mvar\cup (\dfv\ctx {\setminus} \set{\var}) & \mbox{if $\var {\in} \dfv\ctx$}\\
		\dfv\ctx & \mbox{otherwise.}
		\end{nullcases}
\\[12pt]
 \dfv{\dera\evar\var \ctx} & \defeq & \begin{nullcases}
		\set\evar\cup (\dfv\ctx {\setminus} \set{\var}) & \mbox{if $\var {\in} \dfv\ctx$}\\
		\dfv\ctx & \mbox{otherwise.}
		\end{nullcases}

\end{array}$
}
\end{tabular}
\\[9pt]\hline\\[-6pt]
\begin{tabular}{c@{\hspace{.3cm}}|@{\hspace{.3cm}}c}
$\begin{array}{rcl}
\multicolumn{3}{c}{\textsc{Good value contexts}}
\\[4pt]
 \vgctx & \grameq & \ctxhole \,\mid\,  \la\var\gctx \,\mid\, \bang\gctx
 \end{array}$
&
 $\begin{array}{rcl}
\multicolumn{3}{c}{\textsc{Good contexts}}
\\[4pt]
 \gctx & \grameq & \vgctx \,\mid\,  \suba\mvar\val\var \gctx
 \,\mid\, \suba\mvar\vgctx\var \tm 
\\[3pt]
&&\,\mid\,  
\dera\evar\var\gctx \,\mid\,  
 \cuta\val\var\gctx \mbox{ if }\var\notin\dfv\gctx
 \end{array}$
\end{tabular} 
\end{tabular}

\caption{Definitions for the good strategy: dominating free variables and good contexts.}
\label{fig:strategy}
\end{figure*}
Here, we define the \emph{good cut elimination strategy} $\tog$ for ESC, recalling the discussion from \cite{DBLP:journals/lmcs/Accattoli23} that explains the design of the strategy. 
The strategy is conceived as to have the \emph{sub-term property}, which is crucial for time analyses and which is defined as follows: every duplicated (or erased) sub-term during a cut elimination sequence is a sub-term of the initial term (up to variable renamings). See \cite{DBLP:journals/lmcs/Accattoli23} for extensive discussions about the sub-term property.

\subparagraph{Breaking the Sub-Term Property.} When does the sub-term property not hold? One has to duplicate an exponential value $\exval$ \emph{touched} by previous steps. In our setting, \emph{touched} can mean two things. Either a redex \emph{fully} contained in $\exval$ is reduced, obtaining $\exvaltwo$, and then $\exvaltwo$ is duplicated (or erased), as in the step marked with $\bigstar$ in the following diagram (the other, dashed path of which has the sub-term property):
\begin{center}
\begin{tikzpicture}[ocenter]
		\node at (0,0)[align = center](source){\normalsize$  \cuta{\exval} \evar \ctxfp{\evar}$};
		\node at (source.center)[right = 170pt](source-right){\normalsize $\cuta{\exvaltwo} \evar \ctxfp{\evar}$};
		\node at (source.center)[below = 25pt](source-down){\normalsize $\cuta{\exval}\evar\ctxfp{\exval}$};
		
		\node at (source-right|-source-down)(target){\normalsize $\cuta{\exvaltwo}\evar\ctxfp{\exvaltwo}$};
		\node at \med{source-down.east}{target.west}(fifthnode){\normalsize $\cuta{\exvaltwo} \evar \ctxfp{\exval}$};
		
		\draw[->](source) to node[above] {\scriptsize $\escsym $} (source-right);
		\draw[->](source-right) to node[right] {\scriptsize $\axeone$} node[right=20pt] {\scriptsize $\bigstar$}(target);

		\draw[->, dotted](source) to node[left] {\scriptsize $\axeone$}(source-down);
		\draw[->, dotted](source-down) to node[above] {\scriptsize $\escsym$} (fifthnode);
		\draw[->, dotted](fifthnode) to node[above] {\scriptsize $\escsym$} (target);
	\end{tikzpicture}
\end{center}
Preventing these situations from happening, thus forcing evaluation to follow the other (dashed) side of the diagram, is easy. It is enough to forbid the position of the reduced redex to be inside the left sub-term of a cut---we say inside a \emph{cut value} for short. It is however not enough, because cuts are also \emph{created}. Consider:
\begin{center}$
\begin{array}{clllllcl}
\cuta{\la\evar\evar}\mvar\suba\mvar\exval\var\tm &\toesc&
\cuta{\la\evar\evar}\mvar\suba\mvar\exvaltwo\var\tm &\tololli
\\
&&\cuta\exvaltwo\evar\cuta\evar\var\tm &\overset\bigstar\rightarrow_{\axeone} &
\cuta\exvaltwo\evar\cuta\exvaltwo\var\tm
\end{array}
$\end{center}
Reducing inside the subtraction value $\exval$ leads to a \emph{later} breaking of the sub-term property by the $\axeone$ step, because the $\tololli$ step creates a cut with $\exvaltwo$ inside.
Preventing these cases is tricky, because forbidding reducing subtraction values leads to cut elimination stopping too soon, without producing a cut-free term. In the $\l$-calculus, it corresponds to forbidding reducing inside arguments, which leads to \emph{head} reduction, that does not compute normal $\l$-terms. We shall then forbid reducing only subtraction values which are \emph{at risk} of becoming cuts. In $\l$-calculus, leftmost reduction does reduce arguments but only when the left sub-term of the application is normal and not an abstraction, so that the argument is not involved in a $\beta$-redex. We shall do something similar here, but  the sequent calculus formalizes this constraint differently, by checking that some variables are not captured.

\subparagraph{Dominating Variables.} The key notion is the one of \emph{dominating (free) variables} $\dfv\ctx$ of a context (where $\ctx$ is meant to be the position of a redex), defined in \reffig{strategy}, the base case of which is for $\suba\mvar\vctx\var\tm$. If $\ctx$ is a position and $\var\in\dfv\ctx$ then $\cuta\val\var\ctx$ turns $\ctx$ into a dangerous position, that is, a redex of position $\cuta\val\var\ctx:\tm\toms\tmtwo$ might lead to a breaking of the sub-term property later on during cut elimination. In the example, $\evar$ belongs to $\dfv\ctx$ for every context $\ctx\defeq \dera\evar\mvar\suba\mvar\vctx\var\tmtwo$ of $\dera\evar\mvar\suba\mvar\val\var\tmtwo$, for every $\vctx$.

\subparagraph{The Good Strategy.} These considerations lead to the notion of \emph{good contexts} in \reffig{strategy}. A good context forbids the two ways of breaking the sub-term property: its hole cannot be in a cut value (note the absence of the production $\cuta{\vgctx}\var\tm$) nor in a subtraction value such that one of its dominating variables is cut (because of the production $\cuta\val\var\gctx$ if $\var\notin\dfv\gctx
$). 

\begin{definition}[Good steps, good strategy]
A step $\ctx : \tm \toms \tmtwo$ is \emph{good} if its position $\ctx$ is good. In such a case, we write $\tm \tog \tmtwo$. The \emph{good cut elimination strategy} is simply $\tog$. We also use $\togp a$ to stress that the good step is of kind $a\in\set{\axmone,\axmtwo,\lolli, \axeone,\axetwo,\bang,\wsym}$. The non-erasing good strategy $\tognotw$ is the variant of $\tog$ excluding $\togp\wsym$ steps.
\end{definition}

The sub-term property in the following theorem captures the quantitative aspect for cost analyses, {i.e.} the bound on the size of duplicated values by the size of the initial term.

\begin{theorem}[Properties of the good strategy, \cite{DBLP:journals/lmcs/Accattoli23}]
\label{thm:sub-term-prop}
\hfill
\begin{enumerate}
\item \emph{Quantitative sub-term property}: if $\deriv:\tm \tog^{*} \tmtwo$ and $\val$ be a value erased or duplicated along $\deriv$, then $\size\val\leq \size\tm$.
\item \emph{Diamond}: if $ \tmtwo_{1} \lRew{\Gsym}\tm \tog \tmtwo_{2}$ and $\tmtwo_{1} \neq \tmtwo_{2}$ then $ \tmtwo_{1} \tog \tmthree \lRew{\Gsym}  \tmtwo_{2}$ for some $\tmthree$.
\item \emph{Fullness}: if $\tm$ is clash-free and not $\toms$-normal then $\tm \tog\tmtwo$ for some $\tmtwo$.
\item \emph{Good polynomial cost model}: if $\multiform\vdash \tm\hastype \form$ is a typed term then there exist  $k$ and a cut-free term $\tmtwo$ such that $\tm \tog^{k} \tmtwo$, and such a reduction sequence is implementable on RAMs in time polynomial in $k$ and $\size\tm$.
\end{enumerate}
\end{theorem}
The algorithmic aspect of \refthm{sub-term-prop}.4 is proved in  \cite{DBLP:journals/lmcs/Accattoli23} via arguments that do not establish the degree of the polynomial bound. The aim of this paper is exactly to show that a carefully designed abstract machine provides a bound that is linear in both $k$ and $\size\tm$.

\subparagraph{The Diamond Property.} Let us provide some background about the diamond property.
Following Dal Lago and Martini \cite{DBLP:journals/tcs/LagoM08}, we say that a relation $\Rew{\Rule}$ is \emph{diamond} if $\tmtwo_1 \lRew{\Rule} \tm \Rew{\Rule} \tmtwo_2$ and $\tmtwo_1 \neq \tmtwo_2$ imply $\tmtwo_1 \Rew{\Rule} \tmthree \lRew{\Rule} \tmtwo_2$ for some $\tmthree$. 
	The terminology in the literature is inconsistent: Terese \cite[Exercise 1.3.18]{Terese} dubs this property $\textup{CR}^1$, and defines the diamond more restrictively, without requiring $\tmtwo_1 \neq \tmtwo_2$ in the hypothesis: $\tmtwo_1$ and $\tmtwo_2$ have to join even if $\tmtwo_1 = \tmtwo_2$.

Standard corollaries of Dal Lago and Martini's notion are that, if $\Rew{\Rule}$ is diamond, then:
\begin{enumerate}
	\item \emph{Confluence}:
	$\Rew{\Rule}$ is confluent, that is, $\tmtwo_1  \lRewn{\Rule} \tm \Rew{\Rule}^* \tmtwo_2$  implies $\tmtwo_1 \Rew{\Rule}^* \tmthree \lRewn{\Rule} \tmtwo_2$ for some $\tmthree$; 
	\item \emph{Length invariance}: all $\Rule$-evaluations with the same start and $\Rule$-normal end terms have the same length (\ie if $\deriv \colon \tm \Rew{\Rule}^k \tmtwo$  and $\deriv' \colon \tm \Rew{\Rule}^h \tmtwo$ with $\tmtwo$ $\Rule$-normal, then $h = k$);
	\item \emph{Uniform normalization}: $\tm$ is weakly $\Rule$-normalizing if and only if it is strongly $\Rule$-normalizing.
\end{enumerate}
Basically, the diamond property captures a more liberal form of determinism. In particular, length invariance is essential in order to take the number of steps of a strategy as a cost model. Without it, indeed, the number of steps of a non-deterministic strategy would be an ambiguously defined notion of cost.

\section{Preliminaries on Abstract Machines}
\label{sect:prel-abs-mach}
\subparagraph{Abstract Machines Glossary.}  Abstract machines manipulate \emph{pre-terms}, that is, terms without implicit $\alpha$-renaming. In this paper, an \emph{abstract 
machine} is a quadruple $\mach = (\States, \tomach, \compilrel\cdot\cdot, \decode\cdot)$ the components of which are as follows.
\begin{itemize}

\item \emph{States.} A state $\state\in\States$ is composed by the \emph{active term} $\tm$, plus one data structure which depends on the actual machine. Terms in states are actually pre-terms.

\item  \emph{Transitions.} The pair $(\States, \tomach)$ is a transition 
system with transitions $\tomach$ partitioned into \emph{principal transitions}, whose union is noted $\tomachpr$ and that are meant to correspond to cut-elimination steps on the calculus, and \emph{search transitions}, whose union is noted $\tomachsea$, that take care of searching for (principal) redexes.

\item \emph{Initialization.} The component $\compilrel{}{}\subseteq\Lambda\times\States$ is the \emph{initialization relation} associating  terms to 
initial states. It is a \emph{relation} and not a function because $\compilrel\tm\state$ maps a $\l$-term $\tm$ (considered modulo $\alpha$) to a state $\state$ having a \emph{pre-term representant} of $\tm$ (which is not modulo $\alpha$) as active term. Intuitively, any two states $\state$ and $\statetwo$ such that $\compilrel\tm\state$ and $\compilrel\tm\statetwo$ are $\alpha$-equivalent. 
The initializing terms ({i.e.} those $\tm$ such that $\compilrel\tm\state$ for some $\state$) are always  \emph{proper} and \emph{clash-free}.
A state $\state$ is \emph{reachable} if it can be reached starting from an initial state, that is, if $\statetwo \tomach^*\state$ where $\compilrel\tm\statetwo$ for some $\tm$ and $\statetwo$, shortened as $\compilrel\tm\statetwo \tomach^*\state$.

\item \emph{Read-back.} The read-back function $\decode\cdot:\States\to\Lambda$ turns reachable states into 
terms and satisfies the \emph{initialization constraint}: if $\compilrel\tm\state$ then $\decode{\state}=_\alpha\tm$.
\end{itemize}

\subparagraph{Further Terminology and Notations.} A state is \emph{final} if no transitions apply.
 A \emph{run} $\run: \state \tomach^*\statetwo$ is a possibly empty finite sequence of transitions, the length of which is noted 
$\size\run$; note that the first and the last states of a run are not necessarily initial and final. 
If $a$ and $b$ are transitions labels (that is, $\tomachhole{a}\subseteq \tomach$ and 
$\tomachhole{b}\subseteq \tomach$) then $\tomachhole{a,b} \defeq \tomachhole{a}\cup \tomachhole{b}$ and $\sizep\run a$ 
is the number of $a$ transitions in $\run$, and $\sizep\run{\neg a}$ is the number of transitions in $\run$ that are 
not $\tomachhole{a}$.

\subparagraph{Well-Boundness and Renamings.} For the machines at work in this paper, the pre-terms in initial states shall be \emph{well-bound}, that is, they have pairwise distinct bound names; for instance $\suba\mvar{\la\mvartwo\mvartwo}\evar\la\evartwo\evartwo$ is well-bound while $\suba\mvar{\la\mvartwo\mvartwo}\evar\la\mvartwo\mvartwo$ is not. 
We shall also write $\rename{\tm}$ in a state $\state$ for a \emph{fresh well-bound renaming} of $\tm$,
\ie $\rename{\tm}$ is $\alpha$-equivalent to $\tm$, well-bound, and its bound variables
are fresh with respect to those in $\tm$ and in the other components of $\state$.

\subparagraph{Implementation Theorem, Abstractly.} We now formally define the notion of a machine implementing a strategy. Since the good strategy is non-deterministically diamond but the machine that shall implement it is deterministic, we need a slightly unusual form of implementation theorem. As it is standard, machine transitions shall be mapped to equalities or steps on the calculus. For the other direction, however, we obtain only a \emph{big-step} simulation, that is, if the strategy on the calculus terminates / diverges then the same does the machine, and with a related number of steps. But there is no step-by-step simulation of the calculus by the machine, because the non-deterministic strategy might do a step which is not the one done by the machine. Everything works fine because of the properties of diamond strategies. The general scheme is inspired by Accattoli et al.'s scheme for strong call-by-value \cite{DBLP:conf/lics/AccattoliCC21}, but it is here simpler because of the absence of structural equivalence and implosive sharing.

It would also be possible to have a more symmetric setting by either designing a more complex diamond machine, as done by Accattoli and Barenbaum \cite{DBLP:conf/aplas/AccattoliB23}, or by designing a deterministic variant of the good strategy. For the sake of simplicity, we prefer the asymmetric setting as to keep a simple machine and reuse the notion of good strategy from the literature.

\begin{definition}[Big-step implementations]
A machine $\mach = (\States, \tomach, \compilrel\cdot\cdot, \decode\cdot)$ is a \emph{big-step implementation of a strategy} $\tostrat$ on terms when, given a (proper and clash-free) term $\tm$:\label{def:implem}
\begin{enumerate}
\item \emph{Runs to evaluations}: for any $\mach$-run $\run: \compilrel\tm\statetwo \tomachine^* \state$ there is a $\tostrat$-evaluation $\deriv: \tm \tostrat^* \decode\state$ with $\sizepr\run = \size\deriv$.

\item \emph{Normalizing evaluations to runs}: if $\deriv \colon \tm \tostrat^* \tmtwo$ with $\tmtwo$  $\tostrat$-normal then there is a $\mach$-run $\run \colon \compilrel\tm\statetwo \tomachine^* \state$ such that $\decode\state = \tmtwo$ with $\sizepr\run = \size{\deriv}$. 

\item \emph{Diverging evaluations to runs}: if $\compilrel\tm\state$ and $\tostrat$ diverges on $\tm$ then $\mach$ diverges on $\state$ doing infinitely many principal transitions.
\end{enumerate}
\end{definition}

Next, we isolate sufficient conditions for big-step implementations.

\begin{definition}[Lax implementation system]
  A lax implementation system is given by a \emph{machine} $\mach = (\States, \tomach, \compilrel\cdot\cdot, \decode\cdot)$ and a strategy $\tostrat$ such that for every reachable state $\state$:$\label{def:implementation}$
  \begin{enumerate}
		\item\label{p:def-beta-projection} \emph{Principal projection}: if $\state \tomachpr \statetwo$ then $\decode\state \tostrat \decode\statetwo$;
		\item\label{p:def-overhead-transparency} \emph{Search transparency}: if $\state \tomachsea \statetwo$ then $\decode\state  =\decode\statetwo$;
		\item\label{p:def-overhead-terminate}	\emph{Search terminates}:  $\tomachsea$ terminates;

	\item\label{p:def-progress} \emph{Halt}: if $\state$ is final then $\decode\state$ is $\tostrat$-normal;
	
	\item\label{p:def-determinism} \emph{Diamond}:  $\tostrat$ is diamond.
  \end{enumerate}
\end{definition}

\begin{toappendix}
\begin{theorem}[Abstract big-step implementation]
\NoteProof{thm:abs-impl}
  Let $\mach$ and  $\tostrat$ form a lax implementation system.
  Then, $\mach$ is a big-step implementation of $\tostrat$. \label{thm:abs-impl}
\end{theorem}
\end{toappendix}

\subparagraph{Clash-Free and Proper States.} 
Note that we have \emph{not} taken into account properness and clashes on states. One can say that a state $\state$ is proper (resp. clash-free) if its decoding $\decode\state$ is. In this way, these notions are trivially seen to be preserved by runs in a lax implementation system. Note, indeed, that transitions are mapped via decoding to either equalities or rewriting steps (and the proofs of these facts shall not need properness nor clash-freeness). Since states are initialized with  proper and clash-free  terms, and that these notions on terms are preserved by reduction (clash-freeness by definition, for properness see \cite{DBLP:journals/lmcs/Accattoli23}), proper and clash-free states are preserved by transition. Therefore, we shall omit all considerations about properness and clashes for machines.

\section{A Machine for the Closed Basic Case}
\label{sect:BAM}
\begin{figure}[t!]
\centering
\small
\SetTblrInner{rowsep=0pt}
  \begin{tblr}{c}
 $ \begin{array}{r@{\hspace{.4cm}} rcl@{\hspace{.4cm}}|@{\hspace{.4cm}} r@{\hspace{.4cm}} rclll}
%
  
    \textsc{States} & \state,\statetwo & \grameq & \twostate\ectx\tm
    &
  \textsc{Initialization} & \compilrel\tm\state & \mbox{if} & \state=\twostate\ctxhole{\rename\tm}
  \\
  \textsc{Read back} & \decode{\twostate\ectx\tm} & \defeq & \ectx\ctxholep\tm
    \end{array}$
    \\[5pt]
$\begin{tblr}{|r|r||c||r|r|l}
\cline{1-5}
\SetCell[c=1]{c} \textsc{Cut Ctx} & \SetCell[c=1]{c} \textsc{Active Tm} &\textsc{Tran.}& \SetCell[c=1]{c} \textsc{Cut Ctx} & \SetCell[c=1]{c}\textsc{Active Tm}
\\\cline{1-5}
\ectx & \cuta\val\var \tm
&\tomachsea&
\ectxp{\cuta\val\var\ctxhole} &\tm
\\
\cline[dashed]{1-5}
\ectxp{   \cuta{\mvartwo}\mvar   \ectxtwo} & \suba\mvar\val\var\tm
&\tomachaxmtwo&
\ectxp   \ectxtwo & \suba\mvartwo\val\var\tm
\\

\ectxp{   \cuta{\la\vartwo\lctxp\valtwo}\mvar   \ectxtwo} & \suba\mvar\val\var \tm
&\tomachlolli&
\ectxp{  \ectxtwo  \cuta\val\vartwo } &\lctxp{\cuta\valtwo\var\tm}
\\
\ectxp{   \cuta{\evartwo}\evar   \ectxtwo} & \dera\evar\var\tm
&\tomachaxetwo&
\ectxp{   \cuta{\evartwo}\evar   \ectxtwo} & \dera\evartwo\var\tm
\\
\ectxp{   \cuta{\bang\lctxp\val}\evar   \ectxtwo} & \dera\evar\var\tm
&\tomachbang&
\ectxp{   \cuta{\bang\lctxp\val}\evar   \ectxtwo} & \lctxtwo\ctxholep{\cuta{\valtwo}\var\tm}
&\#
\\\cline[dashed]{1-5}

\ectxp{   \cuta{\mval}\mvar   \ectxtwo} & \mvar
&\tomachaxmone&
\ectxp   \ectxtwo &\mval
\\

\ectxp{   \cuta{\exval}\evar   \ectxtwo} & \evar
&\tomachaxeone&
\ectxp{   \cuta{\exval}\evar   \ectxtwo} &\rename\exval
\\\cline{1-5}
\end{tblr}$
\\[4pt]
\SetCell[c=1]{l}\# with $\lctxtwop\valtwo=\rename{\lctxp\val}$
\end{tblr}
\caption{The Basic Abstract Machine (BAM).}
\label{fig:BAM}
\end{figure}
Here, we study a  machine implementing basic evaluation on closed ESC terms. The aim is to give a gentle introduction to some machine concepts.

\subparagraph{BAM.} The \emph{basic abstract machine} (BAM) is defined in \reffig{BAM}. States $\state = \twostate\ectx\tm$ are simply given by the active (pre-)term $\tm$ and a cut context $\ectx$ which is a list containing the cuts encountered so far by the search mechanism, playing the role of the (global) environment in machines such as Accattoli et al.'s Milner Abstract Machine \cite{DBLP:conf/icfp/AccattoliBM14}. 

The initialization relation $\compilrel\tm{\twostate\ctxhole {\rename\tm}}$ pairs (proper and clash-free) terms $\tm$ with states composed by an empty cut context $\ctxhole$ and a well-bound renaming $\rename\tm$ of $\tm$. The BAM has six principal transitions, mimicking the ESC rewriting rules but for the weakening one (BAM never erases), plus one search transition $\tomachsea$, moving cuts from the term to the cut context.

The BAM looks at the topmost constructor of the active term and proceeds applying a transition belonging to one of the following three groups. If the constructor is:
\begin{itemize}
\item \emph{Search}: a cut $\cuta\val\var$, the BAM adds it to the cut context, by applying $\tomachsea$;
\item \emph{Computation}: a subtraction $\suba\mvar\val\var$ or a dereliction $\dera\evar\var$, then the BAM looks for the associated cut in the cut context, it applies the corresponding cut elimination rule, and goes on to execute the modified active term, by applying $\tomachaxmtwo$, $\tomachlolli$, $\tomachaxetwo$, or $\tomachbang$;
\item \emph{Terminal replacements}: a variable $\mvar$ or $\evar$, the BAM looks for the associated cut in the cut context and it applies the matching replacement via $\tomachaxmone$ or $\tomachaxeone$. Note that terminal replacements can only be followed by further terminal replacements, hence the name.
\end{itemize}
Note that in transitions $\tomachbang$ and $\tomachaxeone$ some renaming takes place, using names that are fresh with respect to the whole state.
The domain $\domp\ectx$ of a cut context is defined as $\domp\ctxhole  \defeq  \emptyset$ and $\domp{\cuta\val\var\ectx} \defeq  \set\var\cup\domp\ectx$.

\begin{toappendix}
\begin{lemma}[BAM qualitative invariants]
Let $\state=\twostate\ectx\tm$ be a BAM  reachable state.\label{l:bam-invariants} %
\NoteProof{l:bam-invariants}
\begin{enumerate}
\item 
\emph{Closure}: $\fv\tm \subseteq \domp\ectx$ and if $\ectx=\ectxtwop{\cuta\val\var\ectxthree}$ then $\fv\val\subseteq \domp\ectxtwo$.

\item 
\emph{Well-bound}: if $\la\var\tmtwo$, $\dera\evar\var\tmtwo$, $\suba\mvar\val\var\tmtwo$, or $\cuta\val\var\tmtwo$ occur in $\state$ and $\var$ has any other occurrence in $\state$ then it is a free variable of $\tm$, and if $\ectx=\ectxtwop{\cuta\val\var\ectxthree}$ and $\var$ has any other occurrence in $\state$ then it is a free variable in $\ectxthree$ or $\tm$. 
\end{enumerate}
\end{lemma}
\end{toappendix}

\begin{toappendix}
\begin{proposition}
\label{prop:bam-properties}
\NoteProof{prop:bam-properties}
Let $\state$ be a BAM  reachable state and $a\in\set{\axmone,\axmtwo, \axeone,\axetwo,\lolli,\bang}$.
\begin{enumerate}
\item \emph{Search transparency}: if $\state \tomachsea \statetwo$ then $\decode\state=\decode\statetwo$.
\item \emph{Principal projection}: if $\state \tomachhole{a} \statetwo$ then $\decode\state \Rew{\symfont{b}a} \decode\statetwo$.
\item \emph{Search termination}: transition $\tomachsea$ terminates.
\item \emph{Halt}: if $\state$ is final then $\state = \twostate\ectx\val$ with $\val$ not a variable, and $\decode\state$ is normal for non-erasing basic evaluation.
\end{enumerate}
\end{proposition}
\end{toappendix}

\begin{theorem}
The BAM implements ESC non-erasing basic evaluation on closed terms.
\end{theorem}
\section{SESAME}
\label{sect:sesame}

In this section, we extend the BAM as to perform the analogous of strong evaluation, that is, as to perform cut-elimination also inside values, and implement the good strategy. The obtained machine shall ignore garbage collection and---when it terminates---it returns a term that is cut-free \emph{up to garbage} (\refdef{cut-free-upto-gc}). Garbage collection shall be addressed in \refsect{gc-and-full}.
\begin{figure}[t!]
\centering
\footnotesize
\SetTblrInner{rowsep=0pt}
  \begin{tblr}{c}
  $\begin{tblr}{rrcl|rrclll}
%
  
    \textsc{Pools} & \pool,\pooltwo & \grameq & \emptylist \mid \namctxholep\tm\name \cons \pool
    &
  \textsc{Init.} & \compilrel\tm\state & \mbox{if} & \state=\twostate\namctxholea{\rename\tm}
    \\
    \textsc{States} & \state,\statetwo & \grameq & \twostate\namctx\pool
  &
  \textsc{Read back} & \decode{\twostate\namctx\pool} & \defeq & \namctx\namctxholep{\tm_1}{\name_1}\ldots\namctxholep{\tm_k}{\name_k}
  \\
  &&&&& \SetCell[c=3]{r} \mbox{with }\pool=\namctxholep{\tm_1}{\name_1}\cons\ldots\cons\namctxholep{\tm_k}{\name_k}\cons\emptylist

    \end{tblr}$
    \\[4pt]

$\begin{tblr}{|r|r||c||r|r|l}
\cline{1-5}
\SetCell[c=1]{c} \textsc{MCtx} & \SetCell[c=1]{c} \textsc{Pool} &\textsc{Tran.}& \SetCell[c=1]{c} \textsc{MCtx} & \SetCell[c=1]{c}\textsc{Pool}
\\\cline{1-5}
\namctx & \nctxholep{\cuta\val\var \tm}\name \cons \pool 
&\tomachseaone&
\namctx  \nctxholep{\cuta\val\var \nctxhole\name}\name &\nctxholep{\tm}\name \cons \pool
\\[2pt]
\namctx \mbox{ with }\mvar\notin\domp\namctx& \nctxholep{\suba\mvar\val\var \tm}\name \cons \pool  
&\tomachseatwo&
\namctx  \namctxholep{\suba\mvar{\nctxhole\nametwo}\var \nctxhole\name}\name &\nctxholep{\val}\nametwo \cons \nctxholep{\tm}\name  \cons \pool & *

\\[2pt]
\namctx \mbox{ with }\evar\notin\domp\namctx& \nctxholep{\dera\evar\var\tm}\name \cons \pool
&\tomachseathree&
\namctx   \nctxholep{\dera\evar\var \nctxhole\name}\name &\nctxholep{\tm}\name \cons \pool
\\[2pt]
\namctx    & \nctxholep{\la\var\tm}\name \cons \pool  
&\tomachseafour&
\namctx  \nctxholep{ \la\var\nctxhole\name}\name &\nctxholep{\tm}\name \cons \pool
\\[2pt]
\namctx    & \nctxholep{\bang\tm}\name \cons \pool  
&\tomachseafive&
\namctx  \nctxholep{ \bang\nctxhole\name}\name &\nctxholep{\tm}\name \cons \pool
\\[2pt]
\namctx \mbox{ with }\var\notin\domp\namctx  & \nctxholep{\var}\name  \cons \pool  
&\tomachseasix&
\namctx  \nctxholep\var\name &\pool

\\[4pt]
\cline[dashed]{1-5}

\namctxp{\cuta\mvartwo\mvar\namctxtwo}\mute & \nctxholep{\suba\mvar\val\var\tm}\name \cons \pool    
&\tomachaxmtwo&
\namctxp\namctxtwo\mute & \nctxholep{\suba\mvartwo\val\var\tm}\name \cons \pool
\\[2pt]

\namctxp{\cuta{\la\vartwo\lctxp\valtwo}\mvar\namctxtwo}\mute & \nctxholep{\suba\mvar\val\var \tm}\name \cons \pool    
&\tomachlolli&
\namctxp\namctxtwo\mute &\nctxholep{\cuta\val\vartwo\lctxp{\cuta\valtwo\var\tm}}\name \cons \pool
\\[2pt]

\namctxp{\cuta\evartwo\evar\namctxtwo}\mute & \nctxholep{\dera\evar\var\tm}\name \cons \pool    
&\tomachaxetwo&
\namctxp{\cuta\evartwo\evar\namctxtwo}\mute & \nctxholep{\dera\evartwo\var\tm}\name \cons \pool  
\\[2pt]

\namctxp{\cuta{\bang\lctxp\val}\evar\namctxtwo}\mute & \nctxholep{\dera\evar\var\tm}\name \cons \pool    
&\tomachbang&
\namctxp{\cuta{\bang\lctxp\val}\evar\namctxtwo}\mute & \nctxholep{\lctxtwop{\cuta{\valtwo}\var\tm}}\name \cons \pool
&\#
\\[2pt]

\namctxp{\cuta\mval\mvar\namctxtwo}\mute & \nctxholep{\mvar}\name \cons \pool  
&\tomachaxmone&
\namctxp\namctxtwo\mute &\nctxholep{\mval}\name \cons \pool
\\[2pt]

\namctxp{\cuta\exval\evar\namctxtwo}\mute & \nctxholep{\evar}\name \cons \pool  
&\tomachaxeone&
\namctxp{\cuta\exval\evar\namctxtwo}\mute & \nctxholep{\rename\exval}\name \cons \pool
\\\cline{1-5}
\end{tblr}$
\\[4pt]
\SetCell[c=1]{l}* $\nametwo$ is fresh. \ \ \ \# with $\lctxtwop\valtwo=\rename{\lctxp\val}$
\end{tblr}
\caption{The Strong Exponential Substitution Abstract Machine without Erasure (SESAME).}
\label{fig:SESAM}
\end{figure}

\subparagraph{Pools and Jobs.} The \emph{Strong Exponential Substitution Abstract Machine without Erasure} (SESAME), defined in \reffig{SESAM}, relies on a technique for strong evaluation recently introduced by Accattoli and Barenbaum \cite{DBLP:conf/aplas/AccattoliB23}, which we are now going to explain. 

When evaluation goes under binders, the closure invariant of the basic case (\reflemma{bam-invariants}.1), for which free variables of the active term are associated to a cut in the cut context, is lost. Thus, when the active term is a subtraction $\suba\mvar\val\var\tm$ with no associated cut for $\mvar$,  the subtraction is kept, and the machine has to evaluate $\val$ and $\tm$. Now, the evaluations of $\val$ and $\tm$ cannot affect each other, they are independent, but one of the two sub-terms has to be evaluated first, say $\val$. Usually, strong machines (such as Cr\'egut's \cite{DBLP:journals/lisp/Cregut07}) would run through $\val$, and if such process does terminate, producing a value $\valtwo$, then they \emph{backtrack} to the subtraction generating the fork by \emph{moving sequentially} through $\valtwo$, and then start evaluate $\tm$.

Accattoli and Barenbaum's technique simply circumvents the sequential backtracking process by directly jumping back to the forking point. For that, the machine is equipped with a \emph{pool} $\pool$ of \emph{jobs}, each one paired to a unique name $\name$. In our example, there would be a job $\namctxholep\val\name$ for $\val$ and a job $\namctxholep\tm\nametwo$ for $\tm$. The idea is that when one job ends then the machine jumps to the next one, \emph{without} moving through the structure of the finished job. 

In \cite{DBLP:conf/aplas/AccattoliB23}, the pool is any data structure satisfying a certain interface. The idea is that different data structures implementing the interface realize different job scheduling policies, compactly accounting for different strong strategies within the same framework. We here omit this abstract aspect and fix the pool to be the simplest such structure, namely a LIFO list of named jobs. In the $\lambda$-calculus, LIFO list pools implement leftmost evaluation \cite{DBLP:conf/aplas/AccattoliB23}.

Another simplification with respect to \cite{DBLP:conf/aplas/AccattoliB23} is that here jobs are simply named terms, while in \cite{DBLP:conf/aplas/AccattoliB23} they are named pairs of a term and an applicative stack. The change induces a simpler notion of read back. This difference is not a design choice, it is simply induced by having a machine for sequent calculus terms (here) rather than natural deduction terms (in \cite{DBLP:conf/aplas/AccattoliB23}).

A crucial point of SESAME is that job forking happens on subtractions only, that is, it does not happen on cuts $\cuta\val\var\tm$: the machine ignores $\val$ and goes straight to evaluate $\tm$. This happens in particular to prevent the breaking of the sub-term property. 

We write $\fn\pool$ for the set of names associated to the jobs in $\pool$, that is, if $\pool = \namctxholep{\tm_1}{\name_1}\cons\ldots\cons \namctxholep{\tm_k}{\name_k}$ then $\fn\pool \defeq \set{\name_1,\ldots,\name_k}$.

\subparagraph{Multi-Contexts and Approximants.} Another aspect of the technique in \cite{DBLP:conf/aplas/AccattoliB23} is that the parts of the term that have been evaluated and that shall not be touched again by the machine---sometimes referred to as \emph{stable parts}---are accumulated in the \emph{approximant} $\appr$ (of the normal form), which is a multi-context (that is, a context with possibly many holes, possibly none). The idea is that the name $\name$ of a job $\namctxholep{\tm}{\name}$ in the pool is associated to a (unique) named hole $\namctxholep{\cdot}{\name}$ in $\appr$, and that, whenever a stable piece of term is produced by the job $\namctxholep{\tm}{\name}$, that piece is moved to $\namctxholep{\cdot}{\name}$ in $\appr$, incrementally building the normal form.

We follow this pattern from \cite{DBLP:conf/aplas/AccattoliB23}, but our setting induces a few differences. Firstly, in \cite{DBLP:conf/aplas/AccattoliB23} the machine has also a global environment, akin to the cut context of the BAM. Here, we include the environment/cut context \emph{into} the approximant, having only one data structure. Thus, the approximant shall have cuts, but these cuts bind variables that have non-garbage occurrences only in the active term, i.e., only out of the approximant itself, where instead they have only garbage occurrences. At the end of a complete run, the approximant shall be normal for the non-erasing good strategy, i.e. it shall be cut-free up to garbage (\refdef{cut-free-upto-gc}).

Secondly, in \cite{DBLP:conf/aplas/AccattoliB23} the machine is defined using approximants, while here we define it using the weaker notion of multi-contexts, and then prove invariants guaranteeing that the multi-contexts of reachable states are approximants.

\begin{definition}[Multi-contexts and approximants]
A \emph{(named) multi-context} $\namctx$ is an ESC term in which there might be
occurrences of holes $\namctxholea$ indexed with names, as follows:
\begin{center}$     \arraycolsep=2pt
\begin{array}{r@{\hspace{.3cm}} l@{\hspace{.2cm}}|@{\hspace{.2cm}}r@{\hspace{.3cm}} lllll}
\textsc{Names} & \name,\nametwo,\namethree, \name_1, \nametwo_2, \dots
&
\textsc{Value multi-ctxs} & \namvctx & \grameq & \namctxholea \mid \var \mid \la\var\namctx \mid \bang\namctx 
\\
&&
\textsc{Multi-ctxs} & \namctx & \grameq & \namvctx \mid \cuta\namvctx\var \namctx  \mid \suba\mvar\namvctx\var\namctx \mid \dera\evar\var\namctx
\end{array}$
  \end{center}
The plugging $\namctxp{\namctxtwo}{\name}$ of $\namctxtwo$ on $\name$ in $\namctx$,
is the capture-allowing substitution of $\namctxholea$ by $\namctxtwo$ in $\namctx$. We write $\namctxp{\namctxtwo}{\mute}$ for when $\mute$ is an irrelevant name occurring exactly once in $\namctx$; this notation is meant to be used for decomposing a multi-context in two, as in $\namctx=\namctxtwop{\namctxthree}{\mute}$. 
We write $\fn{\namctx}$ for the set of names that occur in $\namctx$. Out variables and out cuts extend to multi-contexts as expected.

The domain $\domp\namctx=\set{\var_1,\ldots,\var_n}$ of a multi-context $\namctx$ contains the variables on which there is a cut in $\namctx$ (the formal definition is in \techrep{\refapp{sesame}, page \pageref{def:namctx-domain}}\camerar{Appendix E of the tech report \cite{accattoli2024imell}}).

An approximant $\appr$ is a multi-context such that: 
\begin{enumerate}
\item \emph{Unique names}: every name $\name\in\fn\appr$ has exactly one occurrence in $\appr$;
\item \emph{Out cuts are hereditary garbage and hole-free}: for every out cut $\cuta\namvctx\var\namctxtwo$ in $\appr$, if $\var\in\fv\namctxtwo$ then $\var$ occurs only inside cut values, i.e. $\var\notin\ov\namctxtwo$, and $\namvctx$ is a term.
\end{enumerate}
\end{definition}

Note that a multi-context $\namctx$ without holes is simply a term, thus the defined notion of plugging subsumes the plugging $\namctxp{\tm}{\name}$ of terms in multi-contexts.

\subparagraph{Transitions.} The principal transition of SESAME are as for the BAM up to the generalization of the cut context $\ectx$ to a multi-context $\namctx$. Note the mute name $\mute$ for decomposing the multi-context in the transitions. Clearly, the name $\mute$ on the LHS and the RHS of each transition is the same. SESAME has a search transition for each constructor of the calculus. Note that transitions $\tomachaxmone$ and $\tomachaxeone$ are no longer \emph{terminal}, because now they might be followed by a search transition, after which the run might jump to a different job.

\subparagraph{Read Back.} The read back $\decode\state$ of a state $\state = \twostate\namctx\pool$ is defined in \reffig{SESAM} and simply plugs each job $\namctxholep\tm\name$ of the pool $\pool$ in the hole of name $\name$ of the multi context $\namctx$. We also write $\namctx\ctxholep\pool$ for $\decode\state$. Since the holes of $\namctx$ are all independent and jobs contains terms (with no holes), we have that $\namctx\ctxholep\pool = \namctx\ctxholep\pooltwo$ for any pool $\pooltwo$ obtained by permuting the elements of $\pool$.

\subparagraph{Example of SESAME run.} As an example, we show a SESAME run on the following term:
\begin{center}
$\begin{array}{cccc}
\tm &\defeq& {\cuta{\bang{\la{\mvar_1}{\mvar_1}}}{\evar_1}\dera{\evar_1}{\mvar_2}\dera{\evar_1}{\mvar_3}\suba{\mvar_2}{\mvar_3}{\mvar_4}{\mvar_4}}
\end{array}$ 
\end{center}
that just applies an identity function to itself, by making two copies of $\bang{\la{\mvar_1}{\mvar_1}}$ and applying one to the other via the subtraction. To fit it into the margins, we use the abbreviation $\tmtwo \defeq \suba{\mvar_2}{\mvar_3}{\mvar_4}{\mvar_4}$.
\begin{center}$
\begin{array}{l|l|l}
  \SetCell[c=1]{c} \textsc{Tr.} & \SetCell[c=1]{c} \textsc{MCtx} & \SetCell[c=1]{c} \textsc{Pool} 
  \\\hline
  & \nctxhole\name & \nctxholep{\tm}\name \,\cons\, \emptylist \\
  \tomachseaone & \cuta{\bang{\la{\mvar_1}{\mvar_1}}}{\evar_1}\nctxhole\name & \nctxholep{\dera{\evar_1}{\mvar_2}\dera{\evar_1}{\mvar_3}\tmtwo}\name \,\cons\, \emptylist \\
  \tomachbang   &\cuta{\bang{\la{\mvar_1}{\mvar_1}}}{\evar_1}
        \nctxhole\name & \nctxholep{\cuta{\la{\mvar_5}{\mvar_5}}{\mvar_2}\dera{\evar_1}{\mvar_3}\tmtwo}\name \,\cons\, \emptylist \\
  \tomachseaone & \cuta{\bang{\la{\mvar_1}{\mvar_1}}}{\evar_1}
        \cuta{\la{\mvar_5}{\mvar_5}}{\mvar_2}\nctxhole\name & \nctxholep{\dera{\evar_1}{\mvar_3}\tmtwo}\name \,\cons\, \emptylist \\
  \tomachbang   & \cuta{\bang{\la{\mvar_1}{\mvar_1}}}{\evar_1}
        \cuta{\la{\mvar_5}{\mvar_5}}{\mvar_2}\nctxhole\name & \nctxholep{\cuta{\la{\mvar_6}{\mvar_6}}{\mvar_3}\tmtwo}\name \,\cons\, \emptylist \\
  \tomachseaone & \cuta{\bang{\la{\mvar_1}{\mvar_1}}}{\evar_1}
        \cuta{\la{\mvar_5}{\mvar_5}}{\mvar_2}\cuta{\la{\mvar_6}{\mvar_6}}{\mvar_3}\nctxhole\name & \nctxholep{\tmtwo}\name \,\cons\, \emptylist \\
  \tomachlolli   & \cuta{\bang{\la{\mvar_1}{\mvar_1}}}{\evar_1}
        \cuta{\la{\mvar_6}{\mvar_6}}{\mvar_3}\nctxhole\name & \nctxholep{\cuta{{\mvar_3}}{\mvar_5}\cuta{{\mvar_5}}{\mvar_4}{\mvar_4}}\name \,\cons\, \emptylist \\
  \tomachseaone & \cuta{\bang{\la{\mvar_1}{\mvar_1}}}{\evar_1}
        \cuta{\la{\mvar_6}{\mvar_6}}{\mvar_3}\cuta{{\mvar_3}}{\mvar_5}\nctxhole\name & \nctxholep{\cuta{{\mvar_5}}{\mvar_4}{\mvar_4}}\name \,\cons\, \emptylist \\
  \tomachseaone & \cuta{\bang{\la{\mvar_1}{\mvar_1}}}{\evar_1}
        \cuta{\la{\mvar_6}{\mvar_6}}{\mvar_3}\cuta{{\mvar_3}}{\mvar_5}\cuta{{\mvar_5}}{\mvar_4}\nctxhole\name & \nctxholep{{\mvar_4}}\name \,\cons\, \emptylist \\
  \tomachaxmone & \cuta{\bang{\la{\mvar_1}{\mvar_1}}}{\evar_1}\cuta{\la{\mvar_6}{\mvar_6}}{\mvar_3}\cuta{{\mvar_3}}{\mvar_5}\nctxhole\name & \nctxholep{{\mvar_5}}\name \,\cons\, \emptylist \\
  \tomachaxmone & \cuta{\bang{\la{\mvar_1}{\mvar_1}}}{\evar_1}\cuta{\la{\mvar_6}{\mvar_6}}{\mvar_3}\nctxhole\name & \nctxholep{{\mvar_3}}\name \,\cons\, \emptylist \\
  \tomachaxmone & \cuta{\bang{\la{\mvar_1}{\mvar_1}}}{\evar_1}\nctxhole\name & \nctxholep{\la{\mvar_6}{\mvar_6}}\name \,\cons\, \emptylist \\
  \tomachseafour & \cuta{\bang{\la{\mvar_1}{\mvar_1}}}{\evar_1}\la{\mvar_6}\nctxhole\name & \nctxholep{{\mvar_6}}\name \,\cons\, \emptylist \\
  \tomachseasix & \cuta{\bang{\la{\mvar_1}{\mvar_1}}}{\evar_1}\la{\mvar_6}{\mvar_6} & \emptylist \\
\end{array}
$\end{center}
Note that in the last state the cut $\cuta{\bang{\la{\mvar_1}{\mvar_1}}}{\evar_1}$ is garbage but it is not removed. This shall be taken care of by an extra garbage collection phase described in \refsect{gc-and-full}.

%

\section{Analysis of SESAME}
\label{sect:analysis}

\subparagraph{Invariants.} As for the BAM, we prove some invariants. The first one states that the multi context of a reachable state is an approximant. A second invariant states that names in the multi context are exactly those in the pool, where they have exactly one occurrence each. Then, we need a well-bound invariant about binders, relying on a technical definition in \techrep{\refapp{sesame}, page \pageref{def:well-bound-state}}\camerar{Appendix E of the tech report \cite{accattoli2024imell}}, along the lines of the invariant of the BAM. Lastly, the contextual decoding invariant, the most important and sophisticated invariant, states that a reachable state less the first job reads back to a good context. For the invariant to hold, the statement has to be generalized in a technical way, analogously to similar invariants in \cite{DBLP:conf/aplas/AccattoliB23}.

\begin{toappendix}
\begin{lemma}[SESAME qualitative invariants]
\NoteProof{l:sesam-invariants}
Let $\state=\twostate\namctx\pool$ be a SESAME state reachable from a well-bound initial term $\tm_0$.\label{l:sesam-invariants} %
\begin{enumerate}
\item \emph{Approximant}: $\namctx$ is an approximant.

\item \emph{Names}: jobs in the pool $\pool$ have pairwise distinct names, and $\fn\namctx=\fn\pool$.

\item 
\emph{Well-bound}: $\state$ is a well-bound state.


\item \emph{Contextual decoding}: if $\pool=\namctxholep{\tm_1}{\name_1}\cons\ldots\cons\namctxholep{\tm_k}{\name_k}$ has length $k\geq1$ then: 
\begin{center}
$\namctx^\state_{\tmtwo_1,\ldots,\tmtwo_{i-1}|\name_i|\tmtwo_{i+1},\ldots,\tmtwo_k} 
\defeq 
\namctx\namctxholep{\tmtwo_1}{\name_1}\ldots\namctxholep{\tmtwo_{i-1}}{\name_{i-1}}\namctxholep{\ctxhole}{\name_{i}}\namctxholep{\tmtwo_{i+1}}{\name_{i+1}}\ldots\namctxholep{\tmtwo_k}{\name_k}$
\end{center} 
is a good context for $i\in\set{1,\ldots,k}$ and any terms $\tmtwo_1,\ldots,\tmtwo_{i-1},\tmtwo_{i+1},\ldots,\tmtwo_k$ such that $\namctx^\state_{\tmtwo_1,\ldots,\tmtwo_{i-1}|\name_i|\tmtwo_{i+1},\ldots,\tmtwo_k}$ is a term.
\end{enumerate}
\end{lemma}
\end{toappendix}

\subparagraph{Implementation Theorem.} The well-bound and contextual decoding invariants are used to prove principal projection. The approximant and names invariants are used to prove the halt property. Search transparency and termination are straightforward. Then we apply the abstract implementation theorem (\refthm{abs-impl}), and obtain the implementation theorem.
\begin{toappendix}
\begin{proposition}
\label{prop:sesam-properties}
\NoteProof{prop:sesam-properties}
Let $a\in\set{\axmone,\axmtwo, \axeone,\axetwo,\lolli,\bang}$.
\begin{enumerate}
\item \emph{Search transparency}: if $\state \tomachsea \statetwo$ then $\decode\state=\decode\statetwo$.
\item \emph{Principal projection}: if $\state \tomachhole{a} \statetwo$ then $\decode\state \togp{a} \decode\statetwo$.
\item \emph{Search termination}: transition $\tomachsea$ terminates.
\item \emph{Halt}: if $\state$ is final then $\state=\twostate\tm\emptylist$ and $\decode\state=\tm$ is cut-free up to garbage.
\end{enumerate}
\end{proposition}
\end{toappendix}

\begin{theorem}[SESAME is good]
\label{thm:sesame-is-good}
SESAME is a big-step implementation of the non-erasing good strategy $\tognotw$ of ESC.
\end{theorem}

\subsection{Complexity Analysis} 
The complexity analysis of SESAME amounts to bound the cost of implementing a run $\run:\compilrel\tm\state \tosesame^*\statetwo$ on random access machines (RAMs) as a function of two parameters: the number $\sizepr\run$ of principal transitions in $\run$ (which are in bijection with the number of $\tognotw$ steps) and the size $\size\tm$ of the initial term/state $\tm$/$\state$. 

\subparagraph{Sub-Term Property.} The analysis relies on the \emph{sub-term property}, ensuring that the duplicating principal transitions $\tomachbang$ and $\tomachaxeone$ manipulate only sub-terms of the initial term. Therefore, the cost of duplications is connected to size of the initial term. The property for SESAME can be inferred by the one for the good strategy, but we prefer to give a direct simple proof. The statement of the related invariant is about cut values, not duplications, but note that such values are the only terms duplicated by SESAME. 
\begin{toappendix}
\begin{lemma}[Sub-term] 
\NoteProof{l:subterm-invariant}
	Let $\run \colon \compilrel\tm\state \tosesame^* \twostate\namctx\pool$ be a SESAME run.\label{l:subterm-invariant}
	
	\begin{enumerate}
	\item \emph{Invariant}: if $\val$ is a cut value in $\namctx$ or a value in $\pool$ then $\size\val \leq\size\tm$. 
	\item \emph{Property}: if $\val$ is a value duplicated along $\run$ (by $\tomachbang$ or $\tomachaxeone$) then  $\size\val \leq\size\tm$. 
	\end{enumerate}
\end{lemma}
\end{toappendix}

Since search transitions decrease the number of term constructors in the pool, which is only increased by principal transitions and of a quantity bounded by $\size\tm$ (by the sub-term invariant), we obtain the following bound.
\begin{toappendix}
\begin{lemma}[Search transitions are bi-linear in number] 
\NoteProof{l:search-is-bilinear}
Let $\run \colon \compilrel{\tm}\state \tosesame^* \statetwo$ be a SESAME run. Then $\sizesea\run\leq \size\tm\cdot(\sizepr\run+1)$.
 \label{l:search-is-bilinear}
\end{lemma}
\end{toappendix}

\subparagraph{Cost of Single Transitions.} For bounding the total cost, we need to make some hypotheses on how SESAME is going to be itself implemented on RAMs:
  \begin{enumerate}
    \item \emph{Variable occurrences, binders, and cuts}:  a variable is a memory location, a variable occurrence is a reference to it, and a cut $\cuta\val\var$ is the fact that the location associated with $\var$ contains $\val$;
    \item \emph{Random access to variables}: the cuts in $\namctx$ can be accessed in $\bigo(1)$ by just following the reference given by the variable occurrence, with no need to search through $\namctx$;
    \item \emph{Named holes and jobs}: a named hole $\namctxholea$ is again a memory location $\name$, and a job $\namctxholep\tm\name$ is a pointer to $\name$ and expresses the fact that location $\name$ contains $\tm$;
    \item \emph{Sequences of left rules}: cuts have a back pointer to the constructors before them (for instance in $\dera\evar\mvar\cuta\val\var\suba\mvar\valtwo\vartwo\tm$ the cut $\cuta\val\var$ has a pointer to $\dera\evar\mvar$) so that the removal of a cut from the multi context $\namctx$ of a state---needed for the multiplicative transitions $\tomachaxmone$, $\tomachaxmtwo$, and $\tomachlolli$---can be performed in $\bigo(1)$.
  \end{enumerate}
As it is standard for time analyses, we also assume that pointers can be managed in $\bigo(1)$. These hypotheses mimic similar ones behind machines for $\lambda$-calculi, which are shown to be implementable in OCaml by Accattoli and Barras \cite{DBLP:conf/ppdp/AccattoliB17} and Accattoli et al. \cite{DBLP:conf/lics/AccattoliCC21}, and are followed by the OCaml implementation outlines in \refsect{ocaml}. They allow us to consider search transitions as having constant cost, and principal transitions as having cost bound by the initial term, by the sub-term property.

\begin{toappendix}
\begin{lemma}[Cost of single transitions]
\NoteProof{l:cost-single-trans}
Let $\run \colon \compilrel{\tm}\state \tosesame^* \statetwo$ be a SESAME run. Search (resp. principal) transitions of $\run$ are implementable in $\bigo(1)$ (resp. $\bigo(\size\tm)$).\label{l:cost-single-trans}
\end{lemma}
\end{toappendix}

\subparagraph{Summing Up.} By putting together the bounds on the number of search transitions with the cost of single transitions we obtain the complexity of SESAME.

\begin{toappendix}
\begin{theorem}[SESAME bi-linear overhead bound]
\NoteProof{thm:sesame-overhead-bound}
  Let   $\run \colon \compilrel{\tm}\state \tosesame^* \statetwo$ be a SESAME run. Then $\run$ is implementable on RAMs in $\bigo\left(\size\tm\cdot(\sizepr\run+1)\right)$.\label{thm:sesame-overhead-bound}
\end{theorem}
\end{toappendix}

\section{Final Garbage Collection and Full Cut Elimination}
\label{sect:gc-and-full}
By the halt property (\refprop{sesam-properties}.4), SESAME stops on final state $\state=\twostate\appr\emptylist$ where $\appr$ is a (pre-)term $\tm$ that is cut-free \emph{up to garbage}. The final garbage collection process turning $\tm$ into a cut-free proof term $\gcdec\tm$ is formalized as the following function:
\begin{center}
$
\arraycolsep=2pt\begin{array}{rll@{\hspace{.3cm}}|@{\hspace{.3cm}} rll@{\hspace{.3cm}}|@{\hspace{.3cm}} rll}
\multicolumn{9}{c}{\textsc{Final garbage Collection}}
\\[4pt]
\gcdec\var &\defeq& \var
&
\gcdec{\la\var\tm} &\defeq& \la\var\gcdec{\tm}
&
\gcdec{\bang\tm} &\defeq& \bang\gcdec{\tm}
\\
\gcdec{\cuta\val\var\tm} &\defeq& \gcdec{\tm}
&
\gcdec{\suba\mvar\val\var\tm} &\defeq& \suba\mvar{\gcdec\val}\var\gcdec{\tm}
&
\gcdec{\dera\evar\var\tm} &\defeq& \dera\evar\var\gcdec{\tm}
\end{array}$
\end{center}
Clearly, $\gcdec\tm$ is cut-free. The following proposition states the qualitative and quantitative properties of garbage collection. Its first point uses the lemma before it, while the second point rests on the SESAME bi-linear bound.

\begin{toappendix}
\begin{lemma}
\label{l:good-w-redex-exists}
\NoteProof{l:good-w-redex-exists}
Let $\tm$ be a term that is cut-free up to garbage but not cut-free. 
Then in $\tm$ there is at least an out cut that is a $\togp\wsym$-redex.
\end{lemma}
\end{toappendix}

\begin{toappendix}
\begin{proposition}
\label{prop:gc-properties}
\NoteProof{prop:gc-properties}
Let $\run: \compilrel\tm\state \tosesame^* \twostate\tmtwo\emptylist$ a SESAME run ending on a final state.
\begin{enumerate}
\item \emph{Good weakening steps simulate GC}: $\tmtwo \togp\wsym^k \gcdec\tmtwo$ where $k$ is the number of out cuts in $\tmtwo$;
\item \emph{GC is bi-linear}: the final garbage collection function $\tmtwo \mapsto \gcdec\tmtwo$ can be implemented on RAMs in $ \bigo(\size\tm\cdot(\sizepr\run+1))$.
\end{enumerate}
\end{proposition}
\end{toappendix}

By simulating good cut-elimination via the SESAME followed by final garbage collection, we obtain our main result, which---beyond the bi-linear bound---is a further slight (but sort of obvious) refinement of Accattoli's result, as it allows one to count only \emph{non-erasing} good steps, rather than general good steps.
\begin{theorem}[Good cut-elimination is big-step implementable with bi-linear overhead]
  Let   $\deriv:\tm \tog^* \tmtwo$ be a good evaluation sequence with $\tmtwo$ cut-free. Then $\tmtwo$ is computable from $\tm$ on RAMs in $\bigo\left(\size\tm\cdot(\sizep\deriv{\neg\wsym}+1)\right)$.\label{thm:final}
\end{theorem}

\begin{proof}
\applabel{thm:final}
If $\deriv:\tm \tog^* \tmtwo$ with $\tmtwo$ cut-free then $\tognotw$ terminates, by the \emph{uniform normalization} corollary of the diamond property for $\tog$ (see the end of \refsect{strategy}). Then there is an evaluation $\derivtwo:\tm \tognotw^* \tmtwo'$ with $\tmtwo'$ cut-free up to garbage. Since SESAME is a big-step implementation of $\tognotw$ (\refthm{sesame-is-good}), we obtain a SESAME run $\run :\compilrel\tm\state \tosesame^*\statetwo$ with $\decode\statetwo=\tmtwo'$, of cost $\bigo\left(\size\tm\cdot(\sizepr\run+1)\right)=\bigo\left(\size\tm\cdot(\size{\derivtwo}+1)\right)$. By the halt property of SESAME and cut-freeness up to garbage of SESAME, we obtain that $\statetwo$ is final, that is, $\statetwo=\twostate\tmthree\emptylist$ with $\tmthree =_\alpha \tmtwo'$.

Since good weakening steps simulate GC (\refprop{gc-properties}.1), we obtain $\derivthree: \tmthree \togp\wsym^* \gcdec{\tmtwo'}$ with $\gcdec{\tmtwo'}$ cut-free. Concatenating $\derivtwo$ and $\derivthree$ we obtain a good and normalizing evaluation sequence $\derivfour:\tm \tognotw^*\togp\wsym^*\gcdec{\tmtwo'}$. The diamond property of the good strategy (\refthm{sub-term-prop}.2) implies \emph{confluence} and \emph{length invariance} of $\tog$ (see the end of \refsect{strategy}), that is, we obtain $\tmtwo = \gcdec{\tmtwo'}$ and $\size\deriv=\size\derivfour$. Since computing GC is bi-linear (\refprop{gc-properties}.1), and precisely costs again $\bigo\left(\size\tm\cdot(\sizepr\run+1)\right)=\bigo\left(\size\tm\cdot(\size{\derivtwo}+1)\right)$, we obtain that the total cost of SESAME followed by GC is $2\cdot\bigo\left(\size\tm\cdot(\size{\derivtwo}+1)\right)= \bigo\left(\size\tm\cdot(\size{\derivtwo}+1)\right)$.

Lastly, the diamond property of the good strategy (\refthm{sub-term-prop}.2) as stated in \cite{DBLP:journals/lmcs/Accattoli23} does not give information about the kinds of steps, but by looking at its proof it is easily seen that diamond diagrams preserve the kind of steps. Therefore, not only we have $\size\deriv=\size\derivfour$, but also $\sizep\deriv{\neg\wsym}=\sizep\derivfour{\neg\wsym}$. Finally, note that $\sizep\derivfour{\neg\wsym}=\size{\derivtwo}$, so the cost actually is $\bigo\left(\size\tm\cdot(\sizep\deriv{\neg\wsym}+1)\right)$.
\end{proof}

\section{Implementation in OCaml}
\label{sect:ocaml}
An implementation in OCaml can be found on GitHub at
\url{https://github.com/sacerdot/sesame/}. All the datatypes and functions
of the implementation are documented at \url{https://sacerdot.github.io/sesame}.

\subparagraph{Aim and Design.}
The implementation is provided to support evidence for the cost of transitions claimed by \reflemma{cost-single-trans}. Moreover, it
allows one to easily study the computational behaviour of family of terms by observing
their evaluation.

The implementation has \emph{not} been heavily optimized, in order to keep the code readable and
close to the pen-and-paper presentation. Nevertheless, the employed data structures 
are reasonably close to those of an optimized implementation. Advanced recent OCaml
features, like Generalized Algebraic Data Types, could have been used to 
statically enforce more invariants. We sticked instead to a simpler subset of OCaml,
to make the code readable to non experts of the language. As a consequence, in several
places there are assertions to abort the program in case the invariants
are violated (possible only in case of bugs).

\subparagraph{Overview.}
The implementation consists of a Read-Eval-Print-Loop (REPL) that asks the user
to enter an IMELL term to be reduced, in ESC syntax. The accepted BNF is printed by
the executable before starting the REPL. The term is first checked to be proper,
then its strong normal form up to garbage is computed by the SESAME machine and
finally garbage is removed. All intermediate machine steps are shown.

Before starting the REPL, a few test terms are reduced and printed (among which the example at the end of \refsect{sesame}), to check the 
machine functionalities and as examples of the input syntax. In particular, the last tests are terms from the exploding family discussed in the next section.

\subparagraph{Data Structures.}
Terms are encoded at runtime as term graphs, which are Direct Acyclic Graphs (DAG)
augmented with a few back edges. Nodes in the graph are given by constructors
of Algebraic Data Types whose arguments are mutable, to imperatively
change the graph during reduction. All type declarations are given in the
\verb+termGraphs.ml+ file. The data structures used in the implementation are described in \techrep{\refapp{ocaml}}\camerar{Appendix H of the tech report \cite{accattoli2024imell}}. 

\subparagraph{Machine Runs: Auxiliary Functions and Their Complexity.}

The code that implements the runs of SESAME to normal form and the final garbage collection
process can be found in \verb+reduction.ml+, which is independent from all the remaining
files but \verb+termGraphs.ml+. Both files together amount to 412 lines of
commented OCaml code.

A quick inspection of the code in \verb+reduction.ml+ shows only a few non-trivial
functions, all the others being constant time:
\begin{enumerate}
 \item \verb+copy_{term,value,var,bvar}+ that copies the
   DAG in input while visiting it. 
  A $\bigo(1)$ fresh name generator is used
    to assign a new \verb+name+ value to copies of variables.
 \item \verb+alpha+ and \verb+enter_bo+, used respectively in the
   $\bang$ and $\lolli$ transitions, that take the body of a promotion/abstraction,
   traverse it to split it into a left context and a value, and build a new term
   by glueing together the obtained left context, a new cut that uses the value, and
   a remaining term. The two functions differ only in the fact that \verb+alpha+ also
   copies ($\alpha$-renames) the term in input, calling the \verb+copy_*+ functions
   on the sub-terms, while \verb+enter_bo+ consumes the given term.
 \item \verb+steps+, the SESAME main loop, that runs until the pool is empty and
   a normal form is therefore reached, or it diverges otherwise.
 \item \verb+gc_{value,term}+ that traverse in linear time the input to remove
   all garbage cuts from a normal form.
\end{enumerate}

Linearity for all the previous functions but \verb+steps+ is easily established
observing that the functions are based on visits of DAGs that never visit a node twice. In particular, the functions of Point 2 are linear in the size of the given term, which, by the
   sub-term invariant, is a (copy of a) sub-term of the initial term.

The only major source of technicalities---which however do not affect the complexity---is the fact that multiplicative transitions remove the acting cut from the multi context. For example, consider a multiplicative step $\dera\evar\evartwo \cuta\val\mvar \cuta\valtwo\vartwo \tm \to \dera\evar\evartwo  \cuta\valtwo\vartwo \tm'$ involving the cut on $\mvar$. To implement it, beyond manipulating $\val$, $\tm$, and $\tm'$, one also has to connect $\dera\evar\evartwo$ and $\cuta\valtwo\vartwo$. For that, cuts have a back-pointer to the preceding constructor (as mentioned at Point 4, labeled \emph{sequences of left rule}, before \reflemma{cost-single-trans}), which induces inelegant imperative manipulations of the term graph,
in particular in the copy function. The back-pointer on $\cuta\val\mvar$ is used to retrieve the dereliction $\dera\evar\evartwo$ from $\cuta\val\mvar$, while making the new connection of $\dera\evar\evartwo$ and $\cuta\valtwo\vartwo$ requires changing two pointers: the one from the dereliction to its body, which has to target $\cuta\valtwo\vartwo$, plus the back-pointer from $\cuta\valtwo\vartwo$, which now has to target $\dera\evar\evartwo$.

The general scheme for addressing this issue is to
augment the input $\tm$ of functions that take a term with the pointer to the
parent node of $\tm$, in order to be able to reassign the back-pointer when $\tm$ is
a cut.

\subparagraph{Parsing, Pretty Printing, and Check for Properness.}
We have not attempted to achieve the best asymptotic costs for these functionalities
since they are not relevant for the paper and since in practice we are manipulating
small terms in input. To keep the code reproducible in the long term, we have not
used any external library or tool for parsing and pretty-printing. The printed version of the example run at the end of \refsect{sesame} is shown in \techrep{\refapp{example}}\camerar{Appendix I of the tech report \cite{accattoli2024imell}}.
\section{What's Next: an Interesting Family of Terms}
We here hint at what we consider the most interesting future work, already mentioned by Accattoli in the conclusions of \cite{DBLP:journals/lmcs/Accattoli23}. The difference is that here we provide a family of terms showing that such a future work is challenging.

\subparagraph{The Question.} \emph{One} $\beta$-step of the $\l$-calculus is simulated in IMELL by \emph{one} multiplicative step (actually a $\tololli$ step) followed by possibly \emph{many} exponential steps. Since for many strategies the number of $\beta$-steps is a polynomial cost model, it means that one can count only the number of multiplicative (or even $\tololli$) steps, that is, \emph{one can count zero for exponential steps}. Such a surprising fact is established easily in the case of weak evaluation with closed terms (roughly, it can be proved via standard abstract machines), while for strong evaluation it requires a sophisticated additional technique called \emph{useful sharing} \cite{DBLP:journals/corr/AccattoliL16,DBLP:conf/lics/AccattoliCC21,DBLP:conf/csl/AccattoliL22}. A question naturally arises: \emph{is the number of ESC multiplicative/$\tololli$ good steps a polynomial time cost model as well?} This question---left to future work---is far from obvious. In the $\l$-calculus, there is a strong, hardcoded correlation between multiplicatives and exponentials, not present in IMELL. In the standard call-by-name/value encodings of $\l$-calculus in IMELL, indeed, multiplicatives and exponentials connectives \emph{rigidly alternate}, while IMELL also has consecutive exponentials, as in $\bang\bang\form$, enabling wilder exponential behaviour. A hint that question might have a positive answer is given by the strong normalization of untyped exponentials  in ESC/IMELL shown in \cite{DBLP:journals/lmcs/Accattoli23} (while they are \emph{not} SN in MELL, see \cite{DBLP:journals/lmcs/Accattoli23}).

\subparagraph{The Question is Challenging.} We here show that the question mentioned above is challenging, and---surprisingly---it is already challenging in the apparently simple setting of basic evaluation with closed terms, that is the ESC analogous of weak evaluation for the $\l$-calculus, for which instead the question has an easy answer. The challenging aspect is here shown by building a family of ESC terms $\sigma_n$ having size linear in $n$, whose evaluation is basic and made out of exponential steps only, and of a number of steps that is exponential in $n$. This is impossible in the $\l$-calculus, and it is crucially related to iterated exponential terms such as $\bang\bang\tm$. Since the evaluation of $\sigma_n$ uses 0 multiplicative steps, it suggests that the number of multiplicative steps is \emph{not} a polynomial cost model. More precisely, it shows that a smart additional mechanism---akin to useful sharing---is needed in order to circumvent the exponential number of exponential steps. Useful sharing as it appears in the literature, however, cannot be the answer, since useful sharing addresses exponential inefficiencies that are induced by strong evaluation, which is an orthogonal issue.

\subparagraph{The Exponential Exploding Family.} The family is built in three steps. 
\begin{enumerate}
\item For $k \geq 1$, define $\pi_k := \overbrace{\dera\evartwo\_ \ldots \dera\evartwo\_ \dera\evartwo\evar}^{k}\evar$, where $\_$ is an exponential variable with no occurrences, and whose name is irrelevant. These terms reduce to (the cut-free term) $\pi_{k\cdot h}$ via  exponential basic evaluation, that is, $\cuta{\bang{\pi_k}}\evartwo{\pi_h} \toe^* \pi_{k\cdot h}$.
\item Define $\delta_n := \overbrace{\cuta{\bang{\pi_2}}\evartwo \ldots \cuta{\bang{\pi_2}}\evartwo}^{n-1}\pi_2$, and note that, by Point 1, one has $\delta_n \toe^* \pi_{2^n}$.
\item Finally, define the exponential exploding family as $\sigma_n \defeq \cuta{\bang{\bang{\la\mvar\mvar}}}\evartwo\delta_n$. By Point 2, we have $\sigma_n \toe^* \cuta{\bang{\bang{\la\mvar\mvar}}}\evartwo \pi_{2^n}$. Then note that $\cuta{\bang{\bang{\la\mvar\mvar}}}\evartwo \pi_{2^n} \toe^{\Omega(2^n)} \bang{\la\mvar\mvar}$.
\end{enumerate}
Summing up, $\sigma_n$ is an example of term of size $\bigo(n)$ that reduces in $\Omega(2^n)$ exponential steps and 0 multiplicative steps to a cut-free proof using only basic evaluation. As mentioned in the previous section, the OCaml implementation starts by running some tests, including examples of the three points above, respectively for $k=3,h=4$, for $n=3$, and for $n=3$.

\bibliography{main.bbl}

\begin{thebibliography}{10}

\bibitem{DBLP:journals/tcs/Abramsky93}
Samson Abramsky.
\newblock Computational interpretations of linear logic.
\newblock {\em Theor. Comput. Sci.}, 111(1{\&}2):3--57, 1993.
\newblock \href {https://doi.org/10.1016/0304-3975(93)90181-R}
  {\path{doi:10.1016/0304-3975(93)90181-R}}.

\bibitem{DBLP:conf/rta/Accattoli12}
Beniamino Accattoli.
\newblock An abstract factorization theorem for explicit substitutions.
\newblock In Ashish Tiwari, editor, {\em 23rd International Conference on
  Rewriting Techniques and Applications (RTA'12) , {RTA} 2012, May 28 - June 2,
  2012, Nagoya, Japan}, volume~15 of {\em LIPIcs}, pages 6--21. Schloss
  Dagstuhl - Leibniz-Zentrum f{\"{u}}r Informatik, 2012.
\newblock \href {https://doi.org/10.4230/LIPICS.RTA.2012.6}
  {\path{doi:10.4230/LIPICS.RTA.2012.6}}.

\bibitem{DBLP:conf/wollic/Accattoli16}
Beniamino Accattoli.
\newblock The useful {MAM}, a reasonable implementation of the strong
  {\(\lambda\)}-calculus.
\newblock In {\em Logic, Language, Information, and Computation - 23rd
  International Workshop, WoLLIC 2016, Puebla, Mexico, August 16-19th, 2016.
  Proceedings}, pages 1--21, 2016.
\newblock \href {https://doi.org/10.1007/978-3-662-52921-8\_1}
  {\path{doi:10.1007/978-3-662-52921-8\_1}}.

\bibitem{DBLP:journals/lmcs/Accattoli23}
Beniamino Accattoli.
\newblock Exponentials as substitutions and the cost of cut elimination in
  linear logic.
\newblock {\em Log. Methods Comput. Sci.}, 19(4), 2023.
\newblock \href {https://doi.org/10.46298/LMCS-19(4:23)2023}
  {\path{doi:10.46298/LMCS-19(4:23)2023}}.

\bibitem{DBLP:conf/aplas/AccattoliB23}
Beniamino Accattoli and Pablo Barenbaum.
\newblock A diamond machine for strong evaluation.
\newblock In Chung{-}Kil Hur, editor, {\em Programming Languages and Systems -
  21st Asian Symposium, {APLAS} 2023, Taipei, Taiwan, November 26-29, 2023,
  Proceedings}, volume 14405 of {\em Lecture Notes in Computer Science}, pages
  69--90. Springer, 2023.
\newblock \href {https://doi.org/10.1007/978-981-99-8311-7\_4}
  {\path{doi:10.1007/978-981-99-8311-7\_4}}.

\bibitem{DBLP:conf/icfp/AccattoliBM14}
Beniamino Accattoli, Pablo Barenbaum, and Damiano Mazza.
\newblock Distilling abstract machines.
\newblock In Johan Jeuring and Manuel M.~T. Chakravarty, editors, {\em
  Proceedings of the 19th {ACM} {SIGPLAN} international conference on
  Functional programming, Gothenburg, Sweden, September 1-3, 2014}, pages
  363--376. {ACM}, 2014.
\newblock \href {https://doi.org/10.1145/2628136.2628154}
  {\path{doi:10.1145/2628136.2628154}}.

\bibitem{DBLP:journals/corr/AccattoliBM15}
Beniamino Accattoli, Pablo Barenbaum, and Damiano Mazza.
\newblock A strong distillery.
\newblock {\em CoRR}, abs/1509.00996, 2015.
\newblock \href {https://arxiv.org/abs/1509.00996} {\path{arXiv:1509.00996}}.

\bibitem{DBLP:conf/ppdp/AccattoliB17}
Beniamino Accattoli and Bruno Barras.
\newblock Environments and the complexity of abstract machines.
\newblock In Wim Vanhoof and Brigitte Pientka, editors, {\em Proceedings of the
  19th International Symposium on Principles and Practice of Declarative
  Programming, Namur, Belgium, October 09 - 11, 2017}, pages 4--16. {ACM},
  2017.
\newblock \href {https://doi.org/10.1145/3131851.3131855}
  {\path{doi:10.1145/3131851.3131855}}.

\bibitem{DBLP:conf/popl/AccattoliBKL14}
Beniamino Accattoli, Eduardo Bonelli, Delia Kesner, and Carlos Lombardi.
\newblock A nonstandard standardization theorem.
\newblock In Suresh Jagannathan and Peter Sewell, editors, {\em The 41st Annual
  {ACM} {SIGPLAN-SIGACT} Symposium on Principles of Programming Languages,
  {POPL} '14, San Diego, CA, USA, January 20-21, 2014}, pages 659--670. {ACM},
  2014.
\newblock \href {https://doi.org/10.1145/2535838.2535886}
  {\path{doi:10.1145/2535838.2535886}}.

\bibitem{DBLP:conf/lics/AccattoliC15}
Beniamino Accattoli and Claudio~Sacerdoti Coen.
\newblock On the relative usefulness of fireballs.
\newblock In {\em 30th Annual {ACM/IEEE} Symposium on Logic in Computer
  Science, {LICS} 2015, Kyoto, Japan, July 6-10, 2015}, pages 141--155. {IEEE}
  Computer Society, 2015.
\newblock \href {https://doi.org/10.1109/LICS.2015.23}
  {\path{doi:10.1109/LICS.2015.23}}.

\bibitem{DBLP:conf/lics/AccattoliCC21}
Beniamino Accattoli, Andrea Condoluci, and Claudio~Sacerdoti Coen.
\newblock Strong call-by-value is reasonable, implosively.
\newblock In {\em 36th Annual {ACM/IEEE} Symposium on Logic in Computer
  Science, {LICS} 2021, Rome, Italy, June 29 - July 2, 2021}, pages 1--14.
  {IEEE}, 2021.
\newblock \href {https://doi.org/10.1109/LICS52264.2021.9470630}
  {\path{doi:10.1109/LICS52264.2021.9470630}}.

\bibitem{DBLP:conf/ppdp/AccattoliCGC19}
Beniamino Accattoli, Andrea Condoluci, Giulio Guerrieri, and Claudio
  Sacerdoti~Coen.
\newblock Crumbling abstract machines.
\newblock In {\em Proceedings of the 21st International Symposium on Principles
  and Practice of Programming Languages, {PPDP} 2019, Porto, Portugal, October
  7-9, 2019}, pages 4:1--4:15. {ACM}, 2019.
\newblock \href {https://doi.org/10.1145/3354166.3354169}
  {\path{doi:10.1145/3354166.3354169}}.

\bibitem{DBLP:conf/rta/AccattoliL12}
Beniamino Accattoli and Ugo {Dal Lago}.
\newblock {On the Invariance of the Unitary Cost Model for Head Reduction}.
\newblock In {\em 23rd International Conference on Rewriting Techniques and
  Applications (RTA'12) , {RTA} 2012, May 28 - June 2, 2012, Nagoya, Japan},
  pages 22--37. Schloss Dagstuhl - Leibniz-Zentrum f{\"{u}}r Informatik, 2012.
\newblock \href {https://doi.org/10.4230/LIPIcs.RTA.2012.22}
  {\path{doi:10.4230/LIPIcs.RTA.2012.22}}.

\bibitem{DBLP:journals/corr/AccattoliL16}
Beniamino Accattoli and Ugo {Dal Lago}.
\newblock ({L}eftmost-outermost) {B}eta reduction is invariant, indeed.
\newblock {\em Logical Methods in Computer Science}, 12(1), 2016.
\newblock \href {https://doi.org/10.2168/LMCS-12(1:4)2016}
  {\path{doi:10.2168/LMCS-12(1:4)2016}}.

\bibitem{DBLP:journals/pacmpl/AccattoliLV21}
Beniamino Accattoli, Ugo Dal~Lago, and Gabriele Vanoni.
\newblock The (in)efficiency of interaction.
\newblock {\em Proc. {ACM} Program. Lang.}, 5({POPL}):1--33, 2021.
\newblock \href {https://doi.org/10.1145/3434332} {\path{doi:10.1145/3434332}}.

\bibitem{DBLP:conf/lics/AccattoliLV21}
Beniamino Accattoli, Ugo Dal~Lago, and Gabriele Vanoni.
\newblock The space of interaction.
\newblock In {\em 36th Annual {ACM/IEEE} Symposium on Logic in Computer
  Science, {LICS} 2021, Rome, Italy, June 29 - July 2, 2021}, pages 1--13.
  {IEEE}, 2021.
\newblock \href {https://doi.org/10.1109/LICS52264.2021.9470726}
  {\path{doi:10.1109/LICS52264.2021.9470726}}.

\bibitem{DBLP:journals/scp/AccattoliG19}
Beniamino Accattoli and Giulio Guerrieri.
\newblock Abstract machines for open call-by-value.
\newblock {\em Sci. Comput. Program.}, 184, 2019.
\newblock \href {https://doi.org/10.1016/j.scico.2019.03.002}
  {\path{doi:10.1016/j.scico.2019.03.002}}.

\bibitem{DBLP:conf/lics/AccattoliLV22}
Beniamino Accattoli, Ugo~Dal Lago, and Gabriele Vanoni.
\newblock Reasonable space for the {\(\lambda\)}-calculus, logarithmically.
\newblock In Christel Baier and Dana Fisman, editors, {\em {LICS} '22: 37th
  Annual {ACM/IEEE} Symposium on Logic in Computer Science, Haifa, Israel,
  August 2 - 5, 2022}, pages 47:1--47:13. {ACM}, 2022.
\newblock \href {https://doi.org/10.1145/3531130.3533362}
  {\path{doi:10.1145/3531130.3533362}}.

\bibitem{DBLP:conf/csl/AccattoliL22}
Beniamino Accattoli and Maico Leberle.
\newblock Useful open call-by-need.
\newblock In Florin Manea and Alex Simpson, editors, {\em 30th {EACSL} Annual
  Conference on Computer Science Logic, {CSL} 2022, February 14-19, 2022,
  G{\"{o}}ttingen, Germany (Virtual Conference)}, volume 216 of {\em LIPIcs},
  pages 4:1--4:21. Schloss Dagstuhl - Leibniz-Zentrum f{\"{u}}r Informatik,
  2022.
\newblock \href {https://doi.org/10.4230/LIPICS.CSL.2022.4}
  {\path{doi:10.4230/LIPICS.CSL.2022.4}}.

\bibitem{alberti1998efficient}
Francisco Alberti and Eike Ritter.
\newblock An efficient linear abstract machine with single-pointer property.
\newblock ESSLLI, 1998.

\bibitem{DBLP:conf/ppdp/BiernackaCD21}
Malgorzata Biernacka, Witold Charatonik, and Tomasz Drab.
\newblock A derived reasonable abstract machine for strong call by value.
\newblock In Niccol{\`{o}} Veltri, Nick Benton, and Silvia Ghilezan, editors,
  {\em {PPDP} 2021: 23rd International Symposium on Principles and Practice of
  Declarative Programming, Tallinn, Estonia, September 6-8, 2021}, pages
  6:1--6:14. {ACM}, 2021.
\newblock \href {https://doi.org/10.1145/3479394.3479401}
  {\path{doi:10.1145/3479394.3479401}}.

\bibitem{DBLP:journals/pacmpl/BiernackaCD22}
Malgorzata Biernacka, Witold Charatonik, and Tomasz Drab.
\newblock A simple and efficient implementation of strong call by need by an
  abstract machine.
\newblock {\em Proc. {ACM} Program. Lang.}, 6({ICFP}):109--136, 2022.
\newblock \href {https://doi.org/10.1145/3549822} {\path{doi:10.1145/3549822}}.

\bibitem{DBLP:conf/fpca/BlellochG95}
Guy~E. Blelloch and John Greiner.
\newblock Parallelism in sequential functional languages.
\newblock In {\em Proceedings of the seventh international conference on
  Functional programming languages and computer architecture, {FPCA} 1995, La
  Jolla, California, USA, June 25-28, 1995}, pages 226--237, 1995.
\newblock \href {https://doi.org/10.1145/224164.224210}
  {\path{doi:10.1145/224164.224210}}.

\bibitem{DBLP:journals/entcs/Bonelli06}
Eduardo Bonelli.
\newblock The linear logical abstract machine.
\newblock In Stephen~D. Brookes and Michael~W. Mislove, editors, {\em
  Proceedings of the 22nd Annual Conference on Mathematical Foundations of
  Programming Semantics, {MFPS} 2006, Genova, Italy, May 23-27, 2006}, volume
  158 of {\em Electronic Notes in Theoretical Computer Science}, pages 99--121.
  Elsevier, 2006.
\newblock \href {https://doi.org/10.1016/J.ENTCS.2006.04.007}
  {\path{doi:10.1016/J.ENTCS.2006.04.007}}.

\bibitem{DBLP:journals/lisp/Cregut07}
Pierre Cr{\'{e}}gut.
\newblock Strongly reducing variants of the {Krivine} abstract machine.
\newblock {\em High. Order Symb. Comput.}, 20(3):209--230, 2007.
\newblock \href {https://doi.org/10.1007/s10990-007-9015-z}
  {\path{doi:10.1007/s10990-007-9015-z}}.

\bibitem{DBLP:conf/cie/LagoM06}
Ugo Dal~Lago and Simone Martini.
\newblock An invariant cost model for the lambda calculus.
\newblock In Arnold Beckmann, Ulrich Berger, Benedikt L{\"{o}}we, and John~V.
  Tucker, editors, {\em Logical Approaches to Computational Barriers, Second
  Conference on Computability in Europe, CiE 2006, Swansea, UK, June 30-July 5,
  2006, Proceedings}, volume 3988 of {\em Lecture Notes in Computer Science},
  pages 105--114. Springer, 2006.
\newblock \href {https://doi.org/10.1007/11780342\_11}
  {\path{doi:10.1007/11780342\_11}}.

\bibitem{DBLP:conf/fopara/LagoM09}
Ugo Dal~Lago and Simone Martini.
\newblock Derivational complexity is an invariant cost model.
\newblock In Marko C. J.~D. van Eekelen and Olha Shkaravska, editors, {\em
  Foundational and Practical Aspects of Resource Analysis - First International
  Workshop, {FOPARA} 2009, Eindhoven, The Netherlands, November 6, 2009,
  Revised Selected Papers}, volume 6324 of {\em Lecture Notes in Computer
  Science}, pages 100--113. Springer, 2009.
\newblock \href {https://doi.org/10.1007/978-3-642-15331-0\_7}
  {\path{doi:10.1007/978-3-642-15331-0\_7}}.

\bibitem{DBLP:conf/icalp/LagoM09}
Ugo Dal~Lago and Simone Martini.
\newblock On constructor rewrite systems and the lambda-calculus.
\newblock In Susanne Albers, Alberto Marchetti{-}Spaccamela, Yossi Matias,
  Sotiris~E. Nikoletseas, and Wolfgang Thomas, editors, {\em Automata,
  Languages and Programming, 36th Internatilonal Colloquium, {ICALP} 2009,
  Rhodes, Greece, July 5-12, 2009, Proceedings, Part {II}}, volume 5556 of {\em
  Lecture Notes in Computer Science}, pages 163--174. Springer, 2009.
\newblock \href {https://doi.org/10.1007/978-3-642-02930-1\_14}
  {\path{doi:10.1007/978-3-642-02930-1\_14}}.

\bibitem{DBLP:journals/entcs/DanosR96}
Vincent Danos and Laurent Regnier.
\newblock Reversible, irreversible and optimal lambda-machines.
\newblock In Jean{-}Yves Girard, Mitsuhiro Okada, and Andre Scedrov, editors,
  {\em Linear Logic Tokyo Meeting 1996, Keio University, Mita Campus, Tokyo,
  Japan, March 29 - April 2, 1996}, volume~3 of {\em Electronic Notes in
  Theoretical Computer Science}, pages 40--60. Elsevier, 1996.
\newblock \href {https://doi.org/10.1016/S1571-0661(05)80402-5}
  {\path{doi:10.1016/S1571-0661(05)80402-5}}.

\bibitem{DBLP:conf/csl/Herbelin94}
Hugo Herbelin.
\newblock A lambda-calculus structure isomorphic to gentzen-style sequent
  calculus structure.
\newblock In Leszek Pacholski and Jerzy Tiuryn, editors, {\em Computer Science
  Logic, 8th International Workshop, {CSL} '94, Kazimierz, Poland, September
  25-30, 1994, Selected Papers}, volume 933 of {\em Lecture Notes in Computer
  Science}, pages 61--75. Springer, 1994.
\newblock \href {https://doi.org/10.1007/BFB0022247}
  {\path{doi:10.1007/BFB0022247}}.

\bibitem{DBLP:journals/corr/abs-2312-13270}
Delia Kesner and Shane~{\'{O}} Conch{\'{u}}ir.
\newblock Milner's lambda-calculus with partial substitutions.
\newblock {\em CoRR}, abs/2312.13270, 2023.
\newblock \href {https://arxiv.org/abs/2312.13270} {\path{arXiv:2312.13270}}.

\bibitem{DBLP:journals/tcs/Lafont88}
Yves Lafont.
\newblock The linear abstract machine.
\newblock {\em Theor. Comput. Sci.}, 59:157--180, 1988.
\newblock \href {https://doi.org/10.1016/0304-3975(88)90100-4}
  {\path{doi:10.1016/0304-3975(88)90100-4}}.

\bibitem{DBLP:journals/tcs/LagoM08}
Ugo~Dal Lago and Simone Martini.
\newblock The weak lambda calculus as a reasonable machine.
\newblock {\em Theor. Comput. Sci.}, 398(1-3):32--50, 2008.
\newblock \href {https://doi.org/10.1016/J.TCS.2008.01.044}
  {\path{doi:10.1016/J.TCS.2008.01.044}}.

\bibitem{DBLP:conf/popl/Mackie95}
Ian Mackie.
\newblock The geometry of interaction machine.
\newblock In Ron~K. Cytron and Peter Lee, editors, {\em Conference Record of
  POPL'95: 22nd {ACM} {SIGPLAN-SIGACT} Symposium on Principles of Programming
  Languages, San Francisco, California, USA, January 23-25, 1995}, pages
  198--208. {ACM} Press, 1995.
\newblock \href {https://doi.org/10.1145/199448.199483}
  {\path{doi:10.1145/199448.199483}}.

\bibitem{DBLP:journals/tcs/Mackie00}
Ian Mackie.
\newblock Interaction nets for linear logic.
\newblock {\em Theor. Comput. Sci.}, 247(1-2):83--140, 2000.
\newblock \href {https://doi.org/10.1016/S0304-3975(00)00198-5}
  {\path{doi:10.1016/S0304-3975(00)00198-5}}.

\bibitem{DBLP:journals/iandc/MackieP02}
Ian Mackie and Jorge~Sousa Pinto.
\newblock Encoding linear logic with interaction combinators.
\newblock {\em Inf. Comput.}, 176(2):153--186, 2002.
\newblock \href {https://doi.org/10.1006/INCO.2002.3163}
  {\path{doi:10.1006/INCO.2002.3163}}.

\bibitem{DBLP:journals/entcs/MackieS08}
Ian Mackie and Shinya Sato.
\newblock A calculus for interaction nets based on the linear chemical abstract
  machine.
\newblock In Vincent Danos and Mariangiola Dezani, editors, {\em Proceedings of
  the Third International Workshop on Developments in Computational Models,
  DCM@ICALP 2007, Wroclaw, Poland, July 15, 2007}, volume 192 of {\em
  Electronic Notes in Theoretical Computer Science}, pages 59--70. Elsevier,
  2007.
\newblock \href {https://doi.org/10.1016/J.ENTCS.2008.10.027}
  {\path{doi:10.1016/J.ENTCS.2008.10.027}}.

\bibitem{DBLP:conf/tlca/MikamiA99}
Seikoh Mikami and Yohji Akama.
\newblock A study of abramsky's linear chemical abstract machine.
\newblock In Jean{-}Yves Girard, editor, {\em Typed Lambda Calculi and
  Applications, 4th International Conference, TLCA'99, L'Aquila, Italy, April
  7-9, 1999, Proceedings}, volume 1581 of {\em Lecture Notes in Computer
  Science}, pages 243--257. Springer, 1999.
\newblock \href {https://doi.org/10.1007/3-540-48959-2\_18}
  {\path{doi:10.1007/3-540-48959-2\_18}}.

\bibitem{DBLP:journals/entcs/Milner07}
Robin Milner.
\newblock Local bigraphs and confluence: Two conjectures (extended abstract).
\newblock In Roberto~M. Amadio and Iain Phillips, editors, {\em Proceedings of
  the 13th International Workshop on Expressiveness in Concurrency, {EXPRESS}
  2006, Bonn, Germany, August 26, 2006}, volume 175 of {\em Electronic Notes in
  Theoretical Computer Science}, pages 65--73. Elsevier, 2006.
\newblock \href {https://doi.org/10.1016/J.ENTCS.2006.07.035}
  {\path{doi:10.1016/J.ENTCS.2006.07.035}}.

\bibitem{DBLP:conf/birthday/SandsGM02}
David Sands, J{\"{o}}rgen Gustavsson, and Andrew Moran.
\newblock Lambda calculi and linear speedups.
\newblock In Torben~{\AE}. Mogensen, David~A. Schmidt, and Ivan~Hal Sudborough,
  editors, {\em The Essence of Computation, Complexity, Analysis,
  Transformation. Essays Dedicated to Neil D. Jones [on occasion of his 60th
  birthday]}, volume 2566 of {\em Lecture Notes in Computer Science}, pages
  60--84. Springer, 2002.
\newblock \href {https://doi.org/10.1007/3-540-36377-7\_4}
  {\path{doi:10.1007/3-540-36377-7\_4}}.

\bibitem{DBLP:journals/entcs/SatoSY02}
Shinya Sato, Toru Sugimoto, and Shinichi Yamada.
\newblock An implementation model of the typed lambda-calculus based on linear
  chemical abstract machine.
\newblock In Michael Hanus, editor, {\em International Workshop on Functional
  and (Constraint) Logic Programming, {WFLP} 2001, Kiel, Germany, September
  13-15, 2001, Selected Papers}, volume~64 of {\em Electronic Notes in
  Theoretical Computer Science}, pages 292--307. Elsevier, 2001.
\newblock \href {https://doi.org/10.1016/S1571-0661(04)80356-6}
  {\path{doi:10.1016/S1571-0661(04)80356-6}}.

\bibitem{Terese}
Terese.
\newblock {\em {Term Rewriting Systems}}, volume~55 of {\em {Cambridge Tracts
  in Theoretical Computer Science}}.
\newblock Cambridge University Press, 2003.

\bibitem{DBLP:journals/tcs/TurnerW99}
David~N. Turner and Philip Wadler.
\newblock Operational interpretations of linear logic.
\newblock {\em Theor. Comput. Sci.}, 227(1-2):231--248, 1999.
\newblock \href {https://doi.org/10.1016/S0304-3975(99)00054-7}
  {\path{doi:10.1016/S0304-3975(99)00054-7}}.

\bibitem{DBLP:phd/hal/Vanoni22}
Gabriele Vanoni.
\newblock {\em On Reasonable Space and Time Cost Models for the
  {\(\lambda\)}-Calculus. (Sur les mod{\`{e}}les de co{\^{u}}t raisonnable en
  espace et en temps pour le {\(\lambda\)}-calcul)}.
\newblock PhD thesis, University of Bologna, Italy, 2022.
\newblock URL: \url{https://tel.archives-ouvertes.fr/tel-03923206}.

\end{thebibliography}

\techrep{
\newpage
\appendix
\setboolean{appendix}{true}
\section*{Proof Appendix}
\section{Proofs and Definitions Omitted from \refsect{esc} (Exponential Substitution Calculus)}
\label{app:esc}

\begin{definition}[Proper terms]
Proper terms are defined by induction on $\tm$ as follows:\label{def:proper-terms}
\begin{itemize}
\item \emph{Variables}: $\mvar$ and $\evar$ are proper.
\item \emph{Implication}: $\la\var\tm$ is proper if $\tm$ is proper and if $\var$ is a multiplicative variable then $\var \in\mfv\tm$.
\item \emph{Subtraction}: $\suba\mvar\val\var \tm$ is proper if $\tm$ is proper, $\val$ is proper, $\mfv\val \cap (\mfv\tm\setminus\set{\var}) = \emptyset$, $\mvar \notin\mfv\tm$, and if $\var$ is a multiplicative variable then $\var \in\mfv\tm$.
\item \emph{Bang}: $\bang\tm$ is proper if $\tm$ is proper and $\mfv\tm = \emptyset$.
\item \emph{Dereliction}: $\dera\evar\var \tm$ is proper if $\tm$ is proper and if $\var$ is a multiplicative variable then $\var \in\mfv\tm$.
\item \emph{Cut}: $\cuta\val\var\tm$ is proper if $\tm$ is proper, $\val$ is proper, $\mfv\val \cap (\mfv\tm\setminus\set\var) = \emptyset$, and if $\var$ is a multiplicative variable then $\var \in\mfv\tm$.
\end{itemize}
\end{definition}

\gettoappendix{l:weak-harmony}

\gettoappendix{l:basic-harmony}
\begin{proof}
Cases of $\tm$:\applabel{l:basic-harmony}
\begin{itemize}
\item \emph{Abstraction, promotion}: $\tm$ is an answer and has no basic redexes.

\item \emph{Par, subtraction, dereliction, variable}: $\tm$ is open, against hypotheses.

\item \emph{Cut}: then $\tm = \cuta\valtwo\var\lctxp\val$. Two cases:
\begin{itemize}
\item \emph{$\lctx$ is a cut context or empty}. Two cases:
\begin{itemize}
\item \emph{$\val$ is a variable}. By closure, $\val = \var$ or it is captured by $\lctx$. In both cases $\tm$ has a basic non-erasing redex, and it is not an answer because $\var$ is captured.
\item \emph{$\val$ is not a variable}. Then $\tm$ is an answer and it has no non-erasing basic redex.
\end{itemize}

\item \emph{$\lctx$ is non-empty and not a cut context}. Then $\tm$ is not an answer, and $\lctx = \ectxp{\lctxtwo}$ with $\ectx$ cut context and $\lctxtwo$ starting with a par, a subtraction, or a dereliction on a variable $\varthree$. Since $\tm$ is closed, $\varthree$ is bound in $\cuta\valtwo\var\ectx$, and so there is a basic non-erasing redex.\qedhere
\end{itemize}
\end{itemize}
\end{proof}

\section{Adding Tensors}
\label{app:tensors}
The presentation of ESC in \cite{DBLP:journals/lmcs/Accattoli23} includes two constructors $\para\mvar\var\vartwo \tm$  and $\pair\tm\tmtwo$ corresponding to the left and right sequent calculus rules of tensor. The rules decorated with proof terms follow (as explained in the previous section of the Appendix, $\multiForm\#\multiFormtwo$ means that the domains of $\multiForm$ and $\multiFormtwo$ are disjoint):
\begin{center}
\begin{tblr}{c}
	\AxiomC{$  \multiForm, \var\hastype\form, \vartwo\hastype\formtwo \vdash \tm\hastype\formthree$}
	\AxiomC{$\mvar$ fresh}
	\RightLabel{$ \tensLeftRule $}
	\BinaryInfC{$  \multiForm, \mvar\hastype\form \tens \formtwo \vdash \para\mvar\var\vartwo \tm \hastype \formthree$}	
	\DisplayProof
	&
		\AxiomC{$  \multiForm\vdash \tm\hastype\form$}
 	\AxiomC{$  \multiFormtwo\vdash \tmtwo\hastype\formtwo$}
	\AxiomC{$\multiForm\#\multiFormtwo$}
 	\RightLabel{$ \tensRightRule $}
 	\TrinaryInfC{$  \multiForm, \multiFormtwo\vdash \pair\tm\tmtwo \hastype\form\otimes \formtwo$}
	\DisplayProof
\end{tblr}
\end{center}
The associated rewriting rule is:
\begin{center}
$\begin{array}{rcl}
\cuta{\pair\tmtwo\tmthree}\mvar \ctxp{\para\mvar\var\vartwo\tm}
& \rtotens & 
\ctxp{\lctxp{\cuta\val\var \lctxtwop{\cuta\valtwo\vartwo \tm}}} 
\\[4pt]&& \mbox{with }\tmtwo = \lctxp{\val}\mbox{ and }\tmthree = \lctxtwop{\valtwo}
\end{array}$
\end{center}
SESAME then needs to be extended with the following transitions:
\begin{center}
\scriptsize
\begin{tblr}{c}
$\begin{tblr}{|r|r||c||r|r|l}
\cline{1-5}
\SetCell[c=1]{c} \textsc{MCtx} & \SetCell[c=1]{c} \textsc{Pool} &\textsc{Tran.}& \SetCell[c=1]{c} \textsc{MCtx} & \SetCell[c=1]{c}\textsc{Pool}
\\\cline{1-5}
\namctx \mbox{ with }\mvar\notin\domp\namctx & \nctxholep{\para\mvar\var\vartwo \tm}\name \cons \pool 
&\tomachseaseven&
\namctx  \nctxholep{\para\mvar\var\vartwo \nctxhole\name}\name &\nctxholep{\tm}\name \cons \pool
\\[2pt]
\namctx & \nctxholep{\pair\tm\tmtwo}\name \cons \pool  
&\tomachseaeight&
\namctx  \namctxholep{\pair{\nctxhole\name}{\nctxhole\nametwo}}\name &\nctxholep{\tm}\name \cons \nctxholep{\tmtwo}\nametwo  \cons \pool &*
\\[4pt]
\cline[dashed]{1-5}
\namctxp{\cuta\mvartwo\mvar\namctxtwo}\mute & \nctxholep{\para\mvar\var\vartwo\tm}\name \cons \pool    
&\tomachaxmtwob&
\namctxp\namctxtwo\mute & \nctxholep{\para\mvartwo\var\vartwo\tm}\name \cons \pool
\\[2pt]
\namctxp{\cuta{\pair\tmtwo\tmthree}\mvar\namctxtwo}\mute & \nctxholep{\para\mvar\var\vartwo\tm}\name \cons \pool    
&\tomachhole{\tens}&
\namctxp\namctxtwo\mute &\nctxholep{\lctxp{\cuta\val\var \lctxtwop{\cuta\valtwo\vartwo \tm}}}\name \cons \pool
&\#
\\\cline{1-5}
\end{tblr}$
\\[4pt]
\SetCell[c=1]{l}* $\nametwo$ is fresh. \ \ \ \# with $\tmtwo = \lctxp{\val}$ and $\tmthree = \lctxtwop{\valtwo}$.
\end{tblr}
\end{center}
It is easily seen that, with the assumption about how SESAME is implemented that we made in \refsect{analysis}, these transitions can be implemented in constant time and they do not break the sub-term invariant / property. Everything smoothly adapts as to cover the extension with tensors.
\section{Proofs from \refsect{prel-abs-mach} (Preliminaries on Abstract Machines)}
\label{app:prel-abs-mach}

\begin{lemma}[One-step transfer]
  Let $\mach$ and $\tostrat$ be an asymmetric implementation system.
  For any state $\state$ of $\mach$, if $\decode\state \tostrat \tmtwo$ then there is a state $\statetwo$ of $\mach$ 
such that $\state \tomachsea^*\tomachpr \statetwo$.  \label{l:one-step-transfer}
\end{lemma}

\begin{proof}
  For any state $\state$ of $\mach$, let $\nfsea{\state}$ be a normal form of $\state$ with respect to $\tomachsea$: such a state exists because search transitions 
terminate.
  Since $\tomachsea$ is mapped on identities (by search transparency), one has $\decode{\nfsea{\state}} = 
\decode\state$.
  As $\decode\state$ is not $\tostrat$-normal by hypothesis, the halt property entails that 
$\nfsea{\state}$ is not final, therefore $\state \tomachsea^* \nfsea{\state} \tomachpr \statetwo$ for some state $\statetwo$. 
\end{proof}

\gettoappendix{thm:abs-impl}

\begin{proof}
\applabel{thm:abs-impl}\hfill
\begin{enumerate}
\item \emph{Runs to evaluations}: by induction on the length of the run. It follows easily from the properties of principal 
projection and search transparency of lax implementation systems.

\item \emph{Normalizing evaluations to runs}: the statement that we need to prove is:
\begin{center}
if $\deriv \colon \tm \tostrat^* \tmtwo$ with $\tmtwo$  $\tostrat$-normal then there is a $\mach$-run $\run \colon \compilrel\tm\statetwo \tomachine^* \state$ with $\state$ final such that $\decode\state = \tmtwo$ with $\sizepr\run = \size{\deriv}$.
\end{center}
We prove a more general statement where we replace the initial term $\tm$ and the initial state $\compilrel\tm\statetwo$ with 
a general state $\statetwo$ and use its decoding $\decode\statetwo$ for $\tm$: 
\begin{center}
if $\deriv:\decode\statetwo \tostrat^* \tmtwo$ with $\tmtwo$  
$\tostrat$-normal then there is a $\mach$-run $\run: \statetwo \tomach^* \state$ with $\state$ final such that 
$\decode\state = \tmtwo$ with $\sizepr\run = \size\deriv$.
\end{center}
 Then if we instantiate this more general statement on 
initial states we obtain the official statement, because by the initialization constraint of 
the machine if $\compilrel\tm\statetwo$ then $\decode{\statetwo} =_\alpha \tm$.

The proof of the generalized statement is by induction on $\size\deriv$. If $\size\deriv=0$ then consider 
$\nfsea{\statetwo}$, that by search transparency satisfies $\decode{\nfsea{\statetwo}} = 
\decode{\statetwo}$. Now, if $\nfsea{\statetwo}$ can do a principal transition then from $\decode{\statetwo}$ it is possible to do a $\tostrat$ step by  principal projection, which is absurd, because $\size\deriv=0$ and by the diamond property all evaluations sequences from a term have the same number and kind of steps. Then, $\nfsea{\statetwo}$ is a final state and there is a run $\run: \statetwo \tomachsea^*\nfsea{\statetwo}$ such 
that $\sizepr\run= 0$. The statement holds with $\state \defeq \nfsea{\statetwo}$.

If $\size\deriv=k>0$ then $\decode\statetwo \tostrat \tmthree \tostrat^{k-1} \tmtwo$ for some 
$\tmthree$. By the one-step transfer lemma (\reflemma{one-step-transfer}), we obtain $\statetwo\tomachsea^*\tomachpr 
\statethree$ for some $\statethree$. By search transparency and principal projection, we obtain 
$\deriv':\decode\statetwo\tostrat\decode\statethree$. If $\tmthree = \decode\statethree$ then we have and evaluation $\derivtwo:\decode\statethree \tostrat^{k-1} \tmtwo$. If $\tmthree \neq \decode\statethree$, the diamond property of $\tostrat$ gives anyway an 
evaluation $\derivtwo:\decode\statethree \tostrat^{k-1} \tmtwo$. We can then apply the \ih, obtaining a run $\runtwo:\statethree \tomach^* \state$ with $\state$ 
final and such that $\decode{\state} = \tmtwo$, and with $\sizepr\runtwo = k-1$. Note that 
the run $\run:\statetwo \tomachsea^*\tomachpr \statethree\tomach^* \state$ obtained by prefixing $\runtwo$ with the run given by the one-step transfer lemma satisfies the statement because $\sizepr\run = \sizepr\runtwo +1 = \size\derivtwo +1 = \size\deriv$.

\item \emph{Diverging evaluations to runs}: suppose that $\tostrat$ diverges on $\tm$ but $\mach$ terminates, that 
is, that there is a run $\run:\compilrel\tm\statetwo \tomach^* \state$ with $\state$ final. Then the projection 
$\decode\run:\tm \tostrat^*  \decode\state$ given by the \emph{run to evaluations} part (of this theorem) is a normalizing sequence 
by the halt property (guaranteeing that $\decode\state$ is $\tostrat$-normal). Then $\tostrat$ normalizes $\tm$ and so, by the diamond property (precisely by \emph{uniform 
normalization} implied by the diamond property, see \refapp{esc}), $\tostrat$ cannot diverge on 
$\tm$---absurd. Therefore $\mach$ diverges on $\statetwo$. Now, since $\tomachsea$ terminates, the diverging run 
from $\statetwo$ must have infinitely many principal transitions.\qedhere
 
\end{enumerate}
\end{proof}
\section{Proofs Omitted from \refsect{BAM} (Basic Abstract Machine)}

\gettoappendix{l:bam-invariants} %

\begin{proof}
\applabel{l:bam-invariants}
By induction on the length of the run $\run$ leading from the initial state to $\state$. If the run is empty the invariants hold because the initial term is closed and well-bound by hypotheses. If the run is non-empty, the invariants follow very easily from the \ih and the inspection of the transitions in \reffig{BAM}.\qedhere
\end{proof}

\gettoappendix{prop:bam-properties}

\begin{proof}
\applabel{prop:bam-properties}\hfill
\begin{enumerate}
\item If $\state = \twostate\ectx{\cuta\val\var \tm} \tomachsea \twostate{\ectxp{\cuta\val\var\ctxhole} }\tm = \statetwo$ then $\decode \statetwo = \ectxp{\cuta\val\var\ctxhole}\ctxholep\tm = \ectxp{\cuta\val\var\ctxholep\tm} = \ectxp{\cuta\val\var\tm} = \decode\state$.
\item Straightforward. We show it for one transition, they go all in the same way. If $\state = \twostate{ \ectxp{\cuta{\mvartwo}\mvar\ectxtwo} }{ \suba\mvar\val\var\tm }
\tomachaxmtwo
\twostate{ \ectxp\ectxtwo }{ \suba\mvartwo\val\var\tm }$. Then $\decode\state = \ectxp{\cuta{\mvartwo}\mvar\ectxtwop{\suba\mvar\val\var\tm}}$. By the well-bound invariant (\reflemma{bam-invariants}.2), there are no cuts on $\mvar$ in $\ectxtwo$ and $\ectx$. Then 
\begin{center}
$\ectxp{\cuta{\mvartwo}\mvar\ectxtwop{\suba\mvar\val\var\tm}} \toaxmtwo \ectxp{\ectxtwop{\suba\mvartwo\val\var\tm}} = \decode\statetwo$
\end{center}
which clearly is a basic step.
\item Straightforward, since there can only be a finite number of cuts in the active term.
\item Let $\state = \twostate\ectx\tm$ be a final state. Then the outermost constructor of $\tm$ is not a cut, otherwise $\tomachsea$ would apply. By the closure invariant (\reflemma{bam-invariants}.1), the outermost constructor is not a subtraction, nor a dereliction, nor a variable, otherwise there would be a cut on the variable occurrence of that constructor, and by clash-freeness (see the last paragraph of \refsect{prel-abs-mach}) the machine would make a transition. Then $\tm$ is an abstraction or a bang. Then $\decode\state = \ectxp\tm$ is an answer. By the closure invariant, $\ectxp\tm$ is a closed term. By \reflemma{weak-harmony}, $\ectxp\tm$ is normal for basic evaluation.\qedhere
\end{enumerate}
\end{proof}
\section{Definitions omitted from \refsect{sesame} (SESAME)}
\label{app:sesame}
\begin{definition}[Domain of multi contexts]
\label{def:namctx-domain}
The domain of a multi-context is the set of variables $\domp\namctx$ on which there is an out cut $\cuta\namvctx\var \namctxtwo$ in $\namctx$ such that there are some holes in $\namctxtwo$. Formally, it is defined as follows:
\begin{center}
\begin{tabular}{ccc}
$
\begin{array}{rll@{\hspace{.4cm}}@{\hspace{.4cm}} rll}
\multicolumn{6}{c}{\textsc{Domain of value multi contexts}}
\\
\domp\namctxholea & \defeq & \emptyset
&
\domp\var & \defeq & \emptyset
\\
\domp{\la\var\namctx}  & \defeq & \domp\namctx
&
\domp{\bang\namctx}  & \defeq & \domp\namctx
\end{array}$
\\\\
$\begin{array}{rlll}
\multicolumn{4}{c}{\textsc{Domain of multi contexts}}
\\
\domp{\cuta\namvctx\var \namctx}  & \defeq & \set\var\cup\domp\namctx & \mbox{if $\namctx$ is not a term}
\\
\domp{\cuta\namvctx\var \namctx}  & \defeq & \emptyset & \mbox{if $\namctx$ is a term}
\\
\domp{ \suba\mvar\namvctx\var\namctx}  & \defeq & \domp\namvctx\cup\domp\namctx
\\
\domp{\dera\evar\var\namctx}  & \defeq & \domp\namctx
\end{array}$
\end{tabular}
  \end{center}
\end{definition}


\begin{definition}[Well bound states]
\label{def:well-bound-state}
A SESAME state $\state=\twostate\namctx\pool$ is \emph{well bound} if:
\begin{itemize}
\item \emph{Term binders}:  if $\la\var\tm$, $\dera\evar\var\tm$, $\suba\mvar\val\var\tm$, or $\cuta\val\var\tm$ occur anywhere in $\state$ then the only other possible occurrences of $\var$ in $\state$ are as a free variable of $\tm$;

\item \emph{Context binders}: if $\namctx=\namctxtwop{\cuta\val\var\namctxthree}\mute$, $\namctx=\namctxtwop{\suba\mvar\val\var\namctxthree}\mute$, $\namctx=\namctxtwop{\dera\evar\var\namctxthree}\mute$, or $\namctx=\namctxtwop{\la\var\namctxthree}\mute$ then the only other possible occurrences of $\var$ in $\state$ are as a free variable of $\namctxthree$ or of jobs $\nctxholep{\tm}\name$ of the pool $\pool$ such that $\name\in\names{\namctxthree}$.
\end{itemize}
\end{definition}

\section{Proofs Omitted from \refsect{analysis} (Analysis of SESAME)}

\begin{lemma}[Inner extension of good contexts]
Let $\gctx$ be a good context. Then: \label{l:good-ctx-inner-extension}
\begin{enumerate}
\item \label{p:good-ctx-inner-extension-one}
$\gctxp{\cuta\val\var\ctxhole}$, $\gctxp{\suba\mvar\val\var\ctxhole}$, $\gctxp{\dera\evar\var\ctxhole}$, $\gctxp{\la\var\ctxhole}$, and $\gctxp{\bang\ctxhole}$ are good contexts.
\item \label{p:good-ctx-inner-extension-two}
If $\mvar \notin\domp\gctx$ then $\gctxp{\suba\mvar\ctxhole\var\tm}$ is a good context.
\end{enumerate}
\end{lemma}
\begin{proof}
The proofs are straightforward inductions on $\gctx$. Intuitively, for the first point the statement follows from the fact that the inner extensions  of $\gctx$ have empty sets of dominating free variables (and the hole is not in the left sub-term of a cut); the second point has $\mvar$ as dominating variable, which is why the hypothesis $\mvar \notin\domp\gctx$ is needed.\qedhere
\end{proof}

\begin{lemma}[Linear pruning of good contexts]
Let $\ctxp{\cuta\mval\mvar\ctxtwo}$ be a good context. Then $\ctxp\ctxtwo$ is a good context.\label{l:linear-pruning-good-ctxs}
\end{lemma}

\begin{proof}
The proof is a straightforward induction on $\ctx$. Intuitively, removing a cut makes goodness easier to establish, because the constraints on goodness are related to cuts.\qedhere
\end{proof}

\begin{lemma}[Linear pruning of approximants]
Let $\namctxp{\cuta\mval\mvar\namctxtwo}\mute$ be an approximant. Then $\namctxp{\namctxtwo}\mute$ is an approximant.\label{l:linear-pruning-approximants}
\end{lemma}

\begin{proof}
The removal of $\cuta\mval\mvar$ cannot alter the uniqueness of names, nor add holes, nor create occurrences of out variables, thus $\namctxp{\namctxtwo}\mute$ is an approximant because $\namctxp{\cuta\mval\mvar\namctxtwo}\mute$ is.\qedhere
\end{proof}

\gettoappendix{l:sesam-invariants} %

\begin{proof}
\applabel{l:sesam-invariants}
By induction on the length of the run $\run$ leading from the initial state to $\state$. If the run is empty then the invariants straightforwardly hold, because $\namctxholea$ is an approximant, pools of initial states have only one job named $\name$, the initial term is well-bound, and $\namctxholea$ is a good context. If the run is non-empty each point follows from the \ih plus a straightforward inspection of the SESAM transitions, as detailed below.

We define some notations for the contextual decoding invariant. Firstly, note, since consecutive pluggings on distinct names in a multi-context commute, we have:
\begin{center}
$\namctx^\state_{\tmtwo_1,\ldots,\tmtwo_{i-1}|\name_i|\tmtwo_{i+1},\ldots,\tmtwo_k} = \namctx\namctxholep{\tmtwo_1}{\name_1}\ldots\namctxholep{\tmtwo_{i-1}}{\name_{i-1}}\namctxholep{\tmtwo_{i+1}}{\name_{i+1}}\ldots\namctxholep{\tmtwo_k}{\name_k} \namctxholep{\ctxhole}{\name_{i}}$
\end{center} 
Then, we define the following pool:
\begin{center}
$\pool^\state_{\tmtwo_1,\ldots,\tmtwo_{i-1}|\name_i|\tmtwo_{i+1},\ldots,\tmtwo_k} \defeq \namctxholep{\tmtwo_1}{\name_1}\cons\ldots\cons\namctxholep{\tmtwo_{i-1}}{\name_{i-1}}\cons\namctxholep{\tmtwo_{i+1}}{\name_{i+1}}\cons\ldots\cons\namctxholep{\tmtwo_k}{\name_k}$
\end{center} 
So that $\namctx^\state_{\tmtwo_1,\ldots,\tmtwo_{i-1}|\name_i|\tmtwo_{i+1},\ldots,\tmtwo_k} = \namctx\ctxholep{\pool^\state_{\tmtwo_1,\ldots,\tmtwo_{i-1}|\name_i|\tmtwo_{i+1},\ldots,\tmtwo_k}}\namctxholep\ctxhole{\name_i}$.

Now, cases of the last transition $\statetwo \tosesame \state$:
\begin{itemize}
\item $\statetwo = \twostatetrans \namctx { \nctxholep{\cuta\val\var \tm}\name \cons \pooltwo } 
\tomachseaone
{ \namctx  \nctxholep{\cuta\val\var \nctxhole\name}\name } { \nctxholep{\tm}\name \cons \pooltwo }=\state$
\begin{enumerate}
\item \emph{Approximant}: $\namctx$ is an approximant by \ih then $\namctxtwo \defeq \namctx  \nctxholep{\cuta\val\var \nctxhole\name}\name$ is an approximant, as we now argue. About names, $\name$ is removed and immediately re-added, so it preserves its unique occurrence. About out cuts: $\namctxtwo$ does not add holes in cut values, nor out variable occurrences, and in particular, the newly cut variable $\var$ has no occurrences in $\namctxtwo$.
\item \emph{Names}: it follows by the \ih
\item \emph{Binders}: it follows by the \ih In particular, the property for $\var$ in $\namctx  \nctxholep{\cuta\val\var \nctxhole\name}\name$ follows from the property for $\cuta\val\var\tm$ in $\statetwo$ given by the \ih
\item \emph{Contextual decoding}: for all names $\name_i\in\fn\pool$ different from $\name$, the statement follows from the \ih For $\name$, we have that:
\begin{center}
$\begin{array}{lllll}
\namctx^\state_{|\name|\tmtwo_2,\ldots,\tmtwo_k}
& = & (\namctx  \nctxholep{\cuta\val\var \nctxhole\name}\name)\ctxholep{\pool^\state_{|\name|\tmtwo_2,\ldots,\tmtwo_k}}\namctxholep\ctxhole\name
\\
& = & \namctx\ctxholep{\pool^\state_{|\name|\tmtwo_2,\ldots,\tmtwo_k}} \nctxholep{\cuta\val\var \ctxhole}\name
\\
& = & \namctx\ctxholep{\pool^{\statetwo}_{|\name|\tmtwo_2,\ldots,\tmtwo_k}} \nctxholep{\cuta\val\var \ctxhole}\name
\\
& = & \namctx^{\statetwo}_{|\name|\tmtwo_2,\ldots,\tmtwo_k} \nctxholep{\cuta\val\var \ctxhole}\name
\end{array}$
\end{center}
By \ih, $\namctx^{\statetwo}_{|\name|\tmtwo_2,\ldots,\tmtwo_k}$ is a good context. By \reflemmap{good-ctx-inner-extension}{one},\\ $\namctx^{\statetwo}_{|\name|\tmtwo_2,\ldots,\tmtwo_k}\nctxholep{\cuta\val\var \ctxhole}\name$ is a good context.
\end{enumerate}

\item $\statetwo = \twostatetrans \namctx { \nctxholep{\suba\mvar\val\var \tm}\name \cons \pool }
\tomachseatwo
{ \namctx\nctxholep{\suba\mvar{\nctxhole\nametwo}\var \nctxhole\name}\name } { \nctxholep{\val}\nametwo \cons \nctxholep{\tm}\name  \cons\pool } =\state$ if $\mvar\notin\domp\namctx$ and $\nametwo$ is fresh.
\begin{enumerate}
\item \emph{Approximant}: it follows from the $\namctx$ being an approximant by \ih, since the extension with $\suba\mvar{\nctxhole\nametwo}\var \nctxhole\name$ does not add occurrence of out variables, because by hypothesis $\mvar\notin\domp\namctx$, and it preserves the uniqueness of names, since $\name$ is recycled and $\nametwo$ is fresh by hypothesis.
\item \emph{Names}: it follows from the \ih and freshness of $\nametwo$.
\item \emph{Binders}: it follows from the \ih
\item \emph{Contextual decoding}: for all names $\name_i\in\fn\pool$ different from $\name$ and $\nametwo$, the statement follows from the \ih For $\name$, we have that:
\begin{center}
$\begin{array}{lllll}
\namctx^\state_{\valtwo |\name|\tmtwo_2,\ldots,\tmtwo_k}
& = & (\namctx\nctxholep{\suba\mvar{\nctxhole\nametwo}\var \nctxhole\name}\name)\ctxholep{\pool^\state_{\valtwo |\name|\tmtwo_2,\ldots,\tmtwo_k}}\namctxholep\ctxhole\name
\\
& = & \namctx\ctxholep{\pool^{\statetwo}_{|\name|\tmtwo_2,\ldots,\tmtwo_k}} \nctxholep{\suba\mvar\valtwo\var \ctxhole}\name
\\
& = & \namctx^{\statetwo}_{|\name|\tmtwo_2,\ldots,\tmtwo_k} \nctxholep{\suba\mvar\valtwo\var \ctxhole}\name
\end{array}$
\end{center}
By \ih, $\namctx^{\statetwo}_{|\name|\tmtwo_2,\ldots,\tmtwo_k}$ is a good context. By \reflemmap{good-ctx-inner-extension}{one}, $\namctx^{\statetwo}_{|\name|\tmtwo_2,\ldots,\tmtwo_k}\nctxholep{\suba\mvar\valtwo\var \ctxhole}\name$ is a good context.

Similarly, for $\nametwo$ we obtain that $\namctx^\state_{|\nametwo|\tmtwo_1,\tmtwo_2,\ldots,\tmtwo_k}  = \namctx^{\statetwo}_{|\name|\tmtwo_2,\ldots,\tmtwo_k} \nctxholep{\suba\mvar\ctxhole\var \tmtwo_1}\name$, which is a good context by \reflemmap{good-ctx-inner-extension}{two}, which can be applied because:
\begin{enumerate}
\item $\mvar\notin\domp\namctx$ by hypothesis;
\item That implies that $\mvar\notin\domp{\namctx^{\statetwo}_{|\name|\tmtwo_2,\ldots,\tmtwo_k}}$ because plugging terms in a multi context does not extend its dominion.
\end{enumerate}
\end{enumerate}

\item $\statetwo = \twostatetrans \namctx { \nctxholep{\dera\evar\var\tm}\name \cons \pool }
\tomachseathree
 { \namctx   \nctxholep{\dera\evar\var \nctxhole\name}\name } { \nctxholep{\tm}\name \cons \pool }
 =\state$  if $\evar\notin\domp\namctx$.
\begin{enumerate}
\item \emph{Approximant}: it follows from the $\namctx$ being an approximant by \ih, since the extension with $\dera\evar\var \nctxhole\name$ does not add occurrence of out variables, because by hypothesis $\evar\notin\domp\namctx$, and it preserves the uniqueness of names, since $\name$ is recycled.
\item \emph{Names}: it follows from the \ih
\item \emph{Binders}: it follows from the \ih
\item \emph{Contextual decoding}: for all names $\name_i\in\fn\pool$ different from $\name$, the statement follows from the \ih For $\name$ the reasoning is essentially as in the cut case, we have that:
\begin{center}
$\begin{array}{lllll}
\namctx^\state_{|\name|\tmtwo_2,\ldots,\tmtwo_k}
& = & (\namctx  \nctxholep{\dera\evar\var \nctxhole\name}\name)\ctxholep{\pool^\state_{|\name|\tmtwo_2,\ldots,\tmtwo_k}}\namctxholep\ctxhole\name
\\
& = & \namctx\ctxholep{\pool^\state_{|\name|\tmtwo_2,\ldots,\tmtwo_k}} \nctxholep{\dera\evar\var \ctxhole}\name
\\
& = & \namctx\ctxholep{\pool^{\statetwo}_{|\name|\tmtwo_2,\ldots,\tmtwo_k}} \nctxholep{\dera\evar\var \ctxhole}\name
\\
& = & \namctx^{\statetwo}_{|\name|\tmtwo_2,\ldots,\tmtwo_k} \nctxholep{\dera\evar\var \ctxhole}\name
\end{array}$
\end{center}
By \ih, $\namctx^{\statetwo}_{|\name|\tmtwo_2,\ldots,\tmtwo_k}$ is a good context. By \reflemmap{good-ctx-inner-extension}{one},\\ $\namctx^{\statetwo}_{|\name|\tmtwo_2,\ldots,\tmtwo_k}\nctxholep{\dera\evar\var \ctxhole}\name$ is a good context.
\end{enumerate}

\item $\statetwo = \twostatetrans \namctx { \nctxholep{\la\var\tm}\name \cons \pool  }
\tomachseafour
{ \namctx  \nctxholep{ \la\var\nctxhole\name}\name } { \nctxholep{\tm}\name \cons \pool }
 =\state$.
\begin{enumerate}
\item \emph{Approximant}: it follows from the \ih
\item \emph{Names}: it follows from the \ih
\item \emph{Binders}: it follows from the \ih

\item \emph{Contextual decoding}: for all names $\name_i\in\fn\pool$ different from $\name$, the statement follows from the \ih For $\name$ the reasoning is essentially as in the cut case, we have that:
\begin{center}
$\begin{array}{lllll}
\namctx^\state_{|\name|\tmtwo_2,\ldots,\tmtwo_k}
& = & (\namctx  \nctxholep{\la\var \nctxhole\name}\name)\ctxholep{\pool^\state_{|\name|\tmtwo_2,\ldots,\tmtwo_k}}\namctxholep\ctxhole\name
\\
& = & \namctx\ctxholep{\pool^\state_{|\name|\tmtwo_2,\ldots,\tmtwo_k}} \nctxholep{\la\var \ctxhole}\name
\\
& = & \namctx\ctxholep{\pool^{\statetwo}_{|\name|\tmtwo_2,\ldots,\tmtwo_k}} \nctxholep{\la\var \ctxhole}\name
\\
& = & \namctx^{\statetwo}_{|\name|\tmtwo_2,\ldots,\tmtwo_k} \nctxholep{\la\var \ctxhole}\name
\end{array}$
\end{center}
By \ih, $\namctx^{\statetwo}_{|\name|\tmtwo_2,\ldots,\tmtwo_k}$ is a good context. By \reflemmap{good-ctx-inner-extension}{one},\\ $\namctx^{\statetwo}_{|\name|\tmtwo_2,\ldots,\tmtwo_k}\nctxholep{\la\var \ctxhole}\name$ is a good context.
\end{enumerate}

\item $\statetwo = \twostatetrans \namctx { \nctxholep{\bang\tm}\name \cons \pool  }
\tomachseafive
{ \namctx  \nctxholep{ \bang\nctxhole\name}\name } { \nctxholep{\tm}\name \cons \pool }
 =\state$.
\begin{enumerate}
\item \emph{Approximant}: it follows from the \ih
\item \emph{Names}: it follows from the \ih
\item \emph{Binders}: it follows from the \ih
\item \emph{Contextual decoding}: for all names $\name_i\in\fn\pool$ different from $\name$, the statement follows from the \ih For $\name$ the reasoning is essentially as in the cut case, we have that:
\begin{center}
$\begin{array}{lllll}
\namctx^\state_{|\name|\tmtwo_2,\ldots,\tmtwo_k}
& = & (\namctx  \nctxholep{\bang \nctxhole\name}\name)\ctxholep{\pool^\state_{|\name|\tmtwo_2,\ldots,\tmtwo_k}}\namctxholep\ctxhole\name
\\
& = & \namctx\ctxholep{\pool^\state_{|\name|\tmtwo_2,\ldots,\tmtwo_k}} \nctxholep{\bang \ctxhole}\name
\\
& = & \namctx\ctxholep{\pool^{\statetwo}_{|\name|\tmtwo_2,\ldots,\tmtwo_k}} \nctxholep{\bang \ctxhole}\name
\\
& = & \namctx^{\statetwo}_{|\name|\tmtwo_2,\ldots,\tmtwo_k} \nctxholep{\bang \ctxhole}\name
\end{array}$
\end{center}
By \ih, $\namctx^{\statetwo}_{|\name|\tmtwo_2,\ldots,\tmtwo_k}$ is a good context. By \reflemmap{good-ctx-inner-extension}{one},\\ $\namctx^{\statetwo}_{|\name|\tmtwo_2,\ldots,\tmtwo_k}\nctxholep{\bang \ctxhole}\name$ is a good context.
\end{enumerate}

\item $\statetwo = \twostatetrans \namctx { \nctxholep{\var}\name  \cons \pool }
\tomachseasix
{ \namctx  \nctxholep\var\name }\pool
 =\state$  if $\var\notin\domp\namctx$.
\begin{enumerate}
\item \emph{Approximant}: it follows from the $\namctx$ being an approximant by \ih, since the extension with $\var$ does not add occurrence of out variables, because by hypothesis $\var\notin\domp\namctx$, and it preserves the uniqueness of names, since $\name$ is removed from both the multi context and the pool.
\item \emph{Names}: it follows from the \ih
\item \emph{Binders}: it follows from the \ih
\item \emph{Contextual decoding}: let $\name_i\in\fn\pool$ and note that:
\begin{center}
$\begin{array}{lllll}
\namctx^\state_{\tmtwo_1,\ldots,\tmtwo_{i-1}|\name_i|\tmtwo_{i+1},\ldots,\tmtwo_k}
& = & (\namctx  \nctxholep{\var}\name)\ctxholep{\pool^\state_{\tmtwo_1,\ldots,\tmtwo_{i-1}|\name_i|\tmtwo_{i+1},\ldots,\tmtwo_k}}
\\
& = & \namctx\ctxholep{\pool^{\statetwo}_{\var,\tmtwo_1,\ldots,\tmtwo_{i-1}|\name_i|\tmtwo_{i+1},\ldots,\tmtwo_k}}
\end{array}$
\end{center}
By \ih, $\namctx\ctxholep{\pool^{\statetwo}_{\var,\tmtwo_1,\ldots,\tmtwo_{i-1}|\name_i|\tmtwo_{i+1},\ldots,\tmtwo_k}}$ is a good context.
\end{enumerate}

\item $\statetwo = \twostatetrans { \namctxp{\cuta\mvartwo\mvar\namctxtwo}\mute } { \nctxholep{\suba\mvar\val\var\tm}\name \cons \pool }
\tomachaxmtwo
{ \namctxp\namctxtwo\mute } { \nctxholep{\suba\mvartwo\val\var\tm}\name \cons \pool } = \state$.
\begin{enumerate}
\item \emph{Approximant}: it follows from the \ih and the linear pruning of approximants lemma (\reflemma{linear-pruning-approximants}).
\item \emph{Names}: it follows from the \ih and Point 1: approximants have no holes in cut values, so that removing the cut on $\mvar$ from the approximant $\namctxp{\cuta\mvartwo\mvar\namctxtwo}\mute$ does not change the set of names.
\item \emph{Binders}: it follows from the \ih
\item \emph{Contextual decoding}: by definition, we have:
\begin{center}
$\begin{array}{lllll}
\namctx^\state_{\tmtwo_1,\ldots,\tmtwo_{i-1}|\name_i|\tmtwo_{i+1},\ldots,\tmtwo_k}
& = & (\namctxp\namctxtwo\mute)\ctxholep{\pool^\state_{\tmtwo_1,\ldots,\tmtwo_{i-1}|\name_i|\tmtwo_{i+1},\ldots,\tmtwo_k}}
\\
\mbox{and}
\\
\namctx^{\statetwo}_{\tmtwo_1,\ldots,\tmtwo_{i-1}|\name_i|\tmtwo_{i+1},\ldots,\tmtwo_k}
& = & (\namctxp{\cuta\mvartwo\mvar\namctxtwo}\mute)\ctxholep{\pool^\state_{\tmtwo_1,\ldots,\tmtwo_{i-1}|\name_i|\tmtwo_{i+1},\ldots,\tmtwo_k}}
\end{array}$
\end{center}
By \ih, $\namctx^{\statetwo}_{\tmtwo_1,\ldots,\tmtwo_{i-1}|\name_i|\tmtwo_{i+1},\ldots,\tmtwo_k}$ is a good context. By linear pruning of good contexts (\reflemma{linear-pruning-good-ctxs}), $\namctx^\state_{\tmtwo_1,\ldots,\tmtwo_{i-1}|\name_i|\tmtwo_{i+1},\ldots,\tmtwo_k}$ is a good context.
\end{enumerate}

\item $\statetwo = \twostatetrans { \namctxp{\cuta{\la\vartwo\lctxp\valtwo}\mvar\namctxtwo}\mute } { \nctxholep{\suba\mvar\val\var \tm}\name \cons \pool }
\tomachlolli
{ \namctxp\namctxtwo\mute } { \nctxholep{\cuta\val\vartwo\lctxp{\cuta\valtwo\var\tm}}\name \cons \pool } = \state$.
\begin{enumerate}
\item \emph{Approximant}: it follows from the \ih and the linear pruning of approximants lemma (\reflemma{linear-pruning-approximants}).
\item \emph{Names}: it follows from the \ih and Point 1: approximants have no holes in cut values, so that removing the cut on $\mvar$ from the approximant $\namctxp{\cuta{\la\vartwo\lctxp\valtwo}\mvar\namctxtwo}\mute$ does not change the set of names.
\item \emph{Binders}: it follows from the \ih
\item \emph{Contextual decoding}: it goes as in the $\tomachaxmtwo$ case. By definition, we have:
\begin{center}
$\begin{array}{lllll}
\namctx^\state_{\tmtwo_1,\ldots,\tmtwo_{i-1}|\name_i|\tmtwo_{i+1},\ldots,\tmtwo_k}
& = & (\namctxp\namctxtwo\mute)\ctxholep{\pool^\state_{\tmtwo_1,\ldots,\tmtwo_{i-1}|\name_i|\tmtwo_{i+1},\ldots,\tmtwo_k}}
\\
\mbox{and}
\\
\namctx^{\statetwo}_{\tmtwo_1,\ldots,\tmtwo_{i-1}|\name_i|\tmtwo_{i+1},\ldots,\tmtwo_k}
& = & (\namctxp{\cuta{\la\vartwo\lctxp\valtwo}\mvar\namctxtwo}\mute)\ctxholep{\pool^\state_{\tmtwo_1,\ldots,\tmtwo_{i-1}|\name_i|\tmtwo_{i+1},\ldots,\tmtwo_k}}
\end{array}$
\end{center}
By \ih, $\namctx^{\statetwo}_{\tmtwo_1,\ldots,\tmtwo_{i-1}|\name_i|\tmtwo_{i+1},\ldots,\tmtwo_k}$ is a good context. By linear pruning of good contexts (\reflemma{linear-pruning-good-ctxs}), $\namctx^\state_{\tmtwo_1,\ldots,\tmtwo_{i-1}|\name_i|\tmtwo_{i+1},\ldots,\tmtwo_k}$ is a good context.
\end{enumerate}

\item $\statetwo = \twostatetrans { \namctxp{\cuta\evartwo\evar\namctxtwo}\mute } { \nctxholep{\dera\evar\var\tm}\name \cons \pool }
\tomachaxetwo
{ \namctxp{\cuta\evartwo\evar\namctxtwo}\mute } { \nctxholep{\dera\evartwo\var\tm}\name \cons \pool } = \state$.
\begin{enumerate}
\item \emph{Approximant}: it follows from the \ih
\item \emph{Names}: it follows from the \ih
\item \emph{Binders}: it follows from the \ih
\item \emph{Contextual decoding}: it follows from the \ih
\end{enumerate}

\item $\statetwo = \twostatetrans { \namctxp{\cuta{\bang\lctxp\val}\evar\namctxtwo}\mute } { \nctxholep{\dera\evar\var\tm}\name \cons \pool }
\tomachbang
{ \namctxp{\cuta{\bang\lctxp\val}\evar\namctxtwo}\mute } { \nctxholep{\rename\lctx\ctxholep{\cuta{\rename\val}\var\tm}}\name \cons \pool } = \state$.
\begin{enumerate}
\item \emph{Approximant}: it follows from the \ih
\item \emph{Names}: it follows from the \ih
\item \emph{Binders}: it follows from the \ih and the use of the renaming operation $\rename\cdot$ in the target state.
\item \emph{Contextual decoding}: it follows from the \ih
\end{enumerate}

\item $\statetwo = \twostatetrans { \namctxp{\cuta\mval\mvar\namctxtwo}\mute } { \nctxholep{\mvar}\name \cons \pool }
\tomachaxmone
{ \namctxp\namctxtwo\mute } { \nctxholep{\mval}\name \cons \pool } = \state$.
\begin{enumerate}
\item \emph{Approximant}: it follows from the \ih and the linear pruning of approximants lemma (\reflemma{linear-pruning-approximants}).
\item \emph{Names}: it follows from the \ih and Point 1: approximants have no holes in cut values, so that removing the cut on $\mvar$ from the approximant $\namctxp{\cuta\mval\mvar\namctxtwo}\mute$ does not change the set of names.
\item \emph{Binders}: it follows from the \ih
\item \emph{Contextual decoding}: it goes as in the $\tomachaxmtwo$ case. By definition, we have:
\begin{center}
$\begin{array}{lllll}
\namctx^\state_{\tmtwo_1,\ldots,\tmtwo_{i-1}|\name_i|\tmtwo_{i+1},\ldots,\tmtwo_k}
& = & (\namctxp\namctxtwo\mute)\ctxholep{\pool^\state_{\tmtwo_1,\ldots,\tmtwo_{i-1}|\name_i|\tmtwo_{i+1},\ldots,\tmtwo_k}}
\\
\mbox{and}
\\
\namctx^{\statetwo}_{\tmtwo_1,\ldots,\tmtwo_{i-1}|\name_i|\tmtwo_{i+1},\ldots,\tmtwo_k}
& = & (\namctxp{\cuta\mval\mvar\namctxtwo}\mute)\ctxholep{\pool^\state_{\tmtwo_1,\ldots,\tmtwo_{i-1}|\name_i|\tmtwo_{i+1},\ldots,\tmtwo_k}}
\end{array}$
\end{center}
By \ih, $\namctx^{\statetwo}_{\tmtwo_1,\ldots,\tmtwo_{i-1}|\name_i|\tmtwo_{i+1},\ldots,\tmtwo_k}$ is a good context. By linear pruning of good contexts (\reflemma{linear-pruning-good-ctxs}), $\namctx^\state_{\tmtwo_1,\ldots,\tmtwo_{i-1}|\name_i|\tmtwo_{i+1},\ldots,\tmtwo_k}$ is a good context.
\end{enumerate}

\item $\statetwo = \twostatetrans { \namctxp{\cuta\exval\evar\namctxtwo}\mute } { \nctxholep{\evar}\name \cons \pool }
\tomachaxeone
{ \namctxp{\cuta\exval\evar\namctxtwo}\mute } { \nctxholep{\rename\exval}\name \cons \pool } = \state$.
\begin{enumerate}
\item \emph{Approximant}: it follows from the \ih
\item \emph{Names}: it follows from the \ih
\item \emph{Binders}: it follows from the \ih and the use of the renaming operation $\rename\cdot$ in the target state.
\item \emph{Contextual decoding}: it follows from the \ih\qedhere
\end{enumerate}

\end{itemize}
\end{proof}

\gettoappendix{prop:sesam-properties}

\begin{proof}
\applabel{prop:sesam-properties}\hfill
\begin{enumerate}
\item \emph{Search transparency}. Cases:
\begin{itemize}
\item $\state = \twostatetrans \namctx { \nctxholep{\cuta\val\var \tm}\name \cons \pool } 
\tomachseaone
{ \namctx  \nctxholep{\cuta\val\var \nctxhole\name}\name } { \nctxholep{\tm}\name \cons \pool }=\statetwo$. Then:
\[\begin{array}{c@{\hspace{.4cm}} c@{\hspace{.4cm}} c@{\hspace{.4cm}} c@{\hspace{.4cm}} c@{\hspace{.4cm}} c@{\hspace{.4cm}} c@{\hspace{.4cm}} cccc}
\decode\state & = & (\namctx  \nctxholep{\cuta\val\var \nctxhole\name}\name) \ctxholep\pool\namctxholep\tm\name
& = & \namctx  \nctxholep{\cuta\val\var \tm}\name \ctxholep\pool
& = & \decode\statetwo
\end{array}\]

\item $\state = \twostatetrans \namctx { \nctxholep{\suba\mvar\val\var \tm}\name \cons \pool }
\tomachseatwo
{ \namctx\nctxholep{\suba\mvar{\nctxhole\nametwo}\var \nctxhole\name}\name } { \nctxholep{\val}\nametwo \cons \nctxholep{\tm}\name  \cons\pool } =\statetwo$ if $\mvar\notin\domp\namctx$ and $\nametwo$ is fresh. Then:
\[\begin{array}{c@{\hspace{.4cm}} c@{\hspace{.4cm}} c@{\hspace{.4cm}} c@{\hspace{.4cm}} c@{\hspace{.4cm}} c@{\hspace{.4cm}} c@{\hspace{.4cm}} cccc}
\decode\state & = & (\namctx\nctxholep{\suba\mvar{\nctxhole\nametwo}\var \nctxhole\name}\name) \ctxholep\pool\nctxholep{\val}\nametwo\namctxholep\tm\name
& = & \namctx  \nctxholep{\suba\mvar\val\var \tm}\name \ctxholep\pool
& = & \decode\statetwo
\end{array}\]

\item $\state = \twostatetrans \namctx { \nctxholep{\dera\evar\var\tm}\name \cons \pool }
\tomachseathree
 { \namctx   \nctxholep{\dera\evar\var \nctxhole\name}\name } { \nctxholep{\tm}\name \cons \pool }
 =\statetwo$  if $\evar\notin\domp\namctx$. Then:
\[\begin{array}{c@{\hspace{.4cm}} c@{\hspace{.4cm}} c@{\hspace{.4cm}} c@{\hspace{.4cm}} c@{\hspace{.4cm}} c@{\hspace{.4cm}} c@{\hspace{.4cm}} cccc}
\decode\state & = & (\namctx  \nctxholep{\dera\evar\var \nctxhole\name}\name) \ctxholep\pool\namctxholep\tm\name
& = & \namctx  \nctxholep{\dera\evar\var\tm}\name \ctxholep\pool
& = & \decode\statetwo
\end{array}\]

\item $\state = \twostatetrans \namctx { \nctxholep{\la\var\tm}\name \cons \pool  }
\tomachseafour
{ \namctx  \nctxholep{ \la\var\nctxhole\name}\name } { \nctxholep{\tm}\name \cons \pool }
 =\statetwo$. Then:
\[\begin{array}{c@{\hspace{.4cm}} c@{\hspace{.4cm}} c@{\hspace{.4cm}} c@{\hspace{.4cm}} c@{\hspace{.4cm}} c@{\hspace{.4cm}} c@{\hspace{.4cm}} cccc}
\decode\state & = & (\namctx  \nctxholep{\la\var \nctxhole\name}\name) \ctxholep\pool\namctxholep\tm\name
& = & \namctx  \nctxholep{\la\var\tm}\name \ctxholep\pool
& = & \decode\statetwo
\end{array}\]

\item $\state = \twostatetrans \namctx { \nctxholep{\bang\tm}\name \cons \pool  }
\tomachseafive
{ \namctx  \nctxholep{ \bang\nctxhole\name}\name } { \nctxholep{\tm}\name \cons \pool }
 =\statetwo$. Then:
\[\begin{array}{c@{\hspace{.4cm}} c@{\hspace{.4cm}} c@{\hspace{.4cm}} c@{\hspace{.4cm}} c@{\hspace{.4cm}} c@{\hspace{.4cm}} c@{\hspace{.4cm}} cccc}
\decode\state & = & (\namctx  \nctxholep{\bang \nctxhole\name}\name) \ctxholep\pool\namctxholep\tm\name
& = & \namctx  \nctxholep{\bang\tm}\name \ctxholep\pool
& = & \decode\statetwo
\end{array}\]

\item $\state = \twostatetrans \namctx { \nctxholep{\var}\name  \cons \pool }
\tomachseasix
{ \namctx  \nctxholep\var\name }\pool
 =\statetwo$  if $\var\notin\domp\namctx$. Then:
\[\begin{array}{c@{\hspace{.4cm}} c@{\hspace{.4cm}} c@{\hspace{.4cm}} c@{\hspace{.4cm}} c@{\hspace{.4cm}} c@{\hspace{.4cm}} c@{\hspace{.4cm}} cccc}
\decode\state & = & (\namctx  \nctxholep\var\name) \ctxholep\pool
& = & \decode\statetwo
\end{array}\]
 \end{itemize}
 
\item \emph{Principal projection}. The six principal cases all go in the same way. We show one case.
\begin{itemize}
\item $\state = \twostatetrans { \namctxp{\cuta\mvartwo\mvar\namctxtwo}\mute } { \nctxholep{\suba\mvar\val\var\tm}\name \cons \pool }
\tomachaxmtwo
{ \namctxp\namctxtwo\mute } { \nctxholep{\suba\mvartwo\val\var\tm}\name \cons \pool } = \statetwo$. The read back of $\state$ is:
\[\decode\state = \namctxp{\cuta\mvartwo\mvar\namctxtwo}\mute \ctxholep\pool\nctxholep{\suba\mvar\val\var\tm}\name\]
By the well-bound invariant (\reflemma{sesam-invariants}.3), there is only one binder for $\mvar$ in $\state$---and so in $\decode\state$, and $\name\in\names\namctxtwo$, so that the cut $\cuta\mvartwo\mvar$ can indeed act on $\suba\mvar\val\var\tm$ in $\decode\state$. 
Therefore, a $\toaxmtwo$ step can be applied to $\decode\state$:
\[\begin{array}{lllllll}
\decode\state 
&=& \namctxp{\cuta\mvartwo\mvar\namctxtwo}\mute \ctxholep\pool\nctxholep{\suba\mvar\val\var\tm}\name 
&\toaxmtwo& \namctxp\namctxtwo\mute \ctxholep\pool\nctxholep{\suba\mvartwo\val\var\tm}\name 
&=& \decode\statetwo\end{array}\]
By the contextual decoding invariant (\reflemma{sesam-invariants}.4), the position: 
\[\namctxp{\cuta\mvartwo\mvar\namctxtwo}\mute \ctxholep\pool\nctxholep{\ctxhole}\name\] 
Of that step is a good context, so that it is a $\togp\axmtwo$ step, as required.
\end{itemize}

\item \emph{Search termination}. Straightforward, since the search transitions move constructors from the pool to the multi-context, thus decreasing the number of constructors in the pool.

\item \emph{Halt}.  Let $\state = \twostate\namctx\pool$ be a final state. Note that if the pool is non-empty, that is, $\pool = \namctxholep\tm\name\cons\pooltwo$ then a transition of the SESAME applies, no matter what is the topmost constructor of $\tm$, as the transitions cover all cases. Then the SESAME stops only when the pool is empty, so that $\state = \twostate\namctx\emptylist$.

Now, by the approximant invariant (\reflemma{sesam-invariants}.1), $\namctx$ is an approximant, and by the names invariant (\reflemma{sesam-invariants}.2) the names of $\namctx$ and the names of $\pool=\emptylist$ coincide, that is, there are no named holes in $\namctx$. Then $\namctx$ is a pre-term $\tm$. Since $\namctx$ is an approximant, all its out cuts are garbage. Then $\decode\state=\tm$ is cut-free up to garbage.
\qedhere
\end{enumerate}
\end{proof}

\gettoappendix{l:subterm-invariant}
\begin{proof}
\applabel{l:subterm-invariant}
\hfill
\begin{enumerate}
\item 
The proof is by induction on the length $\size{\run}$ of the run $\run$. If $\size{\run}=0$ the statement trivially holds. If $\size{\run}>0$ then one looks at the last transition $\statetwo \tosesame \twostate\namctx\pool$. The only delicate transition,  which is also the reason why the invariant is stated for all values of $\pool$ (and not just cut ones) is $\tomachlolli$, where the value $\val$ in the subtraction (which is not a cut value) becomes a cut value:
\[
\small
\begin{array}{r|r||c||r|r}
\multicolumn{1}{c|}{\textsc{MCtx}} & \multicolumn{1}{c||}{\textsc{Pool}} &\textsc{Trans.}& \multicolumn{1}{|c}{\textsc{MCtx}} & \multicolumn{1}{|c}{\textsc{Pool}}
\\\hline\hline

\namctxtwop{\cuta{\la\vartwo\lctxp\valtwo}\mvar\namctxthree}\mute & \nctxholep{\suba\mvar\val\var \tm}\name \cons \pooltwo    
&\tomachlolli&
\namctxtwop\namctxthree\mute &\nctxholep{\cuta\val\vartwo\lctxp{\cuta\valtwo\var\tm}}\name \cons \pooltwo
\end{array}
\]
With $\namctx = \namctxtwop\namctxthree\mute$ and $\pool = \nctxholep{\cuta\val\vartwo\lctxp{\cuta\valtwo\var\tm}}\name \cons \pooltwo$. The transitions of SESAME move around and/or rename cut values of $\namctx$ and values of $\pool$, but they never increase their size. Thus the invariant for $\twostate\namctx\pool$ follows from the \ih on $\statetwo$.
\item It follows from the first point, since all values that are duplicated along $\run$ are cut values from the multi-context.\qedhere
\end{enumerate}
\end{proof}

\gettoappendix{l:search-is-bilinear}
\begin{proof}
\applabel{l:search-is-bilinear}Note that search transitions strictly decrease the size of the terms in the pool, which is increased only by principal transitions. Note also that principal transitions can increase it only of the size of a cut value, which by the sub-term invariant is bounded by $\size\tm$. Then $\sizesea\run\leq \size\tm\cdot(\sizepr\run+1)$, where the $+1$ accounts for the size of the initial term, for which there can be search transition even without any preceding principal transition.\qedhere
\end{proof}

\gettoappendix{l:cost-single-trans}
\begin{proof}
\applabel{l:cost-single-trans}
Search transitions can be easily implemented in $\bigo(1)$ by only re-arranging pointers to sub-terms. For principal transitions, the sub-term invariant (\reflemma{subterm-invariant}) guarantees that the size of the cut values is bound by $\size\tm$, so that their manipulations for implementing the transitions can be implemented in $\bigo(\size\tm)$. The only tricky operation is the $\alpha$-renaming in transitions $\tobang$ and $\toaxeone$, which amounts to copy sub-terms and can be done adapting the linear-time copy algorithm in Accattoli and Barras \cite{DBLP:conf/ppdp/AccattoliB17}.\qedhere
\end{proof}

\gettoappendix{thm:sesame-overhead-bound}
\begin{proof}
  The cost of implementing $\run$ is the sum of the costs of implementing its principal and search transitions on RAM:\applabel{thm:sesame-overhead-bound}
  \begin{itemize}
    \item \emph{Principal transitions}: each one costs $\bigo(\size\tm)$ and so all together they cost $\bigo(\size\tm\cdot\sizepr{\run})$.
    \item \emph{Search transitions}: each one costs $\bigo(1)$ and they are $\size\tm\cdot(\sizepr\run+1)$ of them by \reflemma{search-is-bilinear}. Therefore, all together they cost $\bigo\left(\size\tm\cdot(\sizepr\run+1)\right)$.
  \end{itemize}
   Therefore, the cost of implementing $\run$ is $\bigo(\size\tm\cdot\sizepr{\run})+ \bigo\left(\size\tm\cdot(\sizepr\run+1)\right)=\bigo\left(\size\tm\cdot(\sizepr\run+1)\right)$.\qedhere
\end{proof}
\section{Proofs Omitted from \refsect{gc-and-full} (Final Garbage Collection and Full Cut Elimination)}

\gettoappendix{l:good-w-redex-exists}
\begin{proof}
\applabel{l:good-w-redex-exists}
If $\tm$ is not cut free then it has at leas an out cut. Let $\tm=\ctxp{\cuta\val\var\tmtwo}$ be an innermost out cut of $\tm$, that is, such that in $\tmtwo$ there are no other out cuts of $\tm$. By definition of cut-freeness up to garbage, $\var\notin\ov\tmtwo$. Note that $\tmtwo$, having no out cuts, is cut-free. Then $\var\notin\fv\tmtwo$ and $\cuta\val\var\tmtwo$ is a $\tow$-redex in $\tm$. We need to show that $\ctx$ is good. By contradiction, suppose that $\ctx$ is bad. Then there are two cases:
\begin{itemize}
\item \emph{The hole of $\ctx$ is in the left sub-term of a cut}, that is, $\ctx=\ctxtwo{\cuta\vctx\vartwo\tmthree}$. Then $\cuta\val\var$ is not an out cut of $\tm$. Absurd.
\item \emph{The hole of $\ctx$ is in a dominated sub-term}, that is, $\ctx=\ctxtwop{\cuta\valtwo\vartwo\ctxthree}$ with $\vartwo\in\dfv\ctxthree$. Note that the fact that $\cuta\val\var$ is an out cut forces $\cuta\valtwo\vartwo$ to be an out cut as well. Moreover, if $\vartwo\in\dfv\ctxthree$ then $\ctxthree = \ctx_1\ctxholep{\dera\vartwo\varthree\ctx_2}$ or $\ctxthree = \ctx_1\ctxholep{\suba\vartwo\valthree\varthree\ctx_2}$ for some $\ctx_1$ and $\ctx_2$. Note, in particular, that $\dera\vartwo\varthree$ and $\suba\vartwo\valthree\varthree$ are out of all cuts, being between two out cuts. Then there is an occurrence of $\vartwo$ out of all cuts for $\cuta\valtwo\vartwo$. This contradicts cut-freeness up to garbage of $\tm$, for which the variables of out cuts have no out occurrences. Absurd.\qedhere
\end{itemize}
\end{proof}

\gettoappendix{prop:gc-properties}
\begin{proof}
\applabel{prop:gc-properties}\hfill
\begin{enumerate}
\item By induction on $k$. If $k=0$ then the statement holds with $\tmtwo\defeq \tm$. Let $k>0$. By \reflemma{good-w-redex-exists}, there is a good $\wsym$-redex in $\tm$. Then $\tm\togp\wsym\tmthree$. By \ih $\tmthree \togp\wsym^{k-1} \tmtwo$ with $\tmtwo$ cut-free. Then $\tm \togp\wsym^k \tmtwo$.

\item By \refthm{sesame-overhead-bound}, the run $\run$ has  cost $\bigo\left(\size\tm\cdot(\sizepr\run+1)\right)$. Then the size of the final state $\twostate\tmtwo\emptylist$, and so of $\tmtwo$, is bound by $\bigo\left(\size\tm\cdot(\sizepr\run+1)\right)$. Since $\gcdec\tmtwo$ amounts to simply remove all out cut in a traversal of the term, this can be done in $\bigo(\size\tmtwo)$, given our hypothesis on the representation of terms (discussed just before \reflemma{cost-single-trans}, page \pageref{l:cost-single-trans}). That is, the cost is $\bigo(\size\tmtwo)=\bigo\left(\size\tm\cdot(\sizepr\run+1)\right)$.\qedhere
\end{enumerate}
\end{proof}



\section{Overview of the Data Structures Used in the OCaml Implementation}
\label{app:ocaml}

\subparagraph{Data Structures.}
Terms are encoded at runtime as term graphs, which are Direct Acyclic Graphs (DAG)
augmented with a few back edges. Nodes in the graph are given by constructors
of Algebraic Data Types whose arguments are declared as mutable to imperatively
change the graph during reduction. All type declarations are given in the
\verb+termGraphs.ml+ file.

We start our description of the various shapes of nodes from the \verb+value+
datatype that is used for values:

\begin{lstlisting}
type value =
  | Var of var
  | Abs of { var : var; mutable bo : term }
  | Bang of { mutable bo : term }
\end{lstlisting}

A value is either a variable occurrence node, that points to the unique node in memory
that represents a variable, an abstraction node, that points to the node for
the bound variable and to the root node of the body of the abstraction,
or a promotion node that points to the root node of the promoted term.

\begin{lstlisting}
and var = {
  mutable is : bvar option;
  kind : me;
  name : int;
  mutable copy : var;
}
\end{lstlisting}

Variables nodes are the only kind of nodes that are shared in a term graph.
The \verb+kind+ field holds either the value \verb+Mul+ for multiplicative
variables or the value \verb+Exp+ for exponential variables. The \verb+name+ field,
that plays no role during reduction, is an unique index to differentiate variables
during pretty-printing: a variable of name $i$ is printed either as $\mvar_i$ or $\evar_i$
according to its kind. The parser only accepts terms that follow the Barendregt's
convention and that are well bounded in order to have a bijection between
variables and \verb+kind-name+ pairs.

The \verb+copy+ field is necessary to implement copy ($\alpha$-renaming) in linear
time. During copy, the first time a bound variable is visited a new variable node is 
allocated and \verb+copy+ is set to it. Later, when occurrences of the variable
to be copied are found, the \verb+copy+ field is used to retrieve the new variable
to use, preserving the sharing. Variables that have not been copied have the
\verb+copy+ field initialized to point to themselves recursively, to avoid introducing
an \verb+option+ type.

A variable can be globally free or bound by an abstraction, in which case the
\verb+is+ field is \verb+None+, or it can be bound in a cut, a dereliction or a
subtraction. The three latter cases are represented with a \verb+Some+ value for
\verb+is+ holding the data for the binder. Storing the cut value associated to the
variable in the \verb+is+ field allows to look in $O(1)$ for cuts in a context,
a prerequisite of most SESAME rules.

In place of storing the cut value inside the variable, an alternative implementation
could have made the variable point back to the binder. Handling back-pointers in OCaml
always require some gymnicks, like creating nodes with dummy values and imperatively
updating them later to create the cycles in the graph. Our design choice, instead,
allows to keep the DAG structure, that is handled very easily.

The SESAME machine has no need to access derelictions and subtractions given their
bound variable. However, we have decided to handle cuts, derelictions and subtractions
in the same way, since they can occur in the term structure in the same position.
This leads to some simplification in the code. It is, however, a quite inessential
design decision.

\begin{lstlisting}
and bvar =
  (* ptr must point to the surrounding Bind;
     it can be None in initial terms *)
  | Cut of { v : value; mutable ptr : ptr option }
  | MElim of { mutable var : var; mutable v : value }
  | EElim of { mutable var : var }
\end{lstlisting}

The \verb+bvar+ describes the binder associated to a variable, excluding the case
of abstraction. \verb+MElim+ and \verb+EElim+ are for subtractions and derelictions
and they point to the inspected variables and, in the case of subtractions, to the
node that is the root of the value.

Cut nodes point to the root of the cut value
\verb+v+, but they also have an optional back-pointer \verb+ptr+ that points to the
node that is the parent of the cut. For example, the node for the term
$\dera\evar\evartwo \cuta\val\var \tm$ points both to the node for the variable $\evartwo$ and to the node for the suberm $\cuta\val\var \tm$. The \verb+ptr+ pointer of the \verb+Cut+ node associated
to $\var$ points back to node for $\dera\evar\evartwo \cuta\val\var \tm$. This is necessary to remove in $O(1)$
the cut from the context in the multiplicative principal SESAME rules: the
the right child of the pointed node is mutated from a pointer to the root of
$\cuta\val\var \tm$ to a pointer of the root of $\tm$.

To simplify the implementation of the parser, cuts in the initial terms may have
the \verb+ptr+ field set to \verb+None+. The SESAME will set it correctly when the
cut is processed the first time, entering the approximant.

\begin{lstlisting}
and term =
  | Val of value
  | Bind of { var : var; mutable t : term }
            (* var must have is != None *)
\end{lstlisting}

A term node is either a value node or a \verb+Bind+ node that points to the
bound variable \verb+var+ and the continuation term \verb+t+. As already explained,
the \verb+is+ field of the variable differentiates between cuts, derelictions and
subtraction, avoiding back-pointers.

When a value is used as a term, an additional indirection \verb+Val+ is introduced.
In a more optimized implementation this indirection can be completely avoided by
merging together the \verb+value+ and \verb+term+ datatypes. The price to pay is
the introduction of more runtime invariants, i.e. places in the code where an OCaml
term of type \verb+term+ is actually expected to contain a value only. To mitigate
the extra checks and dead code branches, one could employ Algebraic Data Types to
parameterize the \verb+term+ datatype with additional information about the nature
of the term (value, variable, cut, etc.). Due to the lack of subtyping in OCaml, this
leads in turn to some gymnicks to typecheck the code (e.g. introduce formers for
existential types) and thus we avoided this solution to maintain the code simple
for non experts of the language.

\begin{lstlisting}
and topterm = term ref
\end{lstlisting}

A \verb+topterm+ is just a reference to the root of a term graph and it represents
the whole term.

In OCaml it is not possible to take the address of a record/constructor field to
mutate it. One needs to have a reference to the whole record/constructor instead.
The \verb+ptr+ datatype, that we have already encountered when describing the
back-pointers inside cuts, represents a pointer inside the term graph.
It lists all the constructors (i.e. labelled records) that have fields that are
sub-terms different from variables, i.e. all the positions where an hole
can occur in an approximant.

\begin{lstlisting}
and ptr =
  | Initial of topterm (** points to the [topterm] *)
  | InsideBang of value (** [value] must be a [Bang] *)
  | InsideAbs of value (** [value] must be an [Abs] *)
  | InsideBind of term (** [value] must be a [Bind] *)
  | InsideMElim of bvar (** [bvar] must be an [MElim] *)
\end{lstlisting}

For example, a pointer to the term $\tm$ inside $\dera\evar\evartwo \tm$ is encoded as
\verb+InsideBind tm+ where \verb+tm+ is the node for $\dera\evar\evartwo \tm$.
If we picked a different programming language like C, not subject to mark\&swepp
garbage collection, it could have been possible to just take directly pointers to
pointers to terms in memory in place of using the \verb+ptr+ type.
In the example above, we could have just used a pointer to the pointer to $\tm$ inside
the record for $\dera\evar\evartwo \tm$.

\begin{lstlisting}
type pool = ptr list
\end{lstlisting}

Named holes and jobs are simply represented as pointers inside a term graph, and
the pool of jobs as a list of \verb+ptr+s.
A pointer in the pool is the position of the corresponding named hole in the
approximant, the pointed sub-term is the term to be reduced by the job and the
approximant is just the part of term outside any sub-terms pointed by the pointers
in the pool.

Thanks to this representation, reduction is implemented as local term rewriting:
at each step the unique term graph in memory is locally modified starting from the
topmost pointer in the pool, until the pool becomes empty and the term graph is the
strong normal form.

\section{The Example of SESAME Run, Executed by the OCaml Implementation}
\label{app:example}

The OCaml implementation starts by running some tests. One of them is the example of SESAME run given at the end of \refsect{sesame}, whose output is given below. The names of the jobs in the pool are represented with numbers, that is, the hole name $\name$ is here represented with the index~$1$. Moreover, the pool and the multi-context are represented as a single term where the job $\nctxholep{\tm}1$ is directly plugged in the multi-context at hole $\nctxholep{\cdot}1$, according to the fact that in the implementation we only keep the term graph together with a pool made of pointers into it.
\begin{lstlisting}[mathescape]
Input: [!\m1m1-e1][e1?m2][e1?m3][m2>m3,m4]m4
Parsed: [!$\lambda$m$_1$m$_1$$\to$e$_1$][e$_1$?m$_2$][e$_1$?m$_3$][m$_2$$\triangleright$m$_3$,m$_4$]m$_4$
The term is proper.
       <[!$\lambda$m$_1$m$_1$$\to$e$_1$][e$_1$?m$_2$][e$_1$?m$_3$][m$_2$$\triangleright$m$_3$,m$_4$]m$_4$>$_1$
->sea$_1$ [!$\lambda$m$_1$m$_1$$\to$e$_1$]<[e$_1$?m$_2$][e$_1$?m$_3$][m$_2$$\triangleright$m$_3$,m$_4$]m$_4$>$_1$
->!    [!$\lambda$m$_1$m$_1$$\to$e$_1$]
        <[$\lambda$m$_5$m$_5$$\to$m$_2$][e$_1$?m$_3$][m$_2$$\triangleright$m$_3$,m$_4$]m$_4$>$_1$
->sea$_1$ [!$\lambda$m$_1$m$_1$$\to$e$_1$]
        [$\lambda$m$_5$m$_5$$\to$m$_2$]<[e$_1$?m$_3$][m$_2$$\triangleright$m$_3$,m$_4$]m$_4$>$_1$
->!    [!$\lambda$m$_1$m$_1$$\to$e$_1$]
        [$\lambda$m$_5$m$_5$$\to$m$_2$]<[$\lambda$m$_6$m$_6$$\to$m$_3$][m$_2$$\triangleright$m$_3$,m$_4$]m$_4$>$_1$
->sea$_1$ [!$\lambda$m$_1$m$_1$$\to$e$_1$]
        [$\lambda$m$_5$m$_5$$\to$m$_2$][$\lambda$m$_6$m$_6$$\to$m$_3$]<[m$_2$$\triangleright$m$_3$,m$_4$]m$_4$>$_1$
->-o   [!$\lambda$m$_1$m$_1$$\to$e$_1$]
        [$\lambda$m$_6$m$_6$$\to$m$_3$]<[m$_3$$\to$m$_5$][m$_5$$\to$m$_4$]m$_4$>$_1$
->sea$_1$ [!$\lambda$m$_1$m$_1$$\to$e$_1$]
        [$\lambda$m$_6$m$_6$$\to$m$_3$][m$_3$$\to$m$_5$]<[m$_5$$\to$m$_4$]m$_4$>$_1$
->sea$_1$ [!$\lambda$m$_1$m$_1$$\to$e$_1$]
        [$\lambda$m$_6$m$_6$$\to$m$_3$][m$_3$$\to$m$_5$][m$_5$$\to$m$_4$]<m$_4$>$_1$
->axm$_1$ [!$\lambda$m$_1$m$_1$$\to$e$_1$][$\lambda$m$_6$m$_6$$\to$m$_3$][m$_3$$\to$m$_5$]<m$_5$>$_1$
->axm$_1$ [!$\lambda$m$_1$m$_1$$\to$e$_1$][$\lambda$m$_6$m$_6$$\to$m$_3$]<m$_3$>$_1$
->axm$_1$ [!$\lambda$m$_1$m$_1$$\to$e$_1$]<$\lambda$m$_6$m$_6$>$_1$
->sea$_4$ [!$\lambda$m$_1$m$_1$$\to$e$_1$]$\lambda$m$_6$<m$_6$>$_1$
->sea$_6$ [!$\lambda$m$_1$m$_1$$\to$e$_1$]$\lambda$m$_6$m$_6$
->*GC  $\lambda$m$_6$m$_6$
\end{lstlisting}

}
\end{document}